\PassOptionsToPackage{svgnames}{xcolor}

\documentclass[a4paper,
UKenglish,
cleveref,
autoref,
thm-restate,
pdfa,
]{lipics-v2021}

\title{Refinements for Multiparty Message-Passing Protocols}
\subtitle{Specification-agnostic theory and implementation}

\titlerunning{Refinements for Multiparty Message-Passing Protocols}

\author{Martin Vassor}{University of Oxford, UK}
	{martin.vassor@cs.ox.ac.uk}{https://orcid.org/0000-0002-2057-0495}{}

\author{Nobuko Yoshida}{University of Oxford, UK}
	{nobuko.yoshida@cs.ox.ac.uk}{https://orcid.org/0000-0002-3925-8557}{}

\authorrunning{M. Vassor and N. Yoshida}
\Copyright{Martin Vassor and Nobuko Yoshida}



\supplementdetails[subcategory={Artifact}]{Software}{https://doi.org/10.4230/DARTS.3.2.13} 

\funding{Work supported by: EPSRC EP/T00006544/2, EP/K011715/1, EP/K034413/1,
EP/L00058X/1, EP/N027833/2, EP/N028201/1, EP/T014709/2, EP/V000462,
EP/X015955/1n NCSS/EPSRC VeTSS, and Horizon EU TaRDIS 101093006.}

\acknowledgements{We thank B. Ekici, M. Giunti, P. Hou, A. Suresh,
and F. Zhou.}

\nolinenumbers 

\EventEditors{John Q. Open and Joan R. Access}
\EventNoEds{2}
\EventLongTitle{42nd Conference on Very Important Topics (CVIT 2016)}
\EventShortTitle{CVIT 2016}
\EventAcronym{CVIT}
\EventYear{2016}
\EventDate{December 24--27, 2016}
\EventLocation{Little Whinging, United Kingdom}
\EventLogo{}
\SeriesVolume{42}
\ArticleNo{23}
\EventEditors{Jonathan Aldrich and Guido Salvaneschi}
\EventNoEds{2}
\EventLongTitle{38th European Conference on Object-Oriented Programming (ECOOP 2024)}
\EventShortTitle{ECOOP 2024}
\EventAcronym{ECOOP}
\EventYear{2024}
\EventDate{September 16--20, 2024}
\EventLocation{Vienna, Austria}
\EventLogo{}
\SeriesVolume{313}
\ArticleNo{35}

\usepackage[disable]{todonotes}

\usepackage{xifthen}
\usepackage{xspace}
\usepackage{nicefrac}
\usepackage[nomargin,inline,index,%
  status=draft 
]{fixme} 
\fxusetheme{colorsig}
\FXRegisterAuthor{fxMV}{anfxMV}{MV}
\FXRegisterAuthor{fxNY}{anfxNY}{NY}
\FXRegisterAuthor{fxFZ}{anfxFZ}{FZ}

\usepackage{mathpartir}
\usepackage[capitalise]{cleveref}
\usepackage{xparse}

\usepackage{formal-grammar}

\usepackage{magicvariables}

\usepackage{tikz}
\usetikzlibrary{
	arrows,
	arrows.meta,
	backgrounds,
	calc,
	fit,
	graphs,
	graphs.standard,
	intersections,
	matrix,
	positioning,
	shapes,
	shapes.geometric,
	tikzmark,
}

\usepackage{atendofenv}

\usepackage{amsthm}
\usepackage{prolog}
\usepackage{algorithm2e}
\usepackage{subcaption}
\usepackage{rust}
\usepackage{scribble}

\usepackage{pifont} 
\usepackage{siunitx} 

\usepackage{thmtools}
\usepackage{thm-restate}

\usepackage{wrapfig}

\usepackage[vertical]{messagepassing}

\newtoggle{full}
\settoggle{full}{true}
\iftoggle{full}
{
  \newcommand{\inApp}[1]{\cref{#1}}
  \newcommand{\inAppOutside}[1]{\cref{#1}}
}{
	\newcommand{\inApp}[1]{\cite{Vassor_refinements_2024_full}}
  \newcommand{\inAppOutside}[1]{\cite{Vassor_refinements_2024_full}}
}

\DeclareSymbolFont{euleroperators}{U}{eur}{m}{n}
\SetSymbolFont{euleroperators}{bold}{U}{eur}{b}{n}

\makeatletter
\renewcommand{\operator@font}{\mathgroup\symeuleroperators}
\makeatother

\newcommand{\proj}[2]{\gSessionType{#2}{\color{black}\upharpoonright}_{\participant{#1}}}
\newcommand{\merge}[2]{\lSessionType{\sqcap}_{\color{black}#2}{\lSessionType{#1}}}
\newcommand{\setof}[1]{\{{#1}\}}
\newcommand{\setcomp}[2]{\{#1\mathrel{|}#2\}}

\newcommand{\fv}[1]{\refinement{\operatorname{fv}(#1)}}
\renewcommand{\emptyset}{\varnothing}

\newcommand{\rustinline}[2][]{\lstinline[language=Rust, style=colouredRust, #1]{#2}}

\newcommand{\bind}[2]{\nicefrac{#2}{#1}}%
\newcommand{\substenum}[1]{\mathord{\left\{{#1}\right\}}}%
\newcommand{\subst}[2]{\substenum{\bind{#1}{#2}}}%

\newcommand{\ifempty}[3]{\ifthenelse{\isempty{#1}}{#2}{#3}}

\NewDocumentCommand{\red}{m m m m m m o}{
	#1 \xrightarrow{\ltsLabel{#5}{#6}{#4}{#3}[#7]} #2
}

\makeatletter
\newcommand{\xRightarrow}[2][]{\ext@arrow 0359\Rightarrowfill@{#1}{#2}}
\makeatother
\newcommand{\csred}[1][]{\xRightarrow{{#1}}} 
\newcommand{\redact}[3]{#1\xrightarrow{#3}{#2}}

\newcommand{\grrec}{\textsc{GRRec}}
\newcommand{\grsnd}{\textsc{GRSnd}}
\newcommand{\drec}{\textsc{DRec}}
\newcommand{\dsnd}{\textsc{DSnd}}

\ifdefined\NOCOLOR%
  \newcommand{\withcolor}[2]{#2} 
\else
  \newcommand{\withcolor}[2]{{\color{#1} #2}}
\fi

\definecolor{amethyst}{rgb}{0.6, 0.4, 0.8}

\newcommand{\colorpart}{Teal} 
\newcommand{\colorexp}{Blue} 
\newcommand{\colorgp}{VioletRed} 

\newcommand{\dpt}[1]{\withcolor{\colorpart}{\sf #1}} 
\newcommand{\dexp}[1]{\withcolor{\colorexp}{#1}}

\makeVars{\lblL}{dlbl}{l}

\NewDocumentCommand{\ltsLabelGeneric}{m m m o}{
	\ltsLabel{#1}{#2}{#3}{\dagger}[#4]
}
\NewDocumentCommand{\ltsLabelSend}{m m m o}{\ltsLabel{#1}{#2}{#3}{!}[#4]}
\NewDocumentCommand{\ltsLabelRecv}{m m m o}{\ltsLabel{#1}{#2}{#3}{?}[#4]}
\NewDocumentCommand{\ltsLabel}{m m m m o}{
	\dpt{#1}
	\withcolor{black}{\mathrel{#4}}
	\dpt{#2}
	\withcolor{black}{\msgM[\langle #3 \rangle]}
	\IfValueT{#5}{
		\withcolor{black}{: \refinement{#5}}
	}
}
\newcommand{\actionGeneric}[4]{\withcolor{black}{\ltsLabelGeneric{#1}{#2}{#3}:\refinement{#4}}}
\newcommand{\actionSend}[4]{\withcolor{black}{\ltsLabelSend{#1}{#2}{#3}:\refinement{#4}}}
\newcommand{\actionReceive}[4]{\withcolor{black}{\ltsLabelRecv{#1}{#2}{#3}:\refinement{#4}}}

\newcommand{\act}[1][\alpha]{{\withcolor{Red!50!Black}{#1}}}

\makeVari{actAi}{noFmt}{\alpha}

\newcommand{\traceRaw}[1]{{\withcolor{yellow!50!black}{#1}}}
\newcommand{\trace}[1][\tau]{{\traceRaw{#1}}}
\newcommand{\tempty}{\trace[\epsilon]}
\newcommand{\tconcat}[2]{\trace[{#1} \cdot {#2}]}
\newcommand{\tconcats}[2]{\trace[{#1} \cdot\withcolor{black}{\ldots}\cdot {#2}]}

\newcommand{\traceP}{\trace}
\newcommand{\tracePi}{\trace^\prime}

\newcommand{\fifo}[1][f]{#1}
\newcommand{\fconcat}[2]{\fifo[{#1} \withcolor{Red}{::} {#2}]}
\newcommand{\fconcats}[2]{\fconcat{#1}{\fconcat{\dots}{#2}}}
\newcommand{\push}[2]{\operatorname{enq}(\fifo[#1], \fifo[#2])}
\newcommand{\pop}[1]{\operatorname{deq}(\fifo[#1])}
\newcommand{\nextElem}[1]{\operatorname{next}(\fifo[#1])}
\newcommand{\emptyfifo}{\fifo[\varepsilon]}

\newcommand{\queue}[1]{\withcolor{Purple}{#1}} 
\newcommand{\lookupQueue}[3]{\fifo[\queue{#1}_{(\participant{#2}, \participant{#3})}]}
\newcommand{\pushQueue}[4]{\operatorname{enq}_{(\participant{#3},
\participant{#4})}(\queue{#1}, \fifo[#2])}
\newcommand{\popQueue}[3]{\operatorname{deq}_{(\participant{#2},
\participant{#3})}(\queue{#1})}
\newcommand{\nextElemQueue}[3]{\operatorname{next}_{(\participant{#2}, \participant{#3})}(\queue{#1})}

\makeVar{queueQ}{queue}{w}
\makeVari{queueQi}{queue}{w}

\newcommand{\emptyqueue}{\queue{w_\emptyset}}

\newcommand{\map}[1]{\withcolor{Green}{#1}} 
\newcommand{\lookupMap}[2]{\map{#1}\withcolor{Black}{(}\processVariable{#2}\withcolor{Black}{)}}
\newcommand{\updateMap}[3]{\map{#1}\withcolor{Black}{[}\processVariable{#2}\withcolor{Black}{\mapsto} \val{#3}\withcolor{Black}{]}} 
\newcommand{\domMap}[1]{\operatorname{dom}(\map{#1})}
\newcommand{\emptyMap}{\map{M_\emptyset}}
\newcommand{\removeMap}[2]{\map{#1}\withcolor{Black}{\backslash}\processVariable{#2}}

\makeVar{varX}{dexp}{x}
\makeVar{numN}{dexp}{\underline{n}}

\newcommand{\varsV}[1][]{\map{V_{#1}}}

\newcommand{\setMsgs}[1][]{\withcolor{purple}{\mathbb{M}_{#1}}}
\newcommand{\msgM}[1][m]{\withcolor{purple}{#1}}

\makeVar{setLabels}{dlbl}{\mathbb{L}}
\makeVar{setActs}{noFmt}{\mathbb{A}}
\makeVar{setVars}{dexp}{\mathbb{V}}
\makeVar{setInts}{dexp}{\mathbb{Z}}
\newcommand{\setVals}{\val{\set{C}}}
\newcommand{\setActions}[1][]{{\act[\setActs_{#1}]}}
\newcommand{\setRoles}{\dpt{\mathbb{P}}}

\makeVar{roleP}{dpt}{p}
\makeVar{roleQ}{dpt}{q}
\makeVar{roleR}{dpt}{r}
\makeVar{roleS}{dpt}{s}
\makeVar{roleT}{dpt}{t}
\makeVari{rolePi}{dpt}{p}
\makeVari{roleQi}{dpt}{q}
\makeVari{roleRi}{dpt}{r}
\makeVari{roleSi}{dpt}{s}
\makeVari{roleTi}{dpt}{t}

\newcommand{\rolesOfGT}[1]{\operatorname{parts}(\gtFmt{#1})}

\newcommand{\refinement}[1]{\withcolor{RoyalBlue}{#1}}
\newcommand{\refeval}[1]{\refinement{\operatorname{eval}}(\refinement{#1})} 
\makeVar{predP}{refinement}{P}
\makeVar{setPreds}{refinement}{\mathbb{R}}
\newcommand{\taut}{\refinement{\top}}

\NewDocumentCommand{\refinedTypedVariable}{o o m o}{ 
	\IfValueTF{#4}{
		\IfValueTF{#1}{
			\typedVariable[#1][#2]{#3}\models \refinement{{#4}_{#1}}
		}{
			\typedVariable[#1][#2]{#3}\models \refinement{{#4}}
		}
	}{
		\typedVariable[#1][#2]{#3}
	}
}

\NewDocumentCommand{\typedVariable}{o o m}{ 
	\IfValueTF{#2}{
		\processVariable[#1]{#2}: \processType[#1]{#3}
	}{
		\processType[#1]{#3}
	}
}
\NewDocumentCommand{\processVariable}{o m}{ 
	\IfValueTF{#1}{
		\notblank{#1}{
			\withcolor{blue}{{#2}_{#1}}
		}{
			\withcolor{blue}{{#2}}
		}
	}{
		\withcolor{blue}{{#2}}
	}
}
\NewDocumentCommand{\processType}{o m}{ 
	\IfValueTF{#1}{
		\withcolor{DarkOrchid}{{\tt #2}_{#1}}
	}{
		\withcolor{DarkOrchid}{{\tt #2}}
	}
}
\newcommand{\typeVariable}[1]{\withcolor{red}{\mathbf{#1}}} 
\newcommand{\gSessionType}[1]{\withcolor{VioletRed}{#1}}
\newcommand{\participant}[1]{\withcolor{Teal}{{\rm \sf #1}}}
\NewDocumentCommand{\commChoiceNoCont}{o m o m o}{
	\IfValueTF{#1}{
		\gtLbl{{#2}_{#1}} (\refinedTypedVariable[#1][#3]{#4}[#5])
	}{
		\gtLbl{{#2}} (\refinedTypedVariable[#1][#3]{#4}[#5])
	}
}
\NewDocumentCommand{\commChoice}{o m o m o m}{
	\IfValueTF{#1}{
		\commChoiceNoCont[#1]{#2}[#3]{#4}[#5] . {#6}_{#1}
	}{
		\commChoiceNoCont[#1]{#2}[#3]{#4}[#5] . {#6}
	}
}

\newcommand{\gtFmt}[1]{\withcolor{\colorgp}{#1}}
\newcommand{\gtLbl}[1]{\withcolor{orange}{\mathit{#1}}}
\NewDocumentCommand{\gtCommRaw}{o o m m m} {
	\IfNoValueTF{#1}{%
		{{\participant{#3}} \to {\participant{#4}} \left\{ {#5} \right\} }
	}{%
		\IfNoValueTF{#2}{
				{\participant{#3} \to \participant{#4} {\left\{ {#5} \right\}}_{\indexstyle{#1}}}
			}{
				{\participant{#3} \to \participant{#4} {\left\{ {#5} \right\}}_{\indexstyle{#1 \in #2}}}
			}
	}
}

\NewDocumentCommand{\gtComm}{o o m m m o m o m}{
	\IfNoValueTF{#1}{%
		\gtFmt{\gtCommRaw{#3}{#4}{\commChoice{#5}[#6]{#7}[#8]{#9}}}
	}{%
		\IfNoValueTF{#2}{
				\gtFmt{\gtCommRaw[#1]{#3}{#4}{\commChoice{#5}[#6]{#7}[#8]{#9}}}
			}{
				\gtFmt{\gtCommRaw[#1][#2]{#3}{#4}{\commChoice[#1]{#5}[#6]{#7}[#8]{#9}}}
			}
	}
}

\newcommand{\gtRec}[2]{\gtFmt{%
  \mu\typeVariable{#1}.{#2}
}}
\newcommand{\gtRecVar}[1]{\gtFmt{%
    \typeVariable{#1}
  }}
\newcommand{\gtEnd}{\gtFmt{\texttt{end}}}

\makeVars{\gtG}{gtFmt}{G}

\newcommand{\lSessionType}[1]{\withcolor{amethyst}{#1}}
\newcommand{\ltFmt}[1]{\withcolor{\colortp}{#1}} 
\NewDocumentCommand{\ltIntC}{o o m m m m o m}{
	\IfNoValueTF{#1}{
		\ltFmt{%
			{\participant{#3}} {\oplus}\mspace{-1mu} \{ \commChoiceNoCont{#4}[#5]{#6}[#7] . {#8} \}
		}
	}{
		\IfNoValueTF{#2}{
			\ltFmt{%
				{\participant{#3}} {\oplus}\mspace{-1mu} \{ \commChoiceNoCont{#4}[#5]{#6}[#7] . {#8} \}_{\withcolor{black}{#1}}
			}
		}{
			\ltFmt{%
				{\participant{#3}} {\oplus}\mspace{-1mu} \{ \commChoiceNoCont[#1]{#4}[#5]{#6}[#7] . {#8_#1} \}_{\withcolor{black}{{#1}\in {#2}}}
			}
		}
	}
}

\NewDocumentCommand{\ltIntCRaw}{o o m m} {
	\IfNoValueTF{#1}{%
		{{\participant{#3}} {\oplus}\mspace{-1mu} \left\{ {#4} \right\} }
	}{%
		\IfNoValueTF{#2}{
				{{\participant{#3}} {\oplus}\mspace{-1mu} {\left\{ {#4} \right\}}_{\indexstyle{#1}}}
			}{
				{{\participant{#3}} {\oplus}\mspace{-1mu} {\left\{ {#4} \right\}}_{\indexstyle{#1 \in #2}}}
			}
	}
}

\NewDocumentCommand{\ltExtC}{o o m m m m m m}{
	\IfNoValueTF{#1}{
		\ltFmt{%
			{\participant{#3}} {\&}\mspace{-1mu} \left\{ \commChoiceNoCont{#4}[#5]{#6}[#7] . {#8} \right\}
		}
	}{
		\IfNoValueTF{#2}{
			\ltFmt{%
				{\participant{#3}} {\&}\mspace{-1mu} \left\{ \commChoiceNoCont{#4}[#5]{#6}[#7] . {#8} \right\}_{\withcolor{black}{#1}}
			}
		}{
			\ltFmt{%
				{\participant{#3}} {\&}\mspace{-1mu} \left\{ \commChoiceNoCont[#1]{#4}[#5]{#6}[#7] . {#8_#1} \right\}_{\withcolor{black}{{#1}\in {#2}}}
			}
		}
	}
}

\NewDocumentCommand{\ltExtCRaw}{o o m m} {
	\IfNoValueTF{#1}{%
		{{\participant{#3}} {\&}\mspace{-1mu} \left\{ {#4} \right\} }
	}{%
		\IfNoValueTF{#2}{
				{{\participant{#3}} {\&}\mspace{-1mu} {\left\{ {#4} \right\}}_{\indexstyle{#1}}}
			}{
				{{\participant{#3}} {\&}\mspace{-1mu} {\left\{ {#4} \right\}}_{\indexstyle{#1 \in #2}}}
			}
	}
}

\newcommand{\ltRec}[2]{\ltFmt{%
\mu\typeVariable{#1}.{#2}
}}
\newcommand{\ltRecVar}[1]{\ltFmt{%
\typeVariable{#1}
}}
\newcommand{\ltEnd}{\ltFmt{\texttt{end}}}

\makeVars{\ltL}{ltFmt}{L}
\makeVar{ltRecVarT}{ltRecVarFmt}{t}
\makeVari{ltRecVarTi}{ltRecVarFmt}{t}

\newcommand{\frv}[1]{\operatorname{frv}(#1)}

\newcommand{\indexstyle}[1]{\withcolor{Black}{#1}}

\newcommand{\occursLoc}[2]{\ltFmt{#1} \in \ltFmt{#2}}
\newcommand{\occursGlob}[2]{\gtFmt{#1} \in \gtFmt{#2}}

\newcommand{\val}[1]{{\color{Blue}#1}}

\newcommand{\tuple}[1]{\langle #1 \rangle}
\newcommand{\suchthat}{\cdot} 
\newcommand{\set}[1]{\mathbb{#1}}
\newcommand{\define}{\overset{{\rm def}}{=}} 
\newcommand{\powerset}[1]{\mathcal{P}(#1)}

\newcommand{\ltToAut}[1]{\mathcal{A}(\lSessionType{#1})}
\newcommand{\RcsOfType}[1]{\mathcal{S}(\gSessionType{#1})}
\newcommand{\ConfOfType}[1]{\mathcal{C}(\gSessionType{#1})}

\newcommand{\DConfOfType}[1]{\mathcal{D}(\gSessionType{#1})}

\usepackage[inline]{enumitem}
\newenvironment{inlineenum}[1][and]{
	\begin{enumerate*}[label={{\rm (\roman*)}}, before={},
	itemjoin={{; }}, itemjoin*={{; #1 }},]
}{
	\end{enumerate*}
}

{\centerline{\bf --- Begin Discussion (and/or informal writing)\ifempty{#1}{}{: {#1}} ---}
  \hrule\vspace{1mm}}%
{\hrule\vspace{1mm}\centerline{\bf --- End Discussion (and/or informal writing) ---}}%

\newcommand{\GG}[1]{\operatorname{GG}(\gtFmt{#1})}

\newcommand{\myparagraph}[1]{\smallskip\noindent\emph{\bf #1\ }}

\newcommand{\citet}{\cite}
\newcommand{\Citet}{\cite}

\begin{document}

	\maketitle

\begin{abstract}
	Multiparty message-passing protocols are notoriously difficult to design, due
to interaction mismatches that lead to errors such as deadlocks. Existing
protocol specification formats have been developed to prevent such errors
(e.g.\ multiparty session types (MPST)). In order to further constrain
protocols,
specifications can be extended with \emph{refinements}, i.e.\
logical predicates to control the behaviour of the protocol based on
previous values exchanged. Unfortunately, existing refinement theories and
implementations are tightly coupled with specification formats.

This paper proposes a framework for multiparty
message-passing protocols with refinements and its implementation in Rust.
Our work \emph{decouples} correctness of refinements from the underlying
model of computation, which results in a \emph{specification-agnostic}
framework.

Our contributions are threefold. First, we introduce a trace system
which characterises \emph{valid refined traces}, i.e.\ a sequence of sending
and receiving actions correct with respect to refinements. Second, we
give a correct model of computation named \emph{refined communicating system}
(RCS),
which is an extension of communicating automata systems with refinements. We
prove that RCS only produce valid refined traces. We show how to generate RCS
from mainstream protocol specification formats, such as \emph{refined
multiparty session types} (RMPST) or \emph{refined choreography automata}.
Third, we illustrate the flexibility of the framework by developing
both a static analysis technique and an improved model of computation for
dynamic refinement evaluation. Finally, we provide a Rust toolchain for
decentralised RMPST, evaluate our implementation with a set of benchmarks from
the literature, and observe that refinement overhead is negligible.

\end{abstract}

\AtEndOfEnv{theorem}{$\triangleleft$}
\AtEndOfEnv{definition}{$\triangleleft$}
\AtEndOfEnv{lemma}{$\triangleleft$}
\AtEndOfEnv{corollary}{$\triangleleft$}
\AtEndOfEnv{remark}{$\triangleleft$}

\section{Introduction}
Message passing programming is a notoriously difficult task with new bugs
arising with respect to sequential programming, for instance deadlocks.
To address this increased complexity, various specifications have been
introduced (e.g., message sequence charts \cite{Z120_MSC}, multiparty session
types \cite{yoshida_very_2020,honda_multiparty_2016,honda_multiparty_2008},
choreography automata \cite{COORDINATION20ChoreographyAutomata}). In general, specifications are used to constrain messages, in order to
prevent errors such as deadlocks (via message ordering) or payload mismatch (by
enforcing the sender and the receiver of a message to agree on the datatype
exchanged).

In this paper, we tackle an important and advanced aspect of protocol
specification, \emph{logical constraints} (or \emph{contracts})
on asynchronous message-passing
communications.
Contracts for heterogeneous systems are
predominant for correctly
designing, implementing, and composing
software services, and have
a long history in distributed software development
as found in Design-by-Contracts \cite{DbC91},
Service Level Agreements, and
Component-Based Software Engineering.
With contracts,
software designers can define more precise (refined)
and verifiable specifications for distributed software components.
Contracts have been investigated from a variety of
perspectives,
using many different analysis techniques and formalisms.
Our goal is to distill an essence of
those models for protocol refinements by answering
the following questions affirmatively:
\begin{inlineenum}
\item what does it mean for an execution of contracts for
message-passing systems
to be correct
\item how do we integrate a theory to a variety of models
\item how do we analyse their correctness?
\item how do we implement correct systems in a programming language?
\end{inlineenum}

To explain our framework,
consider a \emph{guessing game}
(from~\cite{OOPSLA20SessionStar})
with three participants where the
first one (participant \participant{A}) chooses a $\gtLbl{secret}$ integer and
sends it to the second participant (\participant{B}). Then, the third
participant (\participant{C}) tries to $\gtLbl{guess}$ this number. Depending
on the guess, \participant{B} replies with hints ($\gtLbl{more}$ and
$\gtLbl{less}$) until \participant{C} succeeds in guessing the
$\gtLbl{correct}$ value.

The developer writing the specification for such protocol would
like to ensure, \emph{in the specification}, that hints from \participant{B}
are consistent with the previous values exchanged. For instance, if the
$\gtLbl{secret}$ is $\val{5}$ and the guess is $\val{10}$, the
specification should constrain \participant{B} to send $\gtLbl{less}$.
\cref{fig:intro:guessing_game} shows a communication diagram of the protocol
with constraints (which we call \emph{refinements}) shown in light blue.

\begin{wrapfigure}{o}{0.45\textwidth}
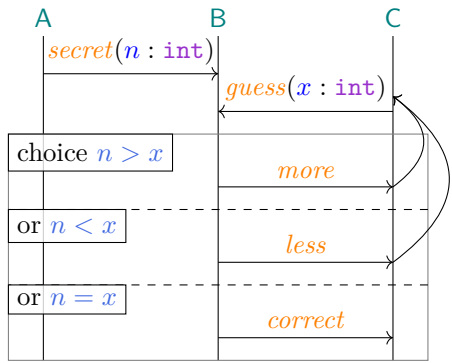

	\centering
	\begin{messagepassing}[yscale=2.3, ]
		\newprocesswithlength{C}[\participant{C}]{4.3}
		\newprocesswithlength{B}[\participant{B}]{4.3}
		\newprocesswithlength{A}[\participant{A}]{4.3}

		\sendwithname{A}{0.5}{B}{0.5}{
			$\gtLbl{secret}(\processVariable{n}: \processType{int})$}[pos=.5]
		\sendwithname{C}{1}{B}{1}{
			$\gtLbl{guess}(\processVariable{x}: \processType{int})$}[pos=.5]
		\sendwithname{B}{2}{C}{2}{$\gtLbl{more}$}[pos=.5]
		\sendwithname{B}{3}{C}{3}{$\gtLbl{less}$}[pos=.5]
		\sendwithname{B}{4}{C}{4}{$\gtLbl{correct}$}[pos=.5]

		\draw[draw=black!50] (1.3, -0.8) rectangle (4.3, -3.2);

		\draw[dashed] (2.3, -0.8) -- (2.3, -3.2);
		\draw[dashed] (3.3, -0.8) -- (3.3, -3.2);
		\node[anchor=north west,
		fill opacity=0.8,
		text opacity=1,
		draw,
		fill=white,
		outer sep=0pt] at (2.3, -3.2) {or $\refinement{n<x}$};
		\node[anchor=north west,
		fill opacity=0.8,
		text opacity=1,
		draw,
		fill=white,
		outer sep=0pt] at (3.3, -3.2) {or $\refinement{n=x}$};
		\node[anchor=north west,
		fill=white,
		fill opacity=0.8,
		draw, text opacity=1,
		outer sep=0pt] at (1.3, -3.2) {choice $\refinement{n>x}$};

		\draw[->] (2, -1) to[bend right] (0.8, -1);
		\draw[->] (3, -1) to[bend right] (0.8, -1);
	\end{messagepassing}
	\caption{Communication diagram for the guessing game protocol with
	refinements.}
	\label{fig:intro:guessing_game}
	\vspace{-\baselineskip}
\end{wrapfigure}

In this paper, we develop a formal framework for
refinements, agnostic to any particular specification formalism. Its
core part is composed of a characterisation of
refinement correctness: \emph{Valid Refined Traces}, and a model
of computation: \emph{Refined Communicating Systems} (RCS), where communication
is asynchronous and refinements are
centrally and dynamically evaluated. For illustration, we use
\emph{Refined Multiparty Session Types} as the main specification format for
multiparty protocols.

In addition, we demonstrate the versatility of our framework with multiple
extensions. First, our framework can accommodate other protocol
specification formats (e.g.\ choreography
automata \cite{COORDINATION20ChoreographyAutomata}). Second, it is used as a
baseline for improved refinement evaluation: we present an optimised
model of computation (decentralised refinement evaluation).
Finally, it is also used as a baseline for implementing static analysis
techniques: we present a simple strategy for statically removing redundant
refinements.

\myparagraph{Valid Refined Traces.}
\label{sec:intro:valid_refined_traces}
The first building block is a common notion of \emph{correct} executions with
respect to added refinements.
We introduce \emph{valid refined traces} which are
consistent traces with respect to refinements. This approach allows us to
establish
a general notion of refinements, which is applicable to different logics
for constraints, type theories, models of computations, and programming
languages. We consider asynchronous communications (FIFOs), distinguishing
\emph{sending} and \emph{receiving} actions in traces.

To illustrate our approach, consider the guessing game example shown above.
Each execution of that protocol is recorded in a trace, i.e.\ a sequence of
the individual events that take place during the execution (c.f.\
\cref{sub:traces}).
For instance, a
possible trace of the first four events of the protocol is the following:\\[1mm]
\centerline{
	$\tconcat{
		\actionSend{A}{B}{\gtLbl{secret},
		\tuple{\processVariable{n},\val{5}}}{\taut}
	}{
	\tconcat{
		\actionReceive{A}{B}{\gtLbl{secret},
		\tuple{\processVariable{n},\val{5}}}{\taut}
	}{
	\tconcat{
		\actionSend{C}{B}{\gtLbl{guess},\tuple{\processVariable{x},\val{5}}}{\taut}
	}{
		\actionReceive{C}{B}{\gtLbl{guess},\tuple{\processVariable{x},\val{5}}}{\taut}
	}}}
$}\\[1mm]
This trace contains four actions, and each action records an event, i.e.\ a
message emission (denoted with a \(!\)) or reception (denoted with a \(?\)). For
instance, \(\actionSend{A}{B}{\gtLbl{secret},
\tuple{\processVariable{n},\val{5}}}{\top}\) records
\(\participant{A}\) sending a message
to
$\participant{B}$, the payload of this message is a variable
$\processVariable{n}$, which has value $\val{5}$. In the first four actions, we
do not need any constraint,
therefore actions are guarded by $\taut$ which denotes a tautology predicate.
The next action following this trace would be for \(\participant{B}\) to send
either \(\gtLbl{more}\), \(\gtLbl{correct}\), or \(\gtLbl{less}\).
Choosing \(\gtLbl{more}\) or \(\gtLbl{less}\) would be inconsistent with our
protocol, since \(\participant{C}\) guessed the correct number.
For instance, choosing $\gtLbl{more}$ would add the action
	$\actionSend{B}{C}{\gtLbl{more}}{n > x}$ at the end of the queue: the
	refinement $\refinement{n > x}$ would be violated, since
	$\processVariable{x} = \processVariable{n} = \val{5}$.

\emph{Valid Refined Traces} characterise consistency based on
the produced trace; and we aim to provide a model of computation
constrained in a way that prevents such inconsistent choices.

\myparagraph{Refined Communicating Systems.}
\label{sec:intro:RCS}
The second building block of our framework is a model of computation that
only produces correct traces.
\emph{Communicating Systems} (CS)
\cite{brand_communicating_1983} are a model of concurrent computation, where
\emph{Communicating Finite State Machines} communicate asynchronously using
unbounded FIFO queues.
CS are often used to model and implement
MPST \cite{Denielou_Multiparty_2012,ICALP13CFSM,cutner_deadlock-free_2022}.
We adapt CS to accommodate refinements, which we call \emph{Refined
Communicating
Systems} (RCS). The semantics is modified in order to check
refinements at every step. For this, we introduce a shared map in order
to keep track of variables and their values that are exchanged in
messages (e.g. the values of $\processVariable{x}$ and $\processVariable{n}$ in
the guessing game example).
This record of values is used to evaluate refinements, preventing
undesired transitions. In this paper,
we show that RCS only produce valid refined traces and we explain how to
generate an RCS from a RMPST.

\myparagraph{Refined MPST.}
\label{sec:intro:RMPST}
Working with CS is cumbersome, and, in practice, we would
prefer to adapt existing specification formats. We present in depth how to
integrate refinements in \emph{Multiparty Session Types} (MPST)
\cite{yoshida_very_2020,honda_multiparty_2016,honda_multiparty_2008}, which
are a family of type systems that aims to prevent communication bugs.

The following refined global type (\(\gtFmt{G_\pm}\)) is a specification
of the guessing game protocol (\cref{fig:intro:guessing_game}), with
refinements: a participant
\(\participant{A}\) begins by sending a $\gtLbl{secret}$ to
\(\participant{B}\); the value of the $\gtLbl{secret}$ is stored in the
variable $\processVariable{n}$. Then,
\(\participant{C}\) tries to guess the value (stored in variable
$\processVariable{x}$), and \(\participant{B}\) replies
with \(\gtLbl{more}\), \(\gtLbl{less}\) (in which case the protocol loops and
$\participant{C}$ can make another guess: $\gtRec{T}{G}$
denotes the recursion) or \(\gtLbl{correct}\), at which point
the protocol terminates ($\gtEnd$ denotes the termination).
The refinements specify conditions upon
which the \(\gtLbl{more}\), \(\gtLbl{less}\), and \(\gtLbl{correct}\) branches
are possible. For instance, the protocol can take the \(\gtLbl{correct}\)
branch only if the values in the \(\gtLbl{secret}\)
and the \(\gtLbl{guess}\) messages are the
same, i.e. if $\refinement{x = n}$.\\[1mm]
\centerline{$
\small
\begin{array}{l}
\gtFmt{G_\pm} =\\
\quad	\gtComm{\participant{A}}{\participant{B}}{secret}[n]{int}[\taut]{
		\gtRec{T}{
			\gtComm{C}{B}{guess}[x]{int}[\taut]{
				\gtCommRaw{B}{C}{
					\begin{aligned}
						&\commChoice{more}{}[x < n]{\gtRecVar{T}},\\
						&\commChoice{less}{}[x > n]{\gtRecVar{T}},\\
						&\commChoice{correct}{}[x = n]{\gtEnd}
					\end{aligned}
				}
			}
		}
	}
\end{array}$}\\[1mm]
Compared to standard MPST, \emph{Refined MPST} (RMPST) contain variable names
($\processVariable{n}$ and $\processVariable{x}$) and refinements (denoted with
$\gtFmt{\models\refinement{r}}$ in the payloads, meaning that to send the
message, $\refinement{r}$ must hold). We present those extensions
as well as the relation between RMPST and RCS.

\myparagraph{Applications and Extensions.} To show the versatility of our
framework, we extend it:
\vspace{-2mm}
\begin{description}
	\item[Decentralised Refinement Evaluation:]
	The canonical semantics for RCS we present uses a single shared map of
	variables to provide a simple way to reason about refinements. Having
	this global map would not be suited for a distributed implementation. We
	extend our framework with an alternative semantics where each participant
	of a protocol has a local map of variables. We show that if variables are
	not duplicated, then this alternative model also produces valid refined
	traces.
	\item[Static Elision of Redundant Refinements:]
	At places where refinements are redundant (e.g.\ where it is entailed by
	previous refinements), we could benefit from removing those refinements. In
	order to show the versatility of our framework, we show how to develop a
	simple static analysis technique to remove such redundant refinements.
	\item[Refined Choreography Automata:]
	While we mostly use RMPST as an example of protocol specification language,
	we sketch another specification by (informally) presenting how
	to integrate refinements in choreography automata
	(in \inApp{sec:refined_choreo_automata}).
\end{description}

\myparagraph{Rust Implementation.}
The last objective of our work is to implement RMPST into Rust. We choose Rust
for several reasons: its affine type system makes it easy to avoid unwanted
reuse of values, which helps to prevent a participant from duplicating
actions; and thanks to its growing popularity, there are already a few existing
toolchains for session types in Rust
\cite{lagaillardie_stay_2022,Chen2021,Kokke2019,JespersenSession2015}.
Among them, we choose Rumpsteak \cite{cutner_deadlock-free_2022}
since it already uses CS to implement MPST participants
inside its toolchain.
We extend Rumpsteak with refinements using the decentralised refinement
evaluation approach.
We finally measure the refinement overhead in Rumpsteak.

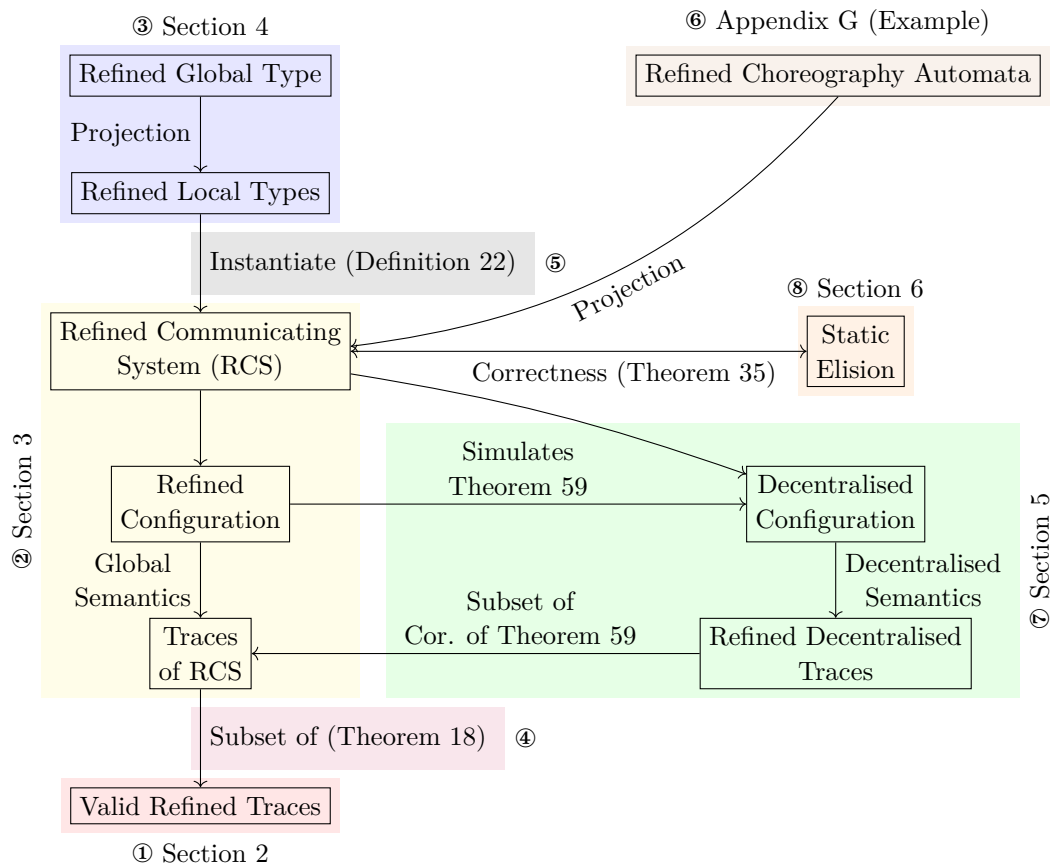
\begin{figure}[h]
	\begin{tikzpicture}
	\node (RGA) [draw=black, align=center] {Refined Communicating\\System (RCS)};
	\node (SC) [draw=black, align=center, below=1cm of RGA]
	{Refined\\Configuration};
	\node (DC) [draw=black, right = 6cm of SC, align=center] {Decentralised\\Configuration};
	\node (RDLAT) [draw=black, below = of DC, align = center] {Refined Decentralised\\Traces};
	\node (RGAT) [draw=black, align=center, below = of SC] {Traces\\of RCS};
	\node (RGT) [draw=black, align=center, below =1.3cm of RGAT] {Valid Refined
	Traces};

	\node (LRT) [draw=black, above=1.3cm of RGA, align=center] {Refined Local
	Types};
	\node (GRT) [draw=black, above = of LRT, align=center] {Refined Global Type};

	\node (RCA) [draw=black, align=center, right=4cm of GRT] {Refined
	Choreography Automata};

	\node (SERCS) [draw=black, align=center, right=6cm of RGA]
	{Static\\Elision};

	\draw[->] (RGA) -> node [pos=.5, anchor=south, sloped, align=center] {} (SC);
	\draw[->] (RGA) to [bend left=5] (DC);
	\draw[align=center,->] (SC) -> node [pos=.5, anchor=south] (Sim)
	{Simulates\\\cref{thm:st_sim_dyn}} (DC);
	\draw[align=center, ->] (SC) -> node [pos=.5, anchor=east] (GS) {Global\\Semantics} (RGAT);
	\draw[align=center,->] (DC) -> node [pos=.5, anchor=west] (DS) {Decentralised\\Semantics} (RDLAT);
	\draw[align=center,->] (RDLAT) -> node [pos=.4, anchor=south, sloped] (Cor Sim) {Subset of\\Cor. of \cref{thm:st_sim_dyn}} (RGAT);
	\draw[align=center,->] (RGAT) -> node [align=center, pos=.5, anchor=west]
	(Th Subset)
	{Subset of (\cref{thm:traces_GRA_valid_refined})} (RGT);

	\draw[align=center, ->] (LRT) edge node[align=center, pos=.5, anchor=west]
	(Inst) {Instantiate (\cref{def:LT_to_RCFSM})} (RGA);
	\draw[align=center, ->] (GRT) edge node[pos=.5, anchor=east] (Proj)
	{Projection} (LRT);

	\draw[align=center, ->] (RCA.south) to [bend left=20] node[anchor=north,
	sloped]	(proj_ca) {Projection} (RGA);

	\draw[align=center, <->] (SERCS) to node[anchor=north,
	sloped, pos=.4]	(se_correct) {Correctness
	(\cref{thm:static_elision:correct})} (RGA);

	\begin{pgfonlayer}{background}
		\node[fit=(GRT) (LRT), fill=blue!10] (BG RMPST) {};
		\node[above=0cm of BG RMPST, anchor=south] {\ding{174}
		\cref{sec:refined_mpst}};
		\node[fit=(RGT), fill=red!10] (BG trace) {};
		\node[below=0cm of BG trace, anchor=north] {\ding{172}
		\cref{sec:traces}};
		\node[fit=(RGA) (SC) (RGAT) (GS), fill=yellow!10] (BG automata){};
		\node[left=0cm of BG automata, anchor=south, rotate=90] {\ding{173}
		\cref{sec:ref_auto}};
		\node[fit=(DC) (RDLAT) (DS) (Sim) (Cor Sim), fill=green!10] (BG Dyn) {};
		\node[right=0cm of BG Dyn, anchor=north, rotate=90] {\ding{178}
		\cref{sec:dynamic}};
		\node[fit=(Th Subset), fill=purple!10] (BG Th Subset) {};
		\node[right=0cm of BG Th Subset, anchor=west] {\ding{175}};
		\node[fit=(Inst), fill=black!10] (BG Inst) {};
		\node[right=0cm of BG Inst, anchor=west] {\ding{176}};
		\node[fit=(RCA), fill=brown!10] (BG_RCA) {};
		\node[above=0cm of BG_RCA, anchor=south]{\ding{177}
		\inAppOutside{sec:refined_choreo_automata} (Example)};
		\node[fit=(SERCS), fill=orange!10] (BG_SE) {};
		\node[above=0cm of BG_SE, anchor=south]{\ding{179}
		\cref{sec:static_elision}};
	\end{pgfonlayer}
\end{tikzpicture}
	\caption{Overview of the framework for RMPST developed in this work. The
	coloured backgrounds show the main steps of this paper.}
	\label{fig:overview_paper}
\end{figure}
\myparagraph{Contributions and Outline.}
Our main contribution is to unify
the different points presented above in a \emph{single} framework
as presented in \cref{fig:overview_paper}.
We introduce a uniform framework which is
agnostic to any particular specification formalism,
model, semantics and language,
defining the correctness of refinements as validity of traces.
We then prove the safety of the framework
(\cref{thm:traces_GRA_valid_refined}).
We demonstrate the \emph{versatility of our framework} by
accommodating
\emph{multiple protocol specifications} such as
(refined) multiparty session types
\cite{yoshida_very_2020,honda_multiparty_2016,honda_multiparty_2008,zhou_statically_2020} and
(refined) choreography
automata \cite{COORDINATION20ChoreographyAutomata,Gheri_Design_2022}, \emph{multiple
semantics} such as (refined) communicating
automata \cite{brand_communicating_1983} with centralised and
decentralised semantics, and
\emph{multiple analysis techniques} such as dynamic and static analyses.
We provide an implementation of an instance of the framework in Rust.
Our framework is the first, to
the best of our knowledge, to achieve such versatility.

The framework is composed of the following parts (circled numbers refer to
\cref{fig:overview_paper}):
\begin{description}
	\item [\ding{172} Valid Refined Traces:] We
		introduce \emph{valid refined traces} which characterise
		valid executions with respect to refinements.
	\item [\ding{173} Refined Communicating Systems (RCS):] We extend
		Communicating Systems to accommodate refinements. From
		a configuration of RCS, we induce a set of possible traces.
		One of our main results is
		\cref{thm:traces_GRA_valid_refined} (\ding{175}), which states
		that all traces produced by RCS are valid refined
		traces, which in turn proves the correctness of the RCS.
	\item [\ding{174} Refined Multiparty Session Types (RMPST):]
		In \Cref{sec:refined_mpst},
		we adapt MPST (which consists of \emph{global types} (which describe a
		multiparty protocol), \emph{local types} (which describe the behaviour
		of a single participant), and a \emph{projection} from global to local
		types which extracts the behaviour of a single participant) to
		accommodate for refinements.
		We show how	to generate a RCS from a set of local types with
		refinements	(\ding{176}).
		In addition, in \inApp{sec:refined_choreo_automata}, we sketch how
		to accommodate refinements in choreography automata, to illustrate the
		versatility of the framework (\ding{177}).
	\item [\ding{178} and \ding{179} Optimisations:]
		In \Cref{sec:dynamic} (\ding{178}), we propose a \emph{decentralised}
		model as an alternative for RCS. We show trace inclusion w.r.t.\ RCS,
		which ensures refinements are correctly checked. In \Cref{sec:implem},
		we implement this improved model in Rust.
		In addition, in \cref{sec:static_elision}, we demonstrate how to
		develop analysis techniques using the framework. We show how redundant
		refinements can, under some conditions, be statically
		removed~(\ding{179}).
\end{description}



\label{sec:intro}
\section{Refined Traces and their Validity}\label{sec:traces}
This section introduces \emph{refined traces} which are sequences of messages
\emph{actions}. We then define their
\emph{validity}, introducing two definitions on traces, \emph{well-queued} and
\emph{well-predicated} traces.
We precede this (in \cref{sec:pred_lang_sem}) with preliminary definitions
used throughout this paper.

\subsection{Preliminaries: Predicates Language and
Semantics\label{sec:pred_lang_sem}}

This first subsection introduces the basic definitions we use in this paper.

Let \(\setVars\) be a set of variables, ranged over by \(\processVariable{x}\),
\(\processVariable{y}\), \(\dots\); and a finite set \(\setVals\) of values (in
this work, we take $32$-bit integers:
\(\dexp{\setInts/2^{32}\setInts}\)).

We use associative maps from variable names to values, noted $\map{M}$.
\(\domMap{M}\) denotes
the domain of a map,
that is the set of variables that appear in the map.
Maps are equipped with
lookup (\(\lookupMap{M}{x}\)), update (\(\updateMap{M}{x}{c}\)) and removal
(\(\removeMap{M}{x}\)) operations. \(\map{M_1}\biguplus\map{M_2}\) denotes the
union of \(\map{M_1}\) and \(\map{M_2}\) if their domains are disjoint (see
\inApp{app:map} for the definition of all those operators). Finally,
$\emptyMap$ denotes an empty map.

In order to keep our work general, we do not strictly specify the language of
predicates, nor their semantics rules.
Instead, we suppose we are
given a language to express refinements, whose terms are produced by a rule
\(\refinement{\nonterm{R}}\).
In this paper, we intentionally leave the logic underspecified so that it
can be fine tuned by the end user. In practice, in our implementation
({\cref{sec:implem}}), custom predicates can easily be added.
In the following, we use a simple grammar
with arithmetic and relational operators as predicates.
Let \(\setPreds\) be the set of refinement expressions. We assume
refinements can have free variables, and that there exist a function
\(\refinement{\operatorname{fv}}:
\setPreds\rightarrow\powerset{\processVariable{\set{V}}}\) that
gives the free variables of each refinement expression.  We note
\(\setPreds_{\processVariable{\set{W}}}\) be the set of
refinements of
\(\setPreds\) whose set of free variables is
\(\processVariable{\set{W}}\subseteq\setVars\).
We assume a variable substitution function,
\(\refinement{\nonterm{R}}\subst{\processVariable[i]{x}}{\val{v_i}}\) that
substitutes every free occurrence of each variable \(\processVariable[i]{x}\)
for
the value \(\val{v_i}\).
For any refinement expression
\(\refinement{r}\),
\(\refinement{r}\subst{\refinement{\operatorname{fv}(r)}}{\val{\dots}}\) is a
closed refinement.
Since our predicates are abstract, we do not explicitly specify their
semantics, nor their well-formedness.
Instead, we assume each closed refinement formula evaluates to
\(\top\) or \(\bot\).
We assume there exists a function \(\refeval{r}\) that
evaluates the refinement \(\refinement{r}\), provided that
	$\refinement{r}$ is closed\footnote{We do not
	discuss the decidability of the actual chosen logic of refinements here.
	For undecidable logics, providing such function is, of course, not possible;
	however this is not in the scope of this work.}. Finally, we assume the
existence of a closed formula \(\taut\) that is a tautology, i.e.
\(\refeval{\taut} = \top\).

Given a map \(\map{M}\) and a refinement \(\refinement{r}\), we note
\(\map{M}\models \refinement{r}\)
if and only if
the refinement \(\refinement{r}\) is closed under the map \(\map{M}\):
\(\refinement{\fv{r}} \subseteq \domMap{M}\),
and
evaluates to \(\top\) after substitution:
\(\refeval{r\subst{\refinement{\fv{r}}}{\lookupMap{M}{\refinement{\fv{r}}}}}=\top\).

In a protocol with multiple participants, let $\setRoles$ be a set of
participants ranged over by \(\participant{A},
\participant{B}, \dots\) and \(\participant{p}, \participant{q}, \dots\) being
meta-variables over participant names. In this work, messages contain a label,
a variable, and a value.
Let $\setLabels$ be a set of labels; \(\gtLbl{\ell}\) and its decorated
variants range over labels in \(\setLabels\).
We define \(\setMsgs = \setLabels \times (\setVars \times \setVals)\) for the
set of messages (as a reminder: $\setLabels$ is the set of labels, $\setVars$
the set of variables, and $\setVals$ the set of values).

\subsection{Traces}
\label{sub:traces}
Let us denote
$\vec{e}=\fconcats{e_1}{e_n}$ ($n\geq 0$) as a \emph{FIFO}, i.e.,
a finite sequence of elements $e_i$ (messages exchanged
in this paper).
We use $\emptyfifo$ for an empty FIFO ($n=0$).
We define: $\push{\vec{e}}{e} \define\fconcat{e}{\vec{e}}$;
$\pop{\fconcat{\vec{e}}{e}}\define \vec{e}$ ($\pop{\emptyfifo}$ is undefined);
and
$\nextElem{\fconcat{\vec{e}}{e}} \define e$ ($\nextElem{\emptyfifo}$ is
undefined). Notice that $\pop{\vec{e}}$ is defined if and only if
$\nextElem{\vec{e}}$ is defined. In this paper, we consider one FIFO channel
per pair of participant. We call \emph{queues} a map of all pairs of
distinct participants to their communication FIFO of a
system. We note $\pushQueue{w}{e}{p}{q}$, $\popQueue{w}{p}{q}$,
$\nextElemQueue{w}{p}{q}$, where the indices indicates which FIFO of the set is
affected (see \inApp{app:map} for the formal definition). We write
\(\emptyqueue\) for the empty queue, which is the queue where
\(\lookupQueue{w}{p}{q} = \emptyfifo\) for all \(\participant{p}\) and
\(\participant{q}\).

Actions are tuples consisting of a sending participant $\roleP$,
a direction of communication $\dagger \in \setof{\mathrel{!}, \mathrel{?}}$
($\mathrel{!}$ stands for sending, and $\mathrel{?}$ stands for
receiving), a receiving participant $\roleQ$, a message $\msgM$
and a predicate
\(\refinement{r}\) associated to the action (as a reminder: $\setPreds$ is the
set of refinements). We require
participants to be distinct
(i.e. $\roleP \neq \roleQ$).

\begin{definition}[Action and Trace\label{def:trace}]
	An action is an element of \(\setActions\) defined as
	follows:
\(\setActions = \setRoles \times \setof{\mathrel{!}, \mathrel{?}} \times
\setRoles \times \setMsgs \times \setPreds\).
	We write
	$\act = \actionGeneric{\roleP}{\roleQ}{\msgM}{r}$ (\(\roleP\neq\roleQ\))
	when
	\(\tuple{\roleP, \dagger, \roleQ, \msgM, \refinement{r}}\in\setActions\).

Traces (denoted by $\trace$ and its decorated variants) are
finite sequences of actions, defined inductively from the rule
\(
	\begin{grammar}
		\firstcase{\trace[T]}{\tconcat{\act}{\nonterm{T}} \gralt \tempty}{}
	\end{grammar}
	\),
	where \(\act\) is an action.
	We write \(\trace[{\setActions^\star}]\) for the set of traces.
\end{definition}

\begin{example}[Trace]
We presented a trace in \cref{sec:intro}.
\end{example}

\noindent We denote
\(\tconcat{\tau_1}{\tau_2}\) for the concatenation of two
traces.
We assume an intuitive notion of the size of trace,
as well as lemmas that allow us to infer that, if the size is \(0\),
then the trace is \(\tempty\).

\subsection{Properties of Refined Traces}

In this subsection, we characterise the \emph{correctness} of traces
w.r.t.\ refinements.

There are two conditions valid traces should verify. First, the
sending/reception of messages should be consistent (as with normal MPST).
Second, for every action of the trace, predicates that guard the action should
hold.
We call traces that satisfy message consistency \emph{well-queued traces}, and
the traces that satisfy the predicates \emph{well-predicated traces}.
In the end, we consider traces that satisfy both conditions: we call
those traces \emph{valid refined traces}.

To start with well-queued traces, we first evaluate the impact of a trace on a
queue, by looking at the effect of each action on that queue
(\cref{def:ends_up_with_queue}).

\begin{definition}[Trace Ending Up with Queues\label{def:ends_up_with_queue},
	well-queued traces\label{def:well_queued_trace}]
	A trace $\traceP$ ends up with the queue $\queueQ[f]$ w.r.t.\ a
	queue \(\queueQ[i]\) if:
  \begin{enumerate}
	  \item If $\traceP = \tempty$, \(\queueQ[i] = \queueQ[f]\); and
		  \label{item:def:ends_up_with_queue:i}
    \item If $\traceP =
	    \tconcat{\actionSend{\roleP}{\roleQ}{\msgM}{r}}{\tracePi}$, then
		  \(\trace[\tau^\prime]\) ends up with \(\queueQ[f]\)
		  w.r.t.\ \(\pushQueue{\queueQ[i]}{\msgM}{p}{q}\); and
		  \label{item:def:ends_up_with_queue:ii}
	  \item \label{item:def:ends_up_with_queue:iii} If $\traceP =
	  \tconcat{\actionReceive{\roleP}{\roleQ}{\msgM}{r}}{\tracePi}$,
	    then
		  \(\trace[\tau^\prime]\) ends up with \(\queueQ[f]\)
		  w.r.t.\ \(\popQueue{\queueQ[i]}{p}{q}\) and
		  \(\nextElemQueue{\queueQ[i]}{p}{q} = \msgM\).

  \end{enumerate}
	A trace $\traceP$ is \emph{well-queued} with regards to the queue
	$\queueQ$ if \(\traceP\) ends up with the empty queue
  \(\queueQ[\emptyset]\) with respect to an initial queue \(\queueQ\).

  A trace \(\traceP\) is valid if \(\traceP\) is well-queued
  with respect
  to the empty queue \(\queueQ[\emptyset]\).
\end{definition}
\begin{remark}
	In \cref{def:ends_up_with_queue}, we say $\queueQ[i]$ is the \emph{initial}
	queue.
\end{remark}

Regarding well-predicated traces, the idea is to record the latest value of
each variable in a map; and to use that map to evaluate refinements
(\cref{def:WPT}).

\begin{definition}[Well-Predicated Traces\label{def:WPT}]
	A trace $\trace$ is \emph{well-predicated} under a map $\map{M}$, if either
	\begin{inlineenum}[or]
	  \item $\traceP = \tempty$\label{def:WPT:i}
	  \item $\trace = \tconcat{\actionGeneric{\roleP}{\roleQ}{\lblL,(\varX,
	  \val{c})}{r}}{\trace[\tau^\prime]}$
		  and \(\updateMap{M}{x}{c}\models \refinement{r}\) and
		  $\trace[\tau^\prime]$ is well-predicated under
	  $\updateMap{M}{x}{c}$\label{def:WPT:ii}.
  \end{inlineenum}
\end{definition}

\begin{example}[Well-Predicated Traces\label{ex:WPT}]
	In \cref{sec:intro}, we presented the trace \(\trace\):\\[1mm]
\centerline{
$	\tconcat{
		\actionSend{A}{B}{\gtLbl{secret}, \tuple{\processVariable{n},\val{5}}}{\taut}
	}{
	\tconcat{
		\actionReceive{A}{B}{\gtLbl{secret}, \tuple{\processVariable{n},\val{5}}}{\taut}
	}{
	\tconcat{
		\actionSend{C}{B}{\gtLbl{guess},\tuple{\processVariable{x},\val{5}}}{\taut}
	}{
		\actionReceive{C}{B}{\gtLbl{guess},\tuple{\processVariable{x},\val{5}}}{\taut}
	}}}
$
}\\[1mm]
To illustrate \cref{def:WPT}, we propose two actions after \(\trace\):
	\begin{inlineenum}
	\item \(\trace[\tau_1] = \actionSend{B}{C}{\gtLbl{more},\tuple{\processVariable{\_},\val{\_}}}{x>n}\)
	\item \(\trace[\tau_2] = \actionSend{B}{C}{\gtLbl{correct},\tuple{\processVariable{\_},\val{\_}}}{x=n}\).
	\end{inlineenum}
	We can investigate whether \(\tconcat{\trace}{\trace[\tau_1]}\)
	(resp.\ \(\tconcat{\trace}{\trace[\tau_2]}\)) is a well-predicated trace
	under
	\(\emptyMap\). According to \cref{def:WPT}, we have to investigate
	whether
	\(\trace[\tau_1]\) (resp. \(\trace[\tau_2]\)) is well predicated under
	\(\map{M} = \map{\{\tuple{\processVariable{n}, \val{5}},
	\tuple{\processVariable{x}, \val{5}}\}}\).

	For \(\trace[\tau_1]\), according to \cref{def:WPT:ii} in \cref{def:WPT},
	then \(\refinement{x>n}\) must hold under \(\map{M}\), which is not the
	case, therefore \(\tconcat{\trace}{\trace[\tau_1]}\) is not
	well-predicated.

	Regarding \(\trace[\tau_2]\), according to \cref{def:WPT:ii} in
	\cref{def:WPT},
	then \(\refinement{x=n}\) must hold under \(\map{M}\), which is the
	case.
\end{example}

Finally, we consider traces that are both
valid with respect to predicates and to messages. We call those \emph{Valid
Refined Traces}. Our overall goal is to show that our framework
only produces such valid refined traces.

\begin{definition}[Valid Refined Traces\label{def:valid_refined_trace}]
	A refined trace $\traceP$ is \emph{valid} if
  \begin{inlineenum}
	  \item $\traceP$ is well-queued with respect to the empty queue
	  \(\emptyqueue\)
    \item $\traceP$ is well-predicated under the empty map \(\emptyMap\).
  \end{inlineenum}
\end{definition}

\section{Refined Communicating Automata}
\label{sec:ref_auto}
In this section, we model message-passing concurrent systems with refinements.
We ensure that this model only generates valid refined traces (c.f.
\cref{def:valid_refined_trace}). Our model of computation is an extension of
\emph{communicating systems} (CS)
\cite{brand_communicating_1983,cece_verification_2005}, which
are sets of Finite State Machines
communicating using queues. We introduce
\emph{refined communicating systems} (RCS), a variant of CS which accounts for
refinements and we show that all traces
produced by RCS are valid refined traces (\cref{thm:traces_GRA_valid_refined}).

\myparagraph{Refined Communicating Finite State Machines.}
\emph{Communicating systems}
\cite{brand_communicating_1983} are a concurrent
model of computation composed of a set of \emph{communicating finite state
machines} (CFSM) that interact with exchanges of messages. CFSM are standard
finite state machines, where labels represent actions (i.e.\
sending or
receiving messages). Individual FSM are then given a concurrent semantics,
which performs messages exchanges. The state of the system is
called a \emph{configuration}, which records the state of the individual CFSMs
as well as the content of the message queues. In this section, we adapt
communicating systems for refinements.

First, we add refinements to the transitions of CFSM, which we call
\emph{refined} CFSM.
This appears in the additional \(\setPreds\) 
in \cref{def:RCFSM} (we recall $\setPreds$ is the set of
refinements).

\begin{definition}[Refined Communicating Finite State Machine
(RCFSM)\label{def:RCFSM}]
An RCFSM is a finite transition system given
	by \(M = \tuple{Q, C, q_0, \setMsgs, \delta}\), where
\(Q\) is a set of states;
\(C =
		\setcomp{\participant{p}\participant{q}\in \setRoles^2}{\participant{p}\neq \participant{q}}\)
			is a set of channels\footnote{The original definition uses
			\emph{channels}, which we do not use. We keep them for the sake of
			consistency.};
\(q_0\in Q\) is an initial state;
\(\setMsgs\) is a finite alphabet of messages;
and
	\(\delta\subseteq Q\times (C\times \{!, ?\}\times
			\setActions\times \setPreds)\times Q\)
			is a finite set of transitions.
	\qedhere
\end{definition}

We write \(\redact{s}{s^\prime}{\actionGeneric{i}{j}{m}{r}}\) for
\(\tuple{s, \tuple{\participant{i}\participant{j}, \dagger, \msgM,
\refinement{r}}, s^\prime} \in \delta\).
\emph{Refined communicating systems} (RCS) are analogous to their non-refined
counterparts and simply consist of a tuple of RCFSM, with one RCFSM per
participant. For \emph{refined
configurations}, as with (non-refined) configurations, we store the states of
the individual CFSM and the content of queues. In addition, contrary to
non-refined configurations, refined configurations also contain a map in order
to keep track of the values of the variables in order to be able to
evaluate refinements.

\begin{definition}[Refined Communicating System (RCS)\label{def:RCS}]
	A refined communicating system is a tuple
	\(R = \tuple{M_\participant{p}}_{\participant{p}\in\setRoles}\)
	of RCFSMs such that \(M_\participant{p} = \tuple{Q_\participant{p}, C,
	q_{0\participant{p}}, \setMsgs, \delta_{\participant{p}}}\).
\end{definition}

An RCS uses one RCFSM per participant $\participant{i}\in\setRoles$. A
\emph{configuration} represents the state of such RCS, where each participant
$\participant{i}$ is in a local state $s_\participant{i}$.

\begin{definition}[Refined Configuration\label{def:RConf}]
	A \emph{refined configuration} of an RCS $R$ is a tuple
	\(S\) as follows:
	\(S \define \tuple{\tuple{s_\participant{1}, \dots, s_\participant{n}},
	\queueQ, \map{M}}_{R}\)
	where each \(s_\participant{i}\in Q_\participant{i}\), \(\queueQ\) is
	a queue of messages, and \(\map{M}\) is a map from
	variables names to values.
	Let \(\set{S}\) be the set of refined configurations.
\end{definition}

\begin{remark}\label{rmk:indexed_configuration}
	Refined configurations are indexed by their RCS. This allows the
	configuration to store the automaton of the participant. The
	semantics developed below uses those (local) transitions to infer the
	global semantics. When the context is clear, we omit this index.
\end{remark}

From that, we characterise \emph{initial} and \emph{final} configurations. We
call a configuration \emph{initial} when it is a possible configuration where
no actions have been taken yet. This means that there is no pending
messages (which would imply a previous \emph{send} action), nor known variables
(which would imply a previous action initialised the variable).
We say a configuration is \emph{final} when there are no pending messages
(otherwise, we would expect a \emph{receive} action to take place). Notice that
it does not mean the system cannot take action at all.

\begin{definition}[Initial and Final Refined Configuration\label{def:init_RGS}]
	A refined configuration\\
	\(\tuple{\tuple{s_\participant{1}, \dots,
	s_\participant{n}}, \queueQ, \map{M}}\in\set{S}\) is \emph{initial} if and
	only if
	\begin{inlineenum}
	\item \(\queueQ = \emptyqueue\)
	\item \(\map{M} = \emptyMap\)
	\item each \(s_\participant{i}\) is initial in the RCFSM.
	\end{inlineenum}

	A refined configuration \(S = \tuple{\tuple{s_\participant{1}, \dots,
	s_\participant{n}}, \queueQ, \map{M}}\in\set{S}\)
	is \emph{final} if and only if \(\queueQ = \emptyqueue\).
\end{definition}

\begin{example}[RCS\label{ex:RCS}]
	The RCFSM of participant \participant{B} in the guessing game is shown in
	\cref{fig:RCS:part_B}. See \inApp{fig:ex:RCFSM_GPM} for the RCFSM of
	\participant{A} and \participant{C}. Together, they form a RCS,
	which initial configuration is
	\(\tuple{\tuple{\participant{A_1}, \participant{B_1}, \participant{C_1}}, \emptyqueue, \emptyMap}\).
\end{example}
\begin{figure}
		\begin{tikzpicture}[node distance=4.2cm]
			\draw node[draw, rounded rectangle] (G1) {\(\participant{B_1}\)};
			\draw node[draw, rounded rectangle, right=3.5cm of G1] (G2)
			{\(\participant{B_2}\)};
			\draw node[draw, rounded rectangle, right=3.5cm of G2] (G3)
			{\(\participant{B_3}\)};
			\draw node[draw, rounded rectangle, right=3.9cm of G3] (G4)
			{\(\participant{B_4}\)};
			\coordinate [above=.5cm of G1] (start);
			\draw[-{Stealth[scale=1.5]}] (start) -- (G1);

			\draw[-{Stealth[scale=1.5]}] (G1) to node[anchor=south, pos=.5,
			sloped] {\(\actionReceive{A}{B}{\gtLbl{secret},
			\tuple{\processVariable{n}, \val{c_n}}}{\taut}\)} (G2);
			\draw[-{Stealth[scale=1.5]}] (G2) to node[anchor=south, pos=.5,
			sloped] {\(\actionReceive{C}{B}{\gtLbl{guess},
			\tuple{\processVariable{x}, \val{c_x}}}{\taut}\)} (G3);

			\coordinate (midG23) at ($ (G3)!.5!(G2) $);
			\coordinate [above=.7cm of midG23] (midG23up);
			\coordinate [below=.7cm of midG23] (midG23down);

			\draw[-{Stealth[scale=1.5]}]
			(G3) |- (midG23up) node[anchor=south]
			{\(\actionSend{B}{C}{\gtLbl{more}, \tuple{\processVariable{\_},
			\val{\_}}}{x < n}\)} -|  (G2);
			\draw[-{Stealth[scale=1.5]}]
			(G3) |- (midG23down) node[anchor=north]
			{\(\actionSend{B}{C}{\gtLbl{less}, \tuple{\processVariable{\_},
			\val{\_}}}{x > n}\)}  -|  (G2);
			\draw[-{Stealth[scale=1.5]}] (G3) to node[anchor=south, pos=.5,
			sloped] {\(\actionSend{B}{C}{\gtLbl{correct},
			\tuple{\processVariable{\_}, \val{\_}}}{x = n}\)} (G4);
		\end{tikzpicture}
		\caption{RCFSM of \(\participant{B}\) in the \(\gtFmt{G_\pm}\)
			protocol.}
		\label{fig:RCS:part_B}
\end{figure}

\myparagraph{Refined Semantics.}
We now define the semantics of RCS in~\cref{def:semantics_RGA} with two
reduction rules \grrec{} and \grsnd{} (the initial \textsc{GR} stands for
\emph{global refined}, to distinguish the rules from variants in the following
parts of this work), which are respectively used for receiving and sending
messages. To avoid confusion with RCFSM reductions, we use a double arrow
($\csred$) to represent reductions at the refined communicating system level.

Rule \grsnd{} specifies that, if a participant
\(\participant{i}\) reduces from state \(s_\participant{i}\) to state
\(s^\prime_\participant{i}\) while sending a message \(\msgM\) and if the
refinement predicate \(\refinement{r}\) attached to the action holds,
then the transition is taken at the global level.
In the resulting refined configuration, the message is
enqueued in the relevant queue and
the map of known variables \(\map{M}\) is updated
to take into account the new value of the carried variable \(\val{c}\).

Rule \grrec{} is similar, with the additional requirement that the message
received must be the next in the participant's queue (the third premise).

Notice that the verification of refinements is \emph{dynamic}, as it is
performed by the corresponding premise in each of the rules, i.e. at execution
time.

\begin{definition}[Refined Global Semantics\label{def:semantics_RGA}]
	Given a RCS $R = \tuple{M_\participant{p}}_{\participant{p}\in\setRoles}$,
	we define:
	\begin{mathpar}
	\inferrule*[left=\grrec]{
		t=\redact{s_\participant{i}}{s_{\participant{i}}^\prime}{\actionReceive{j}{i}{\gtLbl{\ell},
		\tuple{\processVariable{x}, \val{c}}}{r}}\in\delta_\participant{i}
		\and
		\updateMap{M}{x}{c}\models \refinement{r}
		\and
		\nextElemQueue{\queueQ}{j}{i} = \msgM[\tuple{\gtLbl{\ell}, \tuple{\processVariable{x}, \val{c}}}]
	}{
		\tuple{\tuple{\dots, s_\participant{i}, \dots}, \queueQ,
		\map{M}}_{R}
		\csred[t]
		\tuple{\tuple{\dots, s_\participant{i}^\prime, \dots},
		\popQueue{\queueQ}{j}{i},
		\updateMap{M}{x}{c}}_{R}
	}
	\and
	\inferrule*[left=\grsnd]{
		t=\redact{s_\participant{i}}{s_{\participant{i}}^\prime}{\actionSend{i}{j}{\gtLbl{\ell},
		\tuple{\processVariable{x}, \val{c}}}{r}}\in\delta_\participant{i}
		\and
		\updateMap{M}{x}{c}\models \refinement{r}
	}{
		\tuple{\tuple{\dots, s_\participant{i}, \dots}, \queueQ,
		\map{M}}_{R}
		\csred[t]
		\tuple{\tuple{\dots, s_{\participant{i}}^\prime, \dots},
		\pushQueue{\queueQ}{\msgM[\tuple{\gtLbl{\ell},
		\tuple{\processVariable{x}, \val{c}}}]}{i}{j},
		\updateMap{M}{x}{c}}_{R}
	}
	\qedhere
	\end{mathpar}
\end{definition}

\begin{remark}
	Global transitions are labelled with the underlying local transition. When
	the local transition is not relevant, we do not show it.
\end{remark}

\begin{example}[Transitions of a RCS]
	\label{ex:RCS_transitions}
	Considering the RCS
	\iftoggle{full}{in \cref{fig:ex:RCFSM_GPM}}{of $\gtFmt{G_\pm}$ (\cref{fig:RCS:part_B})} in its
	initial configuration
	\(C_i = \tuple{\tuple{\participant{A_1}, \participant{B_1}, \participant{C_1}}, \emptyqueue, \emptyMap}\),
	we have that the automaton of \(\participant{A}\) can fire a transition
	\(\redact{\participant{A_1}}{\participant{A_2}}{\actionSend{A}{B}{\gtLbl{secret},\tuple{\processVariable{n},\val{5}}}{\taut}}\),
	and \(\updateMap{\emptyMap}{n}{5}\models\taut\), by definition of $\taut$.
	Therefore, \(C_i\) can take a \grsnd{} transition and reduce to
	\(\tuple{\tuple{\participant{A_2}, \participant{B_1}, \participant{C_1}}, \queueQ, \map{\{\tuple{\processVariable{n}, \val{5}}\}}}\),
	where \(\queueQ\) contains a single message
	\(\msgM[\tuple{\gtLbl{secret}, \tuple{\processVariable{n}, \val{5}}}]\)
	in \(\lookupQueue{\queueQ}{A}{B}\).

	If the RCS is in the configuration
	\(C = \tuple{\tuple{\participant{A_2}, \participant{B_3}, \participant{C_2}}, \emptyqueue, \map{M}}\)
	with \(\map{M} = \map{\{\tuple{\processVariable{x}, \val{5}},
	\tuple{\processVariable{n}, \val{5}}\}}\),
	the RCFSM of participant \(\participant{B}\) offers three possible
	transitions:
	\begin{inlineenum}
\item
			\(\redact{\participant{B_3}}{\participant{B_2}}{\actionSend{B}{C}{\gtLbl{more},\tuple{\processVariable{\_},\val{\_}}}{x
			 < n}}\)
\item
			\(\redact{\participant{B_3}}{\participant{B_2}}{\actionSend{B}{C}{\gtLbl{less},\tuple{\processVariable{\_},\val{\_}}}{x
			 > n}}\)
\item
			\(\redact{\participant{B_3}}{\participant{B_4}}{\actionSend{B}{C}{\gtLbl{correct},\tuple{\processVariable{\_},\val{\_}}}{x
			 = n}}\).
	\end{inlineenum}
	The predicates carried in first two do not hold under \(\map{M}\):
	\(\map{M}\not\models\refinement{x < n}\) (resp.\ for \(\refinement{x >n}\)).
	Therefore, only \(\redact{\participant{B_3}}{\participant{B_4}}{\actionSend{B}{C}{\gtLbl{correct},\tuple{\processVariable{\_},\val{\_}}}{x = n}}\)
	is feasible as a \grsnd{} transition in the RCS.
	As we will see below (\cref{thm:traces_GRA_valid_refined}), this semantics
	prevents invalid traces.
\end{example}

\myparagraph{Trace of Refined Communicating Systems.}
In order to show that the semantics of RCS captures the intuition
of refinements, we study the traces formed by sequences of
reductions (see \inApp{def:trace_RGA} for the formal definition of traces of
RCS).

\begin{example}[Trace of an RCS\label{ex:trace_RCS}]
The trace \(\tconcat{\trace}{\trace[\tau_2]}\) (\cref{ex:WPT}) is a trace of
the RCS of \(\gtFmt{G_\pm}\).
\end{example}

We conclude this section with our main result, which is that all traces
produced by \(\RcsOfType{G}\) are valid refined traces. A trace is
\emph{initial} (resp. \emph{final}) if it is obtained from a
run whose first (resp. last) state is initial (resp. final).

\begin{restatable}[Traces of Refined Communicating Systems are Valid Refined Traces]{theorem}{tracesGRAvalidrefined}
\label{thm:traces_GRA_valid_refined}
	For all RCS $R$, for all initial and final traces $\trace$ of $R$, $\trace$
	is a valid refined trace.
\end{restatable}
The proof is in \inApp{app:refined_automata}.

\section{Refined Multiparty Session Types (RMPST)}
\label{sec:refined_mpst}
In the two previous sections, we introduced refinement validity and
a variant of CS which is correct with respect to our validity criterion.
However, working with
RCS is cumbersome, in particular if we intend to prove additional properties
(e.g.\ deadlock freedom). Fortunately, various models for message-passing
concurrent computation have been developed in the literature, many of which can
be encoded into CS. Multiparty session types (MPST)
\cite{yoshida_very_2020,honda_multiparty_2016,honda_multiparty_2008} is an
example of such
model. We focus on MPST as they have proved successful for many applications
and the theory enjoy many useful properties (e.g. session fidelity,
deadlock freedom, liveness etc). However, MPST is not
the only possible choice, and we sketch different input models in
\inApp{sec:refined_choreo_automata}. In this section, we introduce
\emph{refined multiparty session types} (RMPST), which are an extension of MPST
annotated with refinement predicates and we show how one can extend
existing models to easily obtain refinements.

In \cref{sec:RMPST_Syntax}, we first present the syntax of \emph{global} and
\emph{local} refined multiparty session types, adapted for refinements. In
\cref{sec:MPST_to_CFSM}, we present how to obtain RCS from local RMPST,
extending a standard approach to implement MPST in CS
\cite{Denielou_Multiparty_2012} with refinements.

\subsection{Syntax of RMPST\label{sec:RMPST_Syntax}}
We define the syntax of RMPST.
First we assume that messages carry different sorts of payload. As a reminder,
for simplicity, in our examples, we only consider \(\processType{int}\)
payloads. Also, we recall the conventions from \cref{sec:pred_lang_sem}:
$\setRoles$ is the set of participants and $\setLabels$ is the set of labels.
For recursion, we introduce type variables that range over \(\{\typeVariable{T},
\typeVariable{U}, \dots\}\); \(\typeVariable{t}\) is a meta-variable taken over
the set of type variables. We assume
all type variables appearing in a type are distinct and we do not (syntactically) distinguish global and local type variables.
Finally,
\(\processVariable[i]{x}\) are meta-variables over payload variables taken from
the set $\setVars$.

We first define \emph{global refined multiparty session types}, which are
inductive data types generated by the production \(\nonterm{\gSessionType{G}}\)
in \cref{fig:syntax_glo_loc_types}. The type
\(\gtComm[i][I]{A}{B}{l}[x]{S}[r]{G}\) describes a protocol where
\(\participant{A}\) chooses a label \(\gtLbl{l_i}\) amongst possible \(I\) and
sends a message to \(\participant{B}\). The message contains a payload of type
\(\processType{S_i}\), which is bound to \(\processVariable{x_i}\) when sent.
\emph{Refinement predicates} we
introduce guard the communication they are attached to, meaning the system can
select a choice with predicate \(\refinement{r_i}\) only if
\(\refinement{r_i}\) holds. In that case, the message is sent and the protocol
continues with its continuation of type \(\gtFmt{G_i}\).
\(\gtRec{T}{G}\) binds \(\typeVariable{T}\) in
\(\gSessionType{G}\), and a bound type variable \(\gtRecVar{T}\) in a type denotes
a protocol recursion. Let $\frv{\gtFmt{G}}$ denotes the free recursion
variables occuring in $\gtFmt{G}$. Finally \(\gtEnd\) describes a terminated
protocol.
Let \(\rolesOfGT{G}\) be the set of participants that appear
in \(\gSessionType{G}\)
(c.f.\ \inApp{def:role_in_gtype} for the
definition of \(\rolesOfGT{G}\)). We write \(\participant{p}\in\gtFmt{G}\) for
\(\participant{p}\in\rolesOfGT{G}\).
\begin{figure}
\begin{mathpar}
\small
\begin{grammar}
\firstcase{\gSessionType{G}}{\gtComm[i\in I]{\participant{p}}{\participant{q}}{l_i}[x_i]{\nonterm{S}}[\nonterm{R}]{\nonterm{G}}\gralt
\gtRec{t}{\nonterm{G}}}{communication, recursive type}
\otherform{\gtRecVar{t} \gralt \gtEnd}{type variable, termination}
\firstcase{\lSessionType{L}}{
			\ltIntC[i\in I]%
			{\participant{p}}{l_{i}}{x_{i}}{\nonterm{S}}[\nonterm{R}]{\lSessionType{\nonterm{L}}}
						\gralt
			\ltRecVar{t} \gralt
			\ltEnd
		}{internal choice, type variable, termination}
		\otherform{
			\ltExtC[i\in I]%
			{\participant{p}}{l_{i}}{x_i}{\nonterm{S}}{\nonterm{R}}{\lSessionType{\nonterm{L}}}
			\gralt
			\ltRec{t}{\nonterm{L}}
			}{external choice, recursive type}
		\firstcase{\processType{\nonterm{S}}}{\processType{int}\gralt\dots}{sort
		 (payload
		types)}
	\end{grammar}
	\end{mathpar}
	\caption{Syntax of Global (\(\nonterm{\gSessionType{G}}\)) and Local (\(\nonterm{\lSessionType{L}}\)) Types and Sorts (\(\processType{\nonterm{S}}\)).}
	\label{fig:syntax_glo_loc_types}
\end{figure}
\begin{example}[Refined Global Multiparty Session Type]
	The type \(\gtFmt{G_\pm}\) presented in \cref{sec:intro} is a refined
	global type; we have \(\rolesOfGT{G_\pm} = \{\participant{A},
	\participant{B}, \participant{C}\}\).
\end{example}
To characterise the
behaviour of individual participants, we define \emph{refined local multiparty
session types}, which are inductive datatypes generated by
\(\nonterm{\lSessionType{L}}\) in \cref{fig:syntax_glo_loc_types}.
Recursion, type variables and termination are similar in local and global types.
Only the communication specifications differs: in a local type
\(\ltIntC[i][I]{\participant{p}}{l}{x}{S}[r]{L}\)
describes an \emph{internal choice}, i.e. the participant chooses a label
\(\gtLbl{l_i}\) and sends it to \(\participant{p}\). Conversely,
\(\ltExtC[i][I]{\participant{p}}{l}{x}{S}{r}{L}\)
describes an \emph{external choice}: \(\participant{p}\) makes a choice amongst
the possible \(\gtLbl{l_i}\) and the local participant \emph{receives} this
choice.

Global and local MPST are related: we can \emph{project} a global type onto the
local types of its participants. Below, we define a \emph{projection}
(partial) operator \(\proj{p}{G}\),
which returns the local type of \(\participant{p}\) with respect
to the global type \(\gtFmt{G}\).

We define a projection with a \emph{merge} (partial) operator, which merges multiple local types of a
participant into a single local type. This is used to merge the (possibly
different) types of the continuations present in the communication branches. The
study of different variants of merge operators is an active field (e.g.\
\cite[Section~3]{StutzAsynchronous2023}).
For the sake of simplicity, in this paper we use a naïve merge operator, which
simply ensures that all types are the same.


\begin{definition}[Projection\label{def:proj}]
	Given \(\participant{p}\), \(\participant{q}\) and \(\participant{r}\)
	three distinct participants:\\
	$
	\begin{aligned}
	\proj{p}{\gtComm[i][I]{\participant{p}}{\participant{q}}{l}[x]{S}[R]{G}}
	&= \ltIntC[i\in I]{\participant{q}}{l_i}{x_i}{{\it
			\processType[i]{S}}}[R_i]{\proj{p}{G_i}} \\
	\proj{p}{\gtComm[i][I]{\participant{q}}{\participant{p}}{l}[x]{S}[R]{G}}
	&= \ltExtC[i\in I]{\participant{q}}{l_i}{x_i}{{\it
			\processType[i]{S}}}{\taut}{\proj{p}{G_i}} \\
	\proj{p}{\gtComm[i][I]{\participant{q}}{\participant{r}}{l}[x]{S}[R]{G}}
	&= \merge{(\proj{p}{G_i})}{i\in I}\\
	\end{aligned}$
\begin{mathpar}
		\proj{p}{\gtRec{t}{G^\prime}} = \begin{cases}
			\ltRec{t}{(\proj{p}{G^\prime})} & \text{if
			}\participant{p}\in\gtFmt{G^\prime} \text{ or }
			\frv{\gtFmt{G^\prime}}\neq\emptyset\\
			\ltEnd & \text{otherwise}
		\end{cases}
\and
		\proj{p}{\gtRecVar{t}} = \ltRecVar{t}
\and
		\proj{p}{\gtEnd} = \ltEnd
\end{mathpar}
where a \emph{merge} operator is defined as:
	\(\merge{L_i}{i\in I}\define
\lSessionType{L}\) if \(\forall i\in I\suchthat \lSessionType{L} =
	\lSessionType{L_i}\),
undefined otherwise.
\end{definition}
\todo{This definition has been rewritten to fix presentation issues and
problems with recursion.}

\noindent Notice that our local RMPST accept refinements on both
receiving and sending, and the semantics developed in \cref{sec:ref_auto}
accept any position for verification. When projecting a global type
\(\gtFmt{G} = \gtComm{A}{B}{\ell}[x]{int}[r]{\gtEnd}\) onto local types, we
therefore have a choice to project the refinement:
\begin{itemize}
	\item on the send side:
		\(\proj{A}{G} = \ltIntC{B}{\ell}{x}{int}[r]{\ltEnd}\) and \(\proj{B}{G}
		= \ltExtC{A}{\ell}{x}{int}{\taut}{\ltEnd}\)
	\item on the receive side:
		\(\proj{A}{G} = \ltIntC{B}{\ell}{x}{int}[\taut]{\ltEnd}\) and
		\(\proj{B}{G} = \ltExtC{A}{\ell}{x}{int}{r}{\ltEnd}\)
	\item or a combination of both\footnote{For instance, if we want to
	implement a centralised server that communicates with (isolated) clients,
	we may want all refinements to be asserted by the server, independently of
	the direction.}.
\end{itemize}
Our projection takes the first option, i.e.\ refinements are checked when the
message is emitted, but with any of these choices, our developments would not
substantially change.

\begin{example}[Projection\label{ex:plus_minus_proj}]
	We project \(\gtFmt{G_\pm}\) (\cref{sec:intro}) onto
	participants \(\participant{A}\) and \(\participant{B}\)\footnote{The
	projection onto \(\participant{C}\) is similar to the recursive part
	of the projection onto
	\(\participant{B}\), with ! and ? swapped.}:
$
\begin{array}{l}
\small\proj{A}{G_\pm} = \ltIntC{\participant{B}}{secret}{n}{int}[\taut]{\ltEnd}\\
	\proj{B}{G_\pm} =\\ \ltExtC{\participant{A}}{secret}{n}{int}{\top}{
	\ltRec{T}{
		\ltExtC{C}{guess}{x}{int}{\top}{
			\ltIntCRaw{C}{
			\begin{aligned}
				&\commChoice{more}{}[x < n]{\ltRecVar{T}},\\
				&\commChoice{less}{}[x > n]{\ltRecVar{T}},\\
				&\commChoice{correct}{}[x = n]{\ltEnd}
			\end{aligned}
			}
		}
	}
	}
\end{array}$
\newline\qedhere
\end{example}

\subsection{From Refined MPST to Refined Communicating
System\label{sec:MPST_to_CFSM}}

In this subsection, we show how to generate an RCS from local RMPSTs.
As shown in \cref{def:proj}, local types are projected from global multiparty
session types. Therefore, this step allows us to complete the conversion from a
global RMPST into an RCS.
We adapt the translation from local type to CFSMs presented
in~\cite{ICALP13CFSM} to accommodate refinements in types.

The intuition behind the translation is to decompose a local type into the
individual steps it specifies. For this, we first need to retrieve all
those steps. We define the set of types that occur nested in
another type: a type $\ltFmt{T^\prime}$ \emph{occurs} in a type $\ltFmt{T}$
(noted \(\occursLoc{T^\prime}{T}\)) if
it appears in the continuations of $\ltFmt{T}$ after one or multiple
steps (see \inApp{def:type_occurring_type}).

Given this, we can proceed to the translation itself, in \cref{def:LT_to_RCFSM}.
This definition says that the states of the RCFSM of a local type
\(\ltFmt{T_0}\) are composed of
the (sub)types that appear in \(\ltFmt{T_0}\), stripped of the leading
\(\ltRec{t}{}\) (the function \(\operatorname{strip}\) removes the
leading recursions variables; this formalises \citet[Item~(2) in
Definition~3.4]{ICALP13CFSM}) and of recursive variables \(\ltRecVar{t}\). We
define the set of transitions of this RCFSM by taking the action each type
(i.e. each state) can take, and adding a transition with this action from the
initial state to the continuation (stripped of leading \(\ltRec{t}{}\)). In the
case that the continuation is a recursion variable \(\ltRecVar{t}\), we have to
search in the original type the continuation. Compared to \citet[Item~(2) in
Definition~3.4]{ICALP13CFSM}, we simply add the support for the refinements
predicates, which appear both in the types (i.e. in the states) and in the
actions (i.e. in the transitions).

\begin{definition}[RCFSM of Refined Local Types (extends
Definition~3.5\ in \citet{ICALP13CFSM})\label{def:LT_to_RCFSM}]
	Given a global type \(\gSessionType{G}\), the RCFSM of participant
	\(\participant{p}\) (with local type \(\lSessionType{T_0} =
	\proj{p}{G}\)) is the automaton \(\ltToAut{T_0} = \tuple{Q, C,
	\operatorname{strip}(\ltFmt{T_0}),
		\setMsgs, \delta}\) where:%
\begin{itemize}[leftmargin=*]
	\item
\(Q = \setcomp{\lSessionType{T^\prime}}{\occursLoc{T^\prime}{T_0}
	\wedge \lSessionType{T^\prime}\neq \ltRecVar{t} \wedge
	\lSessionType{T^\prime}\neq \ltRec{t}{T_\mu}}\);
\item
	\(C = \setcomp{\participant{p}\participant{q}}{\participant{p},
		\participant{q} \in \gSessionType{G}, \participant{p} \neq
			\participant{q}}\); and
		\item
\(\delta\) is the smallest set of transitions such that:
for all \(\occursLoc{T}{T_0}\) in \(Q\), for all \(\val{c}\in\setVals\):
\begin{itemize}[leftmargin=*]
	\item if \(\ltFmt{T}\) is \(\ltIntC[i][I]{q}{\ell}{x}{S}[r]{T}\),
			for all \(\lSessionType{T_i}\):
			\begin{itemize}[leftmargin=*]
				\item if \(\lSessionType{T_i}\neq \ltRecVar{t}\),
				then \(\tuple{\lSessionType{T},
					\actionSend{p}{q}{\gtLbl{\ell_i},
						\tuple{\processVariable{x}, \val{c}}}{r},
					\operatorname{strip}(\lSessionType{T_i})}\in\delta\)
				\item if \(\lSessionType{T_i} = \ltRecVar{t}\) with
				\(\ltRec{t}{T^\prime}\in \lSessionType{T_0}\),
				then \(\tuple{\lSessionType{T},
					\actionSend{p}{q}{\gtLbl{\ell_i},
						\tuple{\processVariable{x}, \val{c}}}{r},
					\operatorname{strip}(\lSessionType{T^\prime})}\in\delta\)
			\end{itemize}
			\item if \(\ltFmt{T}\) is \(\ltExtC[i][I]{q}{\ell}{x}{S}{r}{T}\),
			for all \(\lSessionType{T_i}\):
			\begin{itemize}[leftmargin=*]
				\item if \(\lSessionType{T_i}\neq \ltRecVar{t}\),
				then \(\tuple{\lSessionType{T},
					\actionReceive{q}{p}{\gtLbl{\ell_i},
						\tuple{\processVariable{x}, \val{c}}}{r},
					\operatorname{strip}(\lSessionType{T_i})}\in\delta\)
				\item if \(\lSessionType{T_i} = \ltRecVar{t}\) with
				\(\ltRec{t}{T^\prime}\in \lSessionType{T_0}\),
				then \(\tuple{\lSessionType{T},
					\actionReceive{q}{p}{\gtLbl{\ell_i},
						\tuple{\processVariable{x}, \val{c}}}{r},
					\operatorname{strip}(\lSessionType{T^\prime})}\in\delta\)
			\end{itemize}
	\end{itemize}
	\end{itemize}
	where
$\operatorname{strip}(\ltFmt{T})\define
\operatorname{strip}(\ltFmt{T^\prime})$ if  $\ltFmt{T} =
		\ltRec{t}{T^\prime}$; and
$\operatorname{strip}(\ltFmt{T}) \define \ltFmt{T}$ otherwise.
\end{definition}

Finally, we define the \emph{RCS of a type}.

\begin{definition}[Refined Communicating System of a Type]
	\label{def:RCS_of_type}
	The RCS of a type \(\gSessionType{G}\),
	noted \(\RcsOfType{G}\), is a tuple composed of the RCFSM of all
	participants
	\(\RcsOfType{G}\define
	\tuple{\ltToAut{\proj{p}{G}}}_{\participant{p}\in\rolesOfGT{G}}\).
	We note \(\ConfOfType{G}\) the initial configuration of \(\RcsOfType{G}\).
\end{definition}

\begin{example}[Refined Communicating System of \(\gtFmt{G_\pm}\)]
	The communicating system of \(\gtFmt{G_\pm}\) is
	\(\RcsOfType{G_\pm} = \tuple{\ltToAut{\proj{A}{G_\pm}},
	\ltToAut{\proj{B}{G_\pm}}, \ltToAut{\proj{C}{G_\pm}}}\)%
	\iftoggle{full}{%
		, where the three automata are shown in \cref{fig:ex:RCFSM_GPM}
		}{}.

	The initial configuration \(\ConfOfType{G_\pm}\) of this RCS
	\(\RcsOfType{G_\pm}\) is
	\(\tuple{\tuple{\participant{A_1}, \participant{B_1}, \participant{C_1}},
	\emptyqueue, \emptyMap}\).
\end{example}

\cref{thm:traces_GRA_valid_refined} applies to RCS obtained from RMPST:
RCS generated from \cref{def:RCS_of_type} only produce valid refined traces,
with the refined global semantics presented in \cref{def:semantics_RGA}. Notice
also that, if refinements always hold, RMPST and their semantics coincide with
the semantics presented in \cite{Denielou_Multiparty_2012}.

\section{Decentralised Refined Multiparty Session Types}
\label{sec:dynamic}
In the previous section, we presented RCS and we showed that every trace of an
RCS is a valid refined trace. However, RCS are theoretical constructions and
are not intended to be implemented directly, as they use a global shared map of
variables. In practice, a user may want to develop more precise analysis
techniques on specific classes of RCS to remove this global map, which allows a
decentralised verification of refinements, while keeping the validity of refined
traces.

The goal of this section is twofold: on the one hand, the decentralised
semantics we develop serves as a theoretical background for our implementation
(\cref{sec:implem}). On the other hand, it illustrates the modularity of our
framework. We show that the decentralised approach produces valid refined traces
by showing refined configurations we developed in \cref{sec:ref_auto} simulate
decentralised systems. This approach is not specific to our variant: we expect
other optimisations presented in the literature could be integrated similarly.

This section is divided in the following steps: first, we define what we mean by
\emph{decentralised verification of the refinements}, by adapting the semantics
of RCS (\cref{def:DConf,def:semantics_DGA}). We split the global map of
variables' values into local maps (one per participant). Then, we show that
despite being
modified, the new
variant still produces valid refined traces (\cref{def:valid_refined_trace}).
We justify this claim by proving that under some
conditions, the original RCS \emph{simulates}
(c.f.\ \cite[Exercise~1.4.17, p.\ 26]{SangiorgiIntroduction2012}) the
decentralised variant
(\cref{thm:st_sim_dyn}).
Since trace equivalence is coarser than simulation, this is sufficient to
prove that decentralised configurations that meet the said conditions produce
valid refined traces.

The conditions we mentioned above are:
\begin{inlineenum}
	\item variables should not be duplicated
	\item when evaluating a predicate, the free variables of the predicate must
	be
	in the local map.
\end{inlineenum}
Without the first condition, we can possibly have two distinct values assigned
to the
same variable without being able to distinguish which is the most recent. The
second condition is required to verify the refinements locally
(e.g. predicates that constraint an action of \(\participant{A}\) should be
checked by \(\participant{A}\) itself). To illustrate the importance of the
first condition, consider the type
$\gtComm{A}{B}{\ell_1}[x]{int}{\gtComm{C}{D}{\ell_2}[x]{int}{\gtEnd}}$. In the
centralised approach, $\processVariable{x}$ is aliased, while in the
decentralised approach, the $\processVariable{x}$ exchanged between
$\participant{A}$ and $\participant{B}$ is stored in a local map, and the
$\processVariable{x}$ exchanged between $\participant{C}$ and $\participant{D}$
is stored in another local map; both are not aliased. To prevent different
semantics, we need to prevent such difference, which is the goal of the first
condition.

\myparagraph{Decentralised Configurations and Decentralised
Semantics.\label{sec:dynamic_CS}}
First, we define \emph{decentralised} configurations in \cref{def:DConf}.
Compared to \cref{def:RConf}, instead of a global map in the tuple, we
associate a local map to each automata state. Those maps store the variables
each participant has access to.

\begin{definition}[Decentralised Configuration\label{def:DConf}]
	A \emph{Decentralised Configuration} of an RCS
	\(\RcsOfType{G} =
	\tuple{\tuple{Q_\participant{i}, C_\participant{i}, q_{0, \participant{i}},
	\set{A}, \delta_\participant{i}}}_{\participant{i}\in\rolesOfGT{G}}\) is a
	tuple
	\(\tuple{\tuple{\tuple{s_\participant{1}, \map{M}_\participant{1}}, \dots,
	\tuple{s_\participant{n}, \map{M}_\participant{n}}},
	\queue{w}}_{\RcsOfType{G}}\)
	where each \(s_\participant{i}\in Q_\participant{i}\), each
	\(\map{M}_\participant{i}\) is a local map of variables to values, and
	\(\queue{w}\) is a queue of messages.

	Let \(\set{S}_D\) be the set of decentralised configurations.
	We note \(\DConfOfType{G}\) the initial decentralised configuration of \(\RcsOfType{G}\).
\end{definition}

\Cref{rmk:indexed_configuration} also applies for decentralised configurations.

\begin{example}[Initial decentralised configuration of \(\gtFmt{G_\pm}\)]
	In \cref{ex:RCS}, we presented the refined communicating system of
	\(\gtFmt{G_\pm}\) and its associated refined configuration.
	The \emph{initial decentralised configuration} of this system is
	\(\tuple{\tuple{\participant{A_1}, \emptyMap}, \tuple{\participant{B_1}, \emptyMap}, \tuple{\participant{C_1}, \emptyMap}, \emptyqueue}\).
	In particular, notice that it uses the same set of refined CFSM than the
	refined configuration.
\end{example}

The global reduction rules are adapted
accordingly: in the rules \drec{} and \dsnd{} ("D" stands for
"decentralised"), when a message is sent or received, the corresponding local
map is updated, instead of a global map as in \grrec{} and \grsnd{}.

\begin{remark}
	Contrary to \cref{def:semantics_RGA}, when a variable is
	sent, it is removed from the local map of variables. Intuitively,
	when a participant sends a variable, it erases its knowledge of it,
	to prevent aliasing issues. A direct consequence of this is that, in
	the centralised implementation, the global map of variables is a
	\emph{superset} of the local maps in the corresponding decentralised
	implementation. Indeed, while a variable is in transit, it appears
	neither in the sender's map, nor in the receiver's one. This
	observation will be proved together with the simulation proof
	(\cref{thm:st_sim_dyn}).
\end{remark}

\begin{definition}[Decentralised Global Semantics\label{def:semantics_DGA}]
	Given an RCS $R=\tuple{M_\participant{p}}_{\participant{p}\in\setRoles}$
	\begin{mathpar}
	\inferrule*[left=\drec{}]{
		t=\redact{s_\participant{i}}{s_{\participant{i}}^\prime}{\actionReceive{j}{i}{\gtLbl{\ell},
		\tuple{\processVariable{x}, \val{c}}}{r}}\in\delta_{\participant{i}}
		\and
		\nextElemQueue{w}{j}{i} = \msgM[\tuple{\gtLbl{\ell}, \tuple{\processVariable{x}, \val{c}}}]
		\and
		\updateMap{M_\participant{i}}{x}{c}\models\refinement{r}
	}{
		\tuple{\tuple{ \dots, \tuple{s_\participant{i},
		\map{M}_\participant{i}}, \dots},
		\queue{w}}_{R}
		\csred[t]
		\tuple{\tuple{
		\dots, \tuple{s_\participant{i}, \updateMap{M_\participant{i}}{x}{c}},
		\dots},
		\popQueue{w}{j}{i}}_{R}
	}
	\and
	\inferrule*[left=\dsnd{}]{
		t=\redact{s_\participant{i}}{s_{\participant{i}}^\prime}{\actionSend{i}{j}{\gtLbl{\ell},
		\tuple{\processVariable{x}, \val{c}}}{r}}\in\delta_{\participant{i}}
		\and
		\updateMap{M_\participant{i}}{x}{c}\models \refinement{r}
	}{
		\tuple{\tuple{
		\dots, \tuple{s_\participant{i}, \map{M}_\participant{i}}, \dots},
		\queueQ}_{R}
		\csred[t]
		\tuple{\tuple{
		\dots, \tuple{s_\participant{i},
		\removeMap{\map{M}_\participant{i}}{x}}, \dots
		},
		\pushQueue{\queueQ}{\msgM[\tuple{\gtLbl{\ell},
		\tuple{\processVariable{x},	\val{c}}}]}{i}{j}}_{R}
	}
	\qedhere
	\end{mathpar}
\end{definition}

\myparagraph{Conditions for Decentralised Verification and Correctness Proofs.}
We now focus on proving that this decentralised semantics produces valid refined
traces. As we mentioned above, this holds under two conditions, which we define
first:

\begin{definition}[Conditions for Decentralised Verification Simulation\label{def:cond_dyn_verif}]
	Given a decentralised configuration
	\(\tuple{\tuple{\tuple{s_\participant{i},\map{M_\participant{i}}}, \dots},
	\queueQ}\),
	the \emph{conditions for simulation} are:
	\begin{enumerate}
		\item No duplication: \begin{enumerate}
				\item if
					\(\exists\map{M_\participant{i}}\suchthat\processVariable{x}\in\domMap{M_\participant{i}}\),
					then
					\(\forall\participant{i},\participant{j}\suchthat\processVariable{x}\not\in\lookupQueue{\queueQ}{i}{j}\)
					and
					\(\forall\participant{j}\neq\participant{i}\suchthat\processVariable{x}\not\in\domMap{M_\participant{j}}\).
				\item if
					\(\exists\tuple{\participant{i},\participant{j}}\suchthat\processVariable{x}\in\lookupQueue{\queueQ}{i}{j}\),
					then
					\(\forall\participant{i}\suchthat\processVariable{x}\not\in\domMap{M_\participant{i}}\)
					and
					\(\forall\tuple{\participant{i^\prime},\participant{j^\prime}}\neq\tuple{\participant{i},\participant{j}}\suchthat\processVariable{x}\not\in\lookupQueue{\queueQ}{i^\prime}{j^\prime}\).
		\end{enumerate}
		\item Free variables are in the map:
		\(\forall\participant{i}\suchthat\forall
					s_\participant{i}^\prime\suchthat
					\redact{s_\participant{i}}
						{s_{\participant{i}}^\prime}
						{\actionGeneric{i}{j}
							{\gtLbl{\ell},\tuple{\processVariable{x},\val{c}}}{r}}\suchthat
					\fv{r}\subseteq\domMap{\updateMap{M_\participant{i}}{x}{c}}\)\qedhere
\end{enumerate}
\end{definition}
\begin{definition}[Decentralisable Type\label{def:decentralisable_type}]
	A type $\gtFmt{G}$ is \emph{decentralisable} if the two conditions hold for
	all reachable decentralised configurations
	from \(\DConfOfType{G}\).
\end{definition}

Notice that the second condition is redundant, as the condition
\(\updateMap{M_\participant{i}}{x}{c}\models \refinement{r}\) (in the premises
of the reduction rules) already requires
that \(\fv{r}\) is a subset of the variables in
\(\updateMap{M_\participant{i}}{x}{c}\). Even without making this
condition explicit, the system would stall if a predicate cannot be verified.
For the
sake of clarity, we keep it explicit in the two required conditions.

We now observe a correspondence between the (centralised) refined
configuration and the decentralised configuration of a global type
\(\gtFmt{G}\).
To characterise the correspondence between centralised and
decentralised configuration, we establish a \emph{simulation} relation
between the two (see \inApp{sec:app:simulation} and
\citet{SangiorgiIntroduction2012}). Intuitively, a simulation
captures the fact that one
system (the centralised configuration in our case) can mimic all
transitions of another system (the decentralised one here).

We can now prove the main result of this section, which is that the
decentralised semantics does not induce new (unwanted) behaviours, i.e. all
decentralised transitions can be mimicked by centralised transitions, i.e.\ the
centralised approach simulates the decentralised one.

\begin{restatable}[Centralised simulates
Decentralised\label{thm:st_sim_dyn}]{theorem}{centralisedSimulatesDecentralised}
	For all decentralisable RMPST \(\gtFmt{G}\)
	(\cref{def:decentralisable_type}),
	\(\ConfOfType{G}\) simulates \(\DConfOfType{G}\).
\end{restatable}
\begin{proof}
The proof is available in \inApp{sec:app:centralised_simulated_decentralised}.
\end{proof}

This result shows that any type that verifies the conditions stated in
\cref{def:cond_dyn_verif} can be verified in a decentralised way. The
difficulty is that the conditions are about the execution: we do not know
whether a predicate will have a missing variable during the
execution. With a knowledge flow algorithm, we can infer (from the
communication specifications in the global type) which participant has
access to which variables at any point in the execution of the protocol, i.e.\
we can \emph{localise} each variable throughout the execution of the protocol.
This algorithm (which we present in \inApp{sec:static_verif_conditions}) does
not present major challenges.

Notice that the reverse simulation does
not hold: \(\DConfOfType{G}\) does not simulate \(\ConfOfType{G}\). Indeed,
\(\ConfOfType{G}\) can verify a predicate whose variables are spread over
different participants, i.e. where variables would be spread across multiple
\(\map{M_\participant{i}}\) in the decentralised variant.

\section{Static Elision of Redundant Refinements}
\label{sec:static_elision}
In this section, we present a second optimisation, which is complimentary from
the first one. The main idea is to statically analyse a given protocol to
find and remove redundant refinements. Our approach is to consider a
\emph{target} transition, which we
	aim to remove the refinement, if possible. Our optimisation can then be
	applied
	successively to different target transitions one after each other.
For instance, consider the
following protocol $\gtFmt{G_s}$. We target the
second refinement,
$\refinement{x < 10}$, which necessarily holds if the first one does
(since
$\processVariable{x}$ does not change). Therefore it is redundant and can be
removed.\\[1mm]
\centerline{\small
$\gtFmt{G_s} = \gtComm{A}{B}{\ell_1}[x]{int}[x < 0]{
	\gtComm{A}{B}{\ell_2}[y]{int}[x < 10]{\gtEnd}
}$}

However, removing refinements is not always trivial, since the
communication semantics is asynchronous. Consider for instance the following
type~:
\\[1mm]
\centerline{\small
	$\gtComm{A}{C}{\ell_1}[x]{int}[x > 20]{
	\gtComm{A}{B}{\ell_2}[x]{int}[x < 0]{
	\gtComm{C}{B}{\ell_3}[y]{int}[x < 10]{\gtEnd}
}}$
}\\[1mm]

	A naïve approach would be to remove the refinement of the last
	communication ($\refinement{x < 10}$), since the previous
	communication has a stronger guarantee ($\refinement{x < 0}$). However, due
	to the asynchrony of communications, the second and third
	communications could be swapped at runtime, but the refinement
	($\refinement{x < 10}$)
	does not hold before the action $\gtComm{A}{B}{\ell_2}[x]{int}[x <
	0]{\dots}$ occurs. Therefore, in this case, removing the last refinement is
	incorrect. The optimisation we present takes into account those cases, by
	keeping track of causal relations between actions.

This section is independent of the previous one, although this second
optimisation can help to make
some protocols localisable: for instance, $\gtFmt{G_s}$ above is not
localisable. Since the second step
$\gtComm{A}{B}{\ell_2}[y]{int}[x < 10]{\gtEnd}$
requires $\participant{A}$ to access $\processVariable{x}$, which is at
$\participant{B}$. However, once removed, the protocol becomes localisable, and
can therefore be decentralised, helping the first optimisation introduced in
\cref{sec:dynamic}.

As with the previous section, the optimisation we present could easily be
further improved. Here, we focus on a simple case, as our goal is not to
discuss the optimisation itself, but rather to show the versatility of the
framework.

We present this section in two steps: first, in \cref{sec:static_elision:RCS},
we focus on RCS, which form the core of our framework; then, in
\cref{sec:static_elision:MPST}, we apply the above result to RMPST.

\subsection{Static Elision of Refinements in RCS}
\label{sec:static_elision:RCS}

In a first step, we develop and prove the correctness of our analysis in RCS.
The question is therefore whether, given a RCS $R$ with one CFSM containing a
transition with refinement $\refinement{r}$, this RCS $R$ is equivalent
(bisimilar) to an RCS where $\refinement{r}$ is replaced with $\taut$.

For the sake of
simplicity, in this subsection, we’ll explain the static elision of refinements
in RCS using examples. Formal definitions, lemmas and their proofs are
available in \inApp{sec:app:static_elision_RCS}. We use the RCS of
$\gtFmt{G_s}$
shown in \cref{fig:se:running_example_RCS}.

\begin{figure}
	\begin{subfigure}{\textwidth}
		\begin{center}
			\begin{tikzpicture}[node distance=4.2cm]
				\draw node[draw, rounded rectangle] (G1)
				{\(\participant{A_1}\)};
				\draw node[draw, rounded rectangle, right=3.5cm of G1] (G2)
				{\(\participant{A_2}\)};
				\draw node[draw, rounded rectangle, right=3.5cm of G2] (G3)
				{\(\participant{A_3}\)};
				\coordinate [above=.5cm of G1] (start);
				\draw[-{Stealth[scale=1.5]}] (start) -- (G1);
				\draw[-{Stealth[scale=1.5]}] (G1) to node[anchor=south, pos=.5,
				sloped] {\(\actionSend{A}{B}{\gtLbl{\ell_1},
				\tuple{\processVariable{x}, \val{\_}}}{x<0}\)} (G2);
				\draw[-{Stealth[scale=1.5]}] (G2) to node[anchor=south, pos=.5,
				sloped] {\(\actionSend{A}{B}{\gtLbl{\ell_2},
				\tuple{\processVariable{y}, \val{\_}}}{x<10}\)} (G3);
			\end{tikzpicture}
			\caption{RCFSM of \(\participant{A}\) in the protocol $\gtFmt{G_s}$
			(\(\ltToAut{\proj{A}{G_s}}\)).}
		\end{center}
	\end{subfigure}
	\begin{subfigure}{\textwidth}
	\begin{center}
		\begin{tikzpicture}[node distance=4.2cm]
			\draw node[draw, rounded rectangle] (G1)
			{\(\participant{B_1}\)};
			\draw node[draw, rounded rectangle, right=3.5cm of G1] (G2)
			{\(\participant{B_2}\)};
			\draw node[draw, rounded rectangle, right=3.5cm of G2] (G3)
			{\(\participant{B_3}\)};
			\coordinate [above=.5cm of G1] (start);
			\draw[-{Stealth[scale=1.5]}] (start) -- (G1);
			\draw[-{Stealth[scale=1.5]}] (G1) to node[anchor=south, pos=.5,
			sloped] {\(\actionReceive{A}{B}{\gtLbl{\ell_1},
					\tuple{\processVariable{x}, \val{\_}}}{\taut}\)} (G2);
			\draw[-{Stealth[scale=1.5]}] (G2) to node[anchor=south, pos=.5,
			sloped] {\(\actionReceive{A}{B}{\gtLbl{\ell_2},
					\tuple{\processVariable{y}, \val{\_}}}{\taut}\)} (G3);
		\end{tikzpicture}
		\caption{RCFSM of \(\participant{B}\) in the protocol $\gtFmt{G_s}$
		(\(\ltToAut{\proj{B}{G_s}}\)).}
	\end{center}
\end{subfigure}
\caption{RCFSM of the RCS of $\gtFmt{G_s}$, the running example of
\cref{sec:static_elision}}
\label{fig:se:running_example_RCS}
\end{figure}

If we aim to i.e.\ transitions which payload modify variables
	that do
	not appear free in the refinement of the considered transition.

\begin{example}[Independent transitions]
	In $\RcsOfType{G_s}$,
	$
	\redact{\participant{A_2}}{\participant{A_3}}
	{\ltsLabelSend{A}{B}{\gtLbl{\ell_2},\tuple{\processVariable{y},\val{\_}}}[x<10]}
	$ depends on the variable $\processVariable{x}\in\fv{x<10}$.
	This transition is self-independent.
	Since the payload of
	$
	\redact{\participant{A_1}}{\participant{A_2}}
	{\ltsLabelSend{A}{B}{\gtLbl{\ell_1},\tuple{\processVariable{x},\val{\_}}}[x<0]}
	$ is $\processVariable{x}$, the former transition depends on the later.
\end{example}

\begin{remark}
	We note $\set{T}_{\processVariable{x}}$ the set of transitions
	$\redact{\sigma}{\sigma^\prime}{\actionGeneric{\_}{\_}{\gtLbl{\_},\tuple{\processVariable{x},\val{\_}}}{\_}}$.
	Given a transition $t$ with refinement $\refinement{r}$, if
	$\processVariable{x}\in\fv{r}$, then $t$ depends on all transitions of
	$\set{T}_{\processVariable{x}}$.
\end{remark}

Essentially, when attempting to
remove a refinement from a target transition $t$, we can disregard all
transitions $t$ is independent of.

The second definition we will need is about transitions being
\emph{well-defined}. So far, nothing prevents us to use refinements
with undefined free variables, we simply consider the refinement
does not hold (c.f.\ \cref{def:semantics_RGA}). In this section, we
specifically focus on
systems where free variables of refinements are in the map when the refinement
is evaluated. When it is the case, we call transitions with such refinements
\emph{well-defined}.

\begin{example}[Well-defined transition]
	Considering the RCS in \cref{fig:se:running_example_RCS}. In the RCFSM of
	$\participant{A}$, the (local) state $\participant{A_2}$ is only accessible
	with a transition
	$
	\redact{\participant{A_1}}{\participant{A_2}}
	{\ltsLabelSend{A}{B}{\gtLbl{\ell_1},\tuple{\processVariable{x},\val{\_}}}[x<0]}
	$. Therefore, any global state $\tuple{\tuple{\participant{A_2},
	\participant{B_{\{1, 2, 3\}}}}, \queue{\_}, \map{M}}$ necessarily contains
	a
	preceding transition
	$
	\redact{\participant{A_1}}{\participant{A_2}}
	{\ltsLabelSend{A}{B}{\gtLbl{\ell_1},\tuple{\processVariable{x},\val{\_}}}[x<0]}
	$. Therefore, $\processVariable{x}$ is always in the map $\map{M}$ of that
	state.

	Therefore, the transition
	$
	\redact{\participant{A_2}}{\participant{A_3}}
	{\ltsLabelSend{A}{B}{\gtLbl{\ell_2},\tuple{\processVariable{y},\val{\_}}}[x<10]}
	$ is well-defined.
\end{example}

We can now conclude our analysis technique: consider a target transition
$t$ with refinement $\refinement{r}$ that is self-independent (it does not
modify the variables of its refinement) and well-define. If all transitions
that modify the free variables of $\refinement{r}$ can guarantee (via their
refinement) that the modification they do is correct with respect to
$\refinement{r}$, then we can safely remove $\refinement{r}$.

\begin{restatable}[Correctness of refinement
elision]{theorem}{StaticElisionCorrect}
	\label{thm:static_elision:correct}
	Given an RCS $R$ containing an RCFSM
	$M = \tuple{Q, C, q_0, \set{A}, \delta}$, and
	$t=\redact{s_\participant{i}}{s_\participant{i}^\prime}
	{\actionGeneric{p}{q}{m}{r}}\in\delta$,
	a well-defined self-independent transition.
	Let $t^\prime = \redact{s_\participant{i}}{s_\participant{i}^\prime}
	{\actionGeneric{p}{q}{m}{\taut}}$;
	$\delta^\prime = \delta\setminus\{t\}\cup t^\prime$;
	$M^\prime = \tuple{Q, C, q_0, \set{A}, \delta^\prime}$;
	and $R^\prime$ be $R$ where $M$ is replaced with $M^\prime$.
	If, for each transition
	$t_w = \redact{\_}{\_}{\actionSend{\_}{\_}{\_}{r_w}}$
	in
	$\bigcup_{\processVariable{x}\in\fv{r}}\set{T}_{\processVariable{x}}$,
	for all map $\map{M}$,
$\map{M}\models \refinement{r_w}$ entails $\map{M}\models\refinement{r}$,
then	there exists a bisimulation 
relating the states of
$R^\prime$ and $R$.
\end{restatable}

\begin{proof}
Proving each direction of the bisimulation is direct (see the proof in
\inApp{sec:app:static_elision_RCS}).
\end{proof}

\begin{example}[Application of \cref{thm:static_elision:correct}]
	The following RCFSM, where $\refinement{x<10}$ is removed, is a valid
	replacement for $\ltToAut{\proj{A}{G_s}}$ in $\RcsOfType{G_s}$.
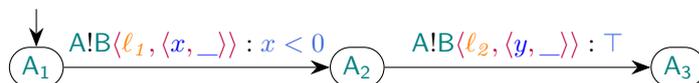
	\begin{center}
	\begin{tikzpicture}[node distance=4.2cm]
		\draw node[draw, rounded rectangle] (G1)
		{\(\participant{A_1}\)};
		\draw node[draw, rounded rectangle, right=3.5cm of G1] (G2)
		{\(\participant{A_2}\)};
		\draw node[draw, rounded rectangle, right=3.5cm of G2] (G3)
		{\(\participant{A_3}\)};
		\coordinate [above=.5cm of G1] (start);
		\draw[-{Stealth[scale=1.5]}] (start) -- (G1);
		\draw[-{Stealth[scale=1.5]}] (G1) to node[anchor=south, pos=.5,
		sloped] {\(\actionSend{A}{B}{\gtLbl{\ell_1},
				\tuple{\processVariable{x}, \val{\_}}}{x<0}\)} (G2);
		\draw[-{Stealth[scale=1.5]}] (G2) to node[anchor=south, pos=.5,
		sloped] {\(\actionSend{A}{B}{\gtLbl{\ell_2},
				\tuple{\processVariable{y}, \val{\_}}}{\taut}\)} (G3);
	\end{tikzpicture}
	\end{center}
\end{example}

\subsection{Application to RMPST Protocols}
\label{sec:static_elision:MPST}

The above subsection explains how to remove some redundant refinements in RCS.
In this subsection, we intend to do the same, focusing on RMPST instead of RCS.

Our goal is the following: we are given an RMPST $\gtFmt{G}$, and we would like
to
remove one of its refinement (which we call the \emph{target}
refinement
$\refinement{r}$).
For the sake of simplicity, in this section, we assume all labels are uniquely
used\iftoggle{full}{ (which we use to prove \cref{thm:happens_before_correct})}{}. For the general
case, we can simply uniquely rename redundant labels. Overall, the roadmap for
this subsection is to show that given the type $\gtFmt{G^\prime}$, which is
$\gtFmt{G}$ where $\refinement{r}$ is \emph{replaced} by $\taut$, $\gtFmt{G}$
and $\gtFmt{G^\prime}$ behave similarly, i.e.\ the RCS the generate are
bisimilar.
To achieve this, we show that \cref{thm:static_elision:correct} applies to
$\RcsOfType{G}$ and $\RcsOfType{G^\prime}$. Therefore, the main point is
finding conditions on RMPST that ensures hypothesis of
\cref{thm:static_elision:correct} holds; we have to verify the following
items:
\begin{enumerate}
	\item all transitions our refinement depends on should entail the refinement
	itself;
	\label{lst:static_elision:mpst:i}
	\item the transition that carries the refinement must be well-defined%
		\iftoggle{full}{ (\cref{def:static:well_defined})}{}. Since variables cannot be
	removed from the map, the first occurrence of the
	target transition must respect the domain condition. Therefore, for this step,
	we can ignore recursion.
	\label{lst:static_elision:mpst:ii}
\end{enumerate}

The main difference with automata is that, in types, we have
\emph{communications}, which possibly contains choices with multiple branches;
and we our goal is to remove the refinement of one of those branches.
Therefore, we first introduce \emph{steps} of a
communication, i.e.\
given a
choice, what are the possible choices it can take. We then extend this to
types. We show that steps in a type correspond to
transitions in the automata of that type.

\begin{example}[Step]
	\label{ex:step_in_type}
%
	The type
	$\gtFmt{G_y} = \gtComm{A}{B}{\ell_2}[y]{int}[x < 10]{\gtEnd}$ has the step
	$\participant{A}\rightarrow\participant{B}\msgM[\tuple{\gtLbl{\ell_2},
	\processVariable{y}}]\models\refinement{x<10}$.
	Since $\gtFmt{G_y}$ occurs in one of the branches of $\gtFmt{G_s}$
	(from the introduction of this section), this step
	\emph{occurs} in $\gtFmt{G_s}$.
\end{example}

Given this notion of steps occurring in a type that is analogous to transitions
in the RCFSM of that type, we can now focus on the conditions of
\cref{thm:static_elision:correct}. Therefore, we have to characterise what
corresponds to \emph{well-defined transitions} in a type. Since transitions (in
RCS) and steps (in types) are analogous, we introduce
\emph{well-define
steps} in a type. We recall that, in a RCS, a
transition is well-defined if the free variables of the refinement it carries
are always known when the transition is fired.
Since variables are never
removed from the map, we can focus on the first occurrence of the transition.
So far, we do not have a notion of \emph{run} for a type. Therefore, we
first
define an \emph{happens-before} relation in RMPST, and we use this
relation to define
\emph{well-defined steps} as steps that contain a refinement which free
variable are all exchanged in a communication that \emph{happens-before} the
step we consider. With those two definitions, we can finally prove that a
well-defined step in a type
corresponds to a well-defined transition in the corresponding RCS.

\begin{example}[Well-define step in a type]
	Consider $\gtFmt{G_s}$ and $\gtFmt{G_y}$ as in \cref{ex:step_in_type}.
	The step
	$\participant{A}\rightarrow\participant{B}\msgM[\tuple{\gtLbl{\ell_2},
	\processVariable{y}}]\models\refinement{x<10}$ is well-defined.
	Indeed, $\fv{x<10} = \{\processVariable{x}\}$, $\gtFmt{G_s} < \gtFmt{G_y}$,
	and $\gtFmt{G_s}$ contains a branch that sends $\processVariable{x}$ and
	which continuation contains $\gtFmt{G_y}$.
\end{example}

We can finally proceed to the overall goal of this section: showing that the
type with and without the target refinement behave similarly. Thanks to the
above lemmata, we simply have to target a refinement with the appropriate
conditions and apply \cref{thm:static_elision:correct}.

\begin{restatable}[Static elision of redundant refinements in
types]{theorem}{StaticElisionRMPST}
	\label{thm:static_elision_types}
	Given two a global types $\gtFmt{G}$ and
	$\gtFmt{G_s} = \gtComm[i][I]{p}{q}{\ell}[x]{S}[r]{G} \in\gtFmt{G}$, such
	that, for one
	$t\in I$,
	$\participant{p}\rightarrow\participant{q}\msgM[\tuple{\gtLbl{\ell_t},
	\processVariable[t]{x}}]\models\refinement{r_t}$ is a well-defined step with
	$\processVariable[t]{x}\not\in\fv{r_t}$.
Let $\gtLbl{\ell_{t^\prime}}=\gtLbl{\ell_t}$,
	$\processVariable[t^\prime]{x}=\processVariable[t]{x}$,
	$\processType[t^\prime]{S}=\processType[t]{S}$,
	$\refinement{r_t^\prime}=\taut$,
	$\gtFmt{G_{t^\prime}} = \gtFmt{G_t}$,
	$\gtFmt{G_{s^\prime}} =
	\gtComm[i][I\setminus\{t\}\cup\{t^\prime\}]{p}{q}{\ell}[x]{S}[r]{G}$;
and $\gtFmt{G^\prime}$ be $\gtFmt{G}$ where $\gtFmt{G_s}$ is replaced with
	$\gtFmt{G_{s^\prime}}$.
	If, for all steps,
	$\participant{r}\rightarrow\participant{s}\msgM[\tuple{\gtLbl{\_},
		\processVariable[w]{x}}]\models\refinement{r_w}$
	occurring in
	$\gtFmt{G}$ (for each $\processVariable{x}\in\fv{r}$),
	$\map{M}\models\refinement{r_w}$ entails $\map{M}\models\refinement{r}$ (for
	all $\map{M}$),
	there exists a bisimulation between the states of $\RcsOfType{G}$ and those
	of $\RcsOfType{G^\prime}$.
\end{restatable}
\begin{proof}
	We prove this by showing that \cref{thm:static_elision:correct} applies to
	$\RcsOfType{G}$ and $\RcsOfType{G^\prime}$.
	The proof is provided \inApp{sec:app:elision_rmpst}.
\end{proof}

\begin{example}[Application of \cref{thm:static_elision_types}]
	Given $\gtFmt{G_s}$ as in \cref{ex:step_in_type} and $\gtFmt{G_s^\prime}$
	as follows (notice the second refinement is replaced by $\taut$),
$\gtFmt{G_s}$ and the following $\gtFmt{G_s^\prime}$ have the same behaviour:\\[1mm]
\centerline{
$\gtFmt{G_s^\prime} = \gtComm{A}{B}{\ell_1}[x]{int}[x < 0]{
		\gtComm{A}{B}{\ell_2}[y]{int}[\taut]{\gtEnd}}$
}
\end{example}

\section{Implementation}
\label{sec:implem}
In the previous section, we introduced an instance of our
framework: a system that accommodates refinements using a decentralised
verification mechanism. In this section, we follow up on this example with an
implementation, based on Rumpsteak, of this system.

Rumpsteak \cite{cutner_deadlock-free_2022} is a framework to write Rust programs
according to an MPST specification. The framework is divided into two parts:
\begin{inlineenum}
	\item a runtime library that provides primitives to write asynchronous
	programs
	in Rust
	\item a tool (\texttt{rumpsteak-generate}) to generate skeleton Rust files
	from
	specification files (i.e. from global types), in two steps.
\end{inlineenum}

Working with Rumpsteak takes two manual steps.
The user specifies (step 1)
the
protocol in a global type(written as Scribble files \cite{Scribble}, see
\cref{fig:nuscr:plus_minus}). This global type is automatically projected using
\(\nu\)Scr \cite{ZhouNuScr} and the projected types are used to generate
skeleton Rust
files (see \cref{fig:rumpsteak:plus_minus}). The generated Rust code
contains Rust types that encode local types (e.g. the type for \participant{A}
is shown in \cref{fig:rumpsteak:plus_minus:type_A} in
\cref{fig:rumpsteak:plus_minus}). The user then manually implements (step 2)
the process of each participant, following their type
(\cref{fig:rumpsteak:plus_minus:process_A}), using provided communication
primitives (\cref{fig:rumpsteak:send}). Rumpsteak relies on Rust's typechecker
to ensure the consistency of the implementation.
For the sake of clarification where needed, we call \emph{Vanilla Rumpsteak} the
framework without refinements (i.e. as presented
in~\cite{cutner_deadlock-free_2022}),
and \emph{Refined Rumpsteak} the framework modified to accommodate refinements.

\begin{figure}
	\begin{subfigure}[t]{0.52\textwidth}
\begin{lstlisting}[language=Scribble, style=boxedScribble]
(*# RefinementTypes #*)

global protocol PlusMinus
(role A, role B, role C)
{
Secret(n: int) from A to B;
rec Loop {
	Guess(x: int) from C to B;
	choice at B {
		More(x: int {x < n}) from B to C;
		continue Loop;
	} or {
		Less(x: int {x > n}) from B to C;
		continue Loop;
	} or {
		Correct(x: int {x = n}) from B to C;
	}}}
\end{lstlisting}
		\caption{$\nu$Scr description of the guessing game protocol.}
		\label{fig:nuscr:plus_minus}
	\end{subfigure}
	\hfill
	\begin{subfigure}[t]{0.45\textwidth}
\begin{lstlisting}[language=Rust, style=boxed, escapechar=@]
type PlusMinusA = @\label{fig:rumpsteak:plus_minus:type_A}@
Send<B, 'n', @\label{fig:rumpsteak:plus_minus:type_A_var_name}@
Secret,
Tautology::<Name, Value, Secret>, @\label{fig:rumpsteak:plus_minus:type_A_ref}@
Constant<Name, Value>, End>;@\label{fig:rumpsteak:plus_minus:type_A_end}@
// ...
async fn a(role: &mut A)@\label{fig:rumpsteak:plus_minus:process_A}@
-> Result<(), Box<dyn Error>> {
	try_session(role,
	HashMap::new(),
	|s: PlusMinusA<'_, _>| async {
		let s =
		s.send(Secret(10)).await?;@\label{fig:rumpsteak:send}@
		return Ok(((), s))
	})
	.await
}@\label{fig:rumpsteak:plus_minus:process_A_end}@\end{lstlisting}
		\caption{Rust type and implementation of participant \participant{A} of
		the guessing game protocol. The handwritten code
		(\cref{fig:rumpsteak:plus_minus:process_A} to
		\cref{fig:rumpsteak:plus_minus:process_A_end})
		 is the same than with Vanilla Rumpsteak.}
		\label{fig:rumpsteak:plus_minus}
	\end{subfigure}
	\caption{Implementation of the guessing game using Rumpsteak.}
\end{figure}

In this section, we explain the main differences between Vanilla and Refined
Rumpsteak: we introduce refinements in the types used in the runtime library,
we modify the program generation step accordingly, and we introduce tools that
ensure the localisation conditions are met (\cref{def:cond_dyn_verif} in
\cref{sec:dynamic_CS}).
The overall workflow is presented in \cref{fig:rumpsteak_workflow}.
We conclude this section by measuring the overhead induced by the refinement w.r.t.
Vanilla Rumpsteak and the time needed for asserting the localisation conditions.

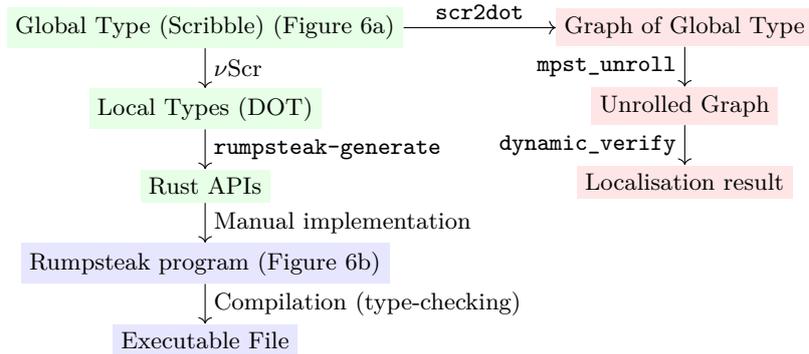
\begin{figure}
	\centering
	{{\small
\begin{tikzpicture}[
		unmodified/.style={fill=blue!10},
		new/.style={fill=red!10},
		modified/.style={fill=green!10},
		node distance=15pt,
		]

	\node[modified, align=center] (start) {Global Type (Scribble)
	(\cref{fig:nuscr:plus_minus})};
	\node[new, right=2cm of start] (step 0) {Graph of Global Type};
	\node[new, below=of step 0] (step 1) {Unrolled Graph};
	\node[new, below=of step 1] (loc_res) {Localisation result};
	\node[modified, below=of start] (step 2) {Local Types (DOT)};
	\node[modified, below=of step 2, align=center] (step 3) {Rust APIs};
	\node[unmodified, below=of step 3, align=center] (step 4) {Rumpsteak
	program
	(\cref{fig:rumpsteak:plus_minus})};
	\node[unmodified, below=of step 4] (step 5) {Executable File};

	\draw[->] (start)  -- node[above] {\texttt{scr2dot}} (step 0);
	\draw[->] (step 0)  -- node[left] {\texttt{mpst\_unroll}} (step 1);
	\draw[->] (step 1)  -- node[left] {\texttt{dynamic\_verify}} (loc_res);
	\draw[->] (start)  -- node[right] {\(\nu\)Scr } (step 2);
	\draw[->] (step 2) -- node[right] {\texttt{rumpsteak-generate}}  (step 3);
	\draw[->] (step 3) -- node[right] {Manual implementation} (step 4);
	\draw[->] (step 4) -- node[right] {Compilation (type-checking)} (step 5);
\end{tikzpicture}
}}
	\caption{Workflow of Rumpsteak. Green nodes represent steps that
	already existed in Vanilla Rumpsteak and that have been adapted to
	accommodate for refinements, red nodes represent new steps, and blue nodes
	represent unmodified steps. The three new steps (\texttt{scr2dot},
	\texttt{mpst\_unroll}, and \texttt{dynamic\_verify}) verify the conditions
	mentioned in \cref{def:cond_dyn_verif}.}
	\label{fig:rumpsteak_workflow}
\end{figure}

\subsection{Refinement Implementation}
\label{sec:implem:rumpsteak}

\myparagraph{Modifications to the Rumpsteak Library.}
In order to accommodate for refinements, we have to introduce new elements in
to the Rumpsteak's encoding of local types. Consider the local type of
participant \participant{A} introduced in \cref{ex:plus_minus_proj}
\(\ltIntC{\participant{B}}{secret}{n}{int}[\taut]{\ltEnd}\):
Rumpsteak now has to take into account the name of the variable sent
($\processVariable{n}$), and the refinement attached to the transition
($\taut$). Consider the type declaration in
\cref{fig:rumpsteak:plus_minus:type_A} to
\cref{fig:rumpsteak:plus_minus:type_A_end}, \cref{fig:rumpsteak:plus_minus}.
Compared to Vanilla Rumpsteak, we introduce \rustinline{'n'}, a const
generic\footnote{\url{https://github.com/rust-lang/rfcs/blob/master/text/2000-const-generics.md}},
that carries the name of the variable sent
(\cref{fig:rumpsteak:plus_minus:type_A_ref}). Regarding the refinement, we
introduce \rustinline{Tautology::<Name, Value, Secret>}, which represent the
refinement $\taut$. The generic parameters are used to specify the type of
variable names (\rustinline{char}s in our case) and values (\rustinline{i32})
as well as the label of the message (\rustinline{Secret}). We modified $\nu$Scr
and \texttt{rumpsteak-generate} to generate skeleton files (the content of the
file up to \cref{fig:rumpsteak:plus_minus:type_A_end}).
Rumpsteak provides a set of available refinements, and additional ones can be
written ad-hoc (for specific needs). To add an ad-hoc refinement, the user
simply implements the trait \rustinline{Predicate}
(which extends \rustinline{Default}), which requires a method
\rustinline{check} that asserts whether the predicate holds. For instance, the
\rustinline{check} function of \rustinline{Tautology} simply returns
\rustinline{true}.

\myparagraph{Verification of the Conditions for Decentralised Refinement
Assertion.}
As we explained in \cref{sec:dynamic_CS}, to make
sure that
refinements can be verified in a decentralised way, we require to check that
variables needed for the refinements are located correctly
(\cref{def:cond_dyn_verif}). To perform this
verification, we implemented new tools for the Rumpsteak framework (in red in
\cref{fig:rumpsteak_workflow}).

Our tools:
\begin{inlineenum}
	\item convert the global type into a graph (\texttt{scr2dot})
	\item unroll the loops once to precisely capture variables initialisations
	(\texttt{unroll\_mpst})
	\item localise variables on the unrolled graph (\texttt{dynamic\_verify}).
\end{inlineenum}

The core part of this verification, \texttt{dynamic\_verify}, finds variables
locations
with simple inference rules written in Datalog. We use the \emph{crepe}
library~\cite{Crepe} which provides a Datalog DSL for Rust.
We provide more details on the algorithm in \inApp{sec:static_verif_conditions}.

\myparagraph{Limitations.\label{sec:par:implem_limitations}}
The current implementation makes extensive use of the Rust feature \emph{const
generics}\footnotemark[9]
which introduces a limited form of dependent types in Rust. It allows to use
constant values in types. As of today, only some \emph{basic types} can be used
as const generics, in particular \lstinline[language=Rust,
style=colouredRust]{char}s and the various integer types. We use such const
generics to encode informations about the variables into the types: for
instance, the predicate \refinement{x < 5} would have the type
\lstinline[language=Rust, style=colouredRust]{LTnConst<L, 'x', 5>}, where the
\lstinline[language=Rust, style=colouredRust]{'x'} and the
\lstinline[language=Rust, style=colouredRust]{5} are const generics.

For readability, we choose to set variables to
\lstinline[language=Rust, style=colouredRust]{char}s, meaning that in the
current implementation, we can only accommodate a limited number of distinct
variables. Should more be needed, one could easily modify our
implementation to replace them with \lstinline[language=Rust,
style=colouredRust]{u64}, which allows \(2^{64}\) variables names.
Similarly, we only consider \lstinline[language=Rust, style=colouredRust]{i32}
as
message payloads. Should different types of messages be needed, they could be
encoded in an \lstinline[language=Rust, style=colouredRust]{enum}.

Finally, the static elision optimisation
(\cref{sec:static_elision}) is not implemented.

\subsection{Runtime and Localisation Benchmarks}

We evaluate how Rumpsteak with refinements performs
with respect to Rumpsteak without refinements. First, we measure the runtime of
our analysis
tool which verifies the two conditions in \cref{def:cond_dyn_verif}
(\texttt{scr2dot}, \texttt{unroll\_mpst} and \texttt{dynamic\_verify}).
Although not a runtime cost, and while we expect this analysis to be possibly
expensive, we would like to ensure that it is still
practical for test cases from the literature.  Secondly, we evaluate the
runtime overhead of adding
refinements with respect to Rumpsteak without refinements.

\myparagraph{Setup and Benchmark Programs.}
We evaluate the performance of Rumpsteak with refinement with benchmarks. Most
of them are taken from the
literature
(\Cref{tab:benchmarks}).
This set of program contains various micro-benchmarks
with a variety of combination of properties (whether the protocol is binary or
multiparty, contains recursivity or choice).
\begin{wraptable}{o}{0.51\textwidth}
	\vspace{-\baselineskip}
	\centering
	\begin{tabular}{ll|ccc}
		Name 		& MP	& Rec	& Choice \\
		\hline
		\ding{172} simple adder \cite{hu_hybrid_2016}	& no	& no	&
		no \\
		\ding{173} travel\_agency \cite{hu_session-based_2008} 	& no    &
		no    & yes \\
		\ding{174} ping pong \cite{zhou_statically_2020}	& no	& yes
		& no \\
		\ding{175} simple auth.	& no    & yes   & yes \\
		\ding{176} ring max	& yes	& no	& no \\
		\ding{177} three\_buyers \cite{honda_multiparty_2016}  	& yes   &
		no    & yes \\
		\ding{178} plus or minus	& yes	& yes	& yes \\
	\end{tabular}
	\caption{The set of micro benchmarks together with their
		characteristics. "MP" denotes a multiparty protocol, "Rec" the
		presence of recursion, and "Choice" the presence of choice.}
	\label{tab:benchmarks}
	\vspace{-1cm}
\end{wraptable}
Notice that protocols that contain
recursivity with no choice (e.g. \emph{simple auth} are infinite). Therefore, such
protocols are only measured in the variable localisation paragraph. Also, where
it applies, protocols were modified in order to add
relevant refinements; such modifications are listed below. By default, we add
\rustinline{Tautology} predicates (\cref{sec:implem:rumpsteak}).
The tests were performed on a machine running Ubuntu 22.04.1 LTS x86\_64 (kernel
5.15.0-60) with an Intel i7-6700 processor (4 cores, 8 threads running at 4GHz
maximum) and 16GB of memory\footnote{The micro-benchmarks are not memory
intensive. The memory size is not a limiting factor. However, the
benchmarks seem to be dominated by the startup time, which includes memory
access time.}. We compare Rumpsteak with refinement vs. Vanilla Rumpsteak. For
a comparison between Vanilla Rumpsteak and other libraries, see
\cite[Figure~6]{cutner_deadlock-free_2022}.

\myparagraph{Added Refinements \& Protocol Modifications.}
Some benchmarks from the literature were adapted in order to accommodate
refinements.  In addition, we introduce three benchmarks. Those benchmarks are
close to examples from the literature, adapted to better highlight
refinements.
\begin{description}
	\item [simple adder:]This example is adjusted from the
		\emph{Adder} (\citet{hu_hybrid_2016}) protocol, but we remove the
		choice of operation in order to increase the benchmark
		diversity;
	\item [ping pong:]In~\citet{zhou_statically_2020}, some of the loops
		were statically unrolled, and the protocol contained a choice
		to exit. Ours is equivalent to an infinite
		\emph{PingPong\(_1\)} in \citet{zhou_statically_2020}.
	\item [simple authentication:]This example is a binary example of an
		authentication protocol (e.g. OAuth \cite{scalas_less_2019}).
		The added refinements enforce that access is granted if and only
		if the given password is correct.
	\item [ring max:]A multiparty protocol where participants receive a
		value from their predecessor (except for the initial
		participant), and forward an other value to their successor (the
		final participant forwards it to the initial one). Refinements
		ensure that the value forwarded is greater than or equal to the
		value received.
	\item [plus or minus:]An implementation of our running example.
\end{description}

\begin{wraptable}{o}{0.4\textwidth}
	\vspace{-0.5\baselineskip}
	\centering
	\begin{tabular}{l|ccc|c}
		& \(|S|\) & \(|U|\) & \(|V|\) & \emph{et} (\(\mu\pm\sigma\))\\
		\hline
		\ding{172} & 4 & 4 & 3 & \(5.5\pm0.2\)\\
		\ding{173} & 7 & 7 & 6 & \(5.5\pm 0.2\)\\
		\ding{174} & 2 & 4 & 1 & \(5.5\pm0.2\)\\
		\ding{175} & 6 & 11 & 3 & \(5.6 \pm 0.2\)\\
		\ding{176} & 8 & 8 & 7 & \(5.7\pm0.2\)\\
		\ding{177} & 10 & 10 & 7 & \(5.6\pm 0.2\)\\
		\ding{178} & 4 & 19 & 2 & \(5.6\pm0.2\)\\
	\end{tabular}
	\caption{Benchmark of the localisation analysis (Red branch in
		\cref{fig:rumpsteak_workflow}). \(|S|\) denotes the number of states of
		the graph of the protocol; \(|U|\) denotes the number of states after
		unrolling the recursion loops once; and \(|V|\) denotes the number of
		variables in the protocol. $|S|$, $|U|$ and $|V|$ are computed manually
		to
		give an insight on how protocols compare. \emph{et} is the execution
		time,
		measured by the benchmark (in \si{\milli\second}).}
	\label{tab:benchmark_localisation}
	\vspace{-0.5\baselineskip}
	\begin{tabular}{l|ccc}
		& \(p\) & \(m\) & \(m_r\) \\\hline
		\ding{172} & 0.00 & 0.7 & 0.8\\
		\ding{173} & 0.11 & 0.8 & {\color{lightgray} N/A}\\
		\ding{175} & 0.29 & 0.8 & {\color{lightgray} N/A}\\
		\ding{176} & 0.17 & 0.8 & {\color{lightgray} N/A}\\
		\ding{177} & 0.68 & 0.7 & {\color{lightgray} N/A}\\
		\ding{178} & 0.04 & 0.7 & 0.8\\
	\end{tabular}
	\caption{Evaluation of the runtime overhead due to the addition of
		refinements in Rumpsteak. \(p\) is the MWU \(p\)-value, \(m\) is the
		baseline median runtime and
		\(m_r\) is the median runtime with refinements when applicable (\(p <
		0.05\)). All times are in~\si{\milli\second}.}
	\label{tab:runtime_benchmarks}
	\vspace{-1cm}
\end{wraptable}
\myparagraph{Static Analysis of Variable Locations.\label{sec:benchmark_dynamic_verif}}
\Cref{tab:benchmark_localisation} shows the decentralised verification time cost for
each refined global label. As shown in~\cref{fig:rumpsteak_workflow}, this
static analysis is performed with three tools. The results shown account for
the whole pipeline, and were measured over \(50\) samples, with \(10\) warmup
runs (excluded from the measurements).
Overall, the runtime for variable localisation is stable (around 5.6ms). We
suspect that, for graphs with a low number of states, the runtime is dominated
by the accesses to the file.

\myparagraph{Runtime Overhead of Refinement
Feature.\label{sec:benchmark_runtime}}
Our second set of benchmarks aims to measure the overhead of runtime refinement
verification with respect to the original Rumpsteak framework.
We are expecting Rumpsteak with refinements to be slower than the original
Rumpsteak, due to the additional cost of evaluating refinements. This benchmark
has two objectives: first, to find out whether there is an actual,
statistically significant, overhead; and second, if so, estimate this overhead.
To measure this overhead, we only consider the protocols
that
terminate from the benchmark set.

To fulfil the first objective, we use a Mann-Whitney U test (MWU). We used MWU
as it is a non-parametric test, and our runtime distributions do not follow a
normal distribution, which prevents us to do simpler analysis. As MWU is
sensitive to the number of samples, we run each benchmark \(30\) times, on both
the original Rumpsteak and Rumpsteak with refinements. We perform the MWU test
on the collected \(30\) samples, preceded by \(10\) iterations to warm the
system up. Our
hypothesis for the MWU test are the following:
\begin{quote}
	\(H_0\): \ The distributions of runtimes with and without refinements are
identical.\\
 \(H_1\): \ The distributions of runtimes with and without refinements are
distincts.
\end{quote}

The \(p\)-values obtained from the MWU test are reported in the first column
of~\cref{tab:runtime_benchmarks}. We also report the baseline (Rumpsteak without
refinements) median run time (over the \(30\) runs) in the second column of the
table. Most often, the overhead is not significant (\(p \geq 0.05\)) and
$H_0$ can not be rejected. When the overhead is
statistically significant, we also report the median runtime (over the \(30\)
runs) of Rumpsteak with refinements in the third column.
With our set of microbenchmarks, in most cases we cannot distinguish Rumpsteak
with refinement from Rumpsteak without refinements. We suspect Rumpsteak runtime
is dominated by communications and context switching. However, as our
refinements can be arbitrarily complex, specific instances could show real
slowdown due to refinement evaluation.

\section{Related Work and Conclusion}
\label{sec:related}
\myparagraph{Design-by-Contract for (Multiparty) Session Types.}
In binary session types, \citet{villard-phd2011} introduces contracts for
binary sessions, and provides an analysis tool which verifies whether a given
program comply with its associated contract. The verification is done with
symbolic execution. Compared to this paper, we address multiparty sessions.
Besides, our framework is more generic (specific instances could be based on
symbolic execution, but we can also accommodate other verification methods).
Bocchi et al. \citet{bocchi_theory_2010} present a variant of MPST
that allows predicates on exchanges, that must hold for a typed process to take
transitions. The main difference with our work is that their approach focuses on
\emph{correctness by construction}, i.e. they accept only correct protocols, while we
can accept protocols that fail, and we simply prevent them to generate incorrect traces.
More precisely, the authors statically ensure that there is a
satisfiable path, which prevents some valid runs to be accepted.
For instance, consider the following
type:\\[1mm]
\centerline{$
	\gtComm{A}{B}{\ell_1}[x]{int}[x < 10]{
		\gtComm{B}{A}{\ell_2}[y]{int}[x > y \wedge y > 6]{
			\gtEnd
		}
	}
$}\\[1mm]
This type would be rejected in \cite{bocchi_theory_2010} since if
\(\participant{A}\) sends \(\processVariable{x} = \val{5}\) (which is allowed by
\(\refinement{x < 10}\)), then there is no \(\processVariable{y}\) that
satisfies \(\refinement{5 > y\wedge y > 6}\).  By rejecting this, they also
reject all possibly valid runs (e.g.\ if \(\participant{A}\) sends
\(\processVariable{x} = \val{9}\) and \(\participant{B}\) replies with
\(\processVariable{y} = \val{7}\)).
A follow-up on this work is~\cite{bocchi_multiparty_2012} which introduces
local states, i.e.\ the authors allow participants to have local variables, which can
be updated during process execution. The session types
reflect those elements and contain predicates on exchanged variables and local
variables.

With respect to these two papers, our criteria for the validity of refinements
(expressed as a property of the generated trace) is decoupled from the
semantics of the model. This approach allows us to be more flexible than
enforcing statically the refinements, and to lower the cost of adopting
refinements, in particular to retrofit refinements into existing systems. For
instance, using our framework, one can simply use the centralised semantics at
first, which is very expressive, without having to prove the correctness of the
implementation. In a second step, users can then develop different verification
or analysis techniques which can be plugged-in transparently. For instance,
switching from Vanilla Rumpsteak to Refined Rumpsteak does not involve changes
in the implementation, as the modifications do not happen in the programming
interface.
Also, compared to these papers, our framework is not bound to MPST only, and
provide an actual implementation of our framework.

\myparagraph{Design-by-Contract in Choreography Automata.}
Choreography Automata (CA) are graphs that represent the global behaviour of a
concurrent system. The behaviour of individual participants is obtained by
\emph{projecting} well-formed CA, i.e. erasing all actions that do not concern
a given participant. The result is a FSM which, after determinising and
minimising, is used as a CFSM. The projection of all participants leads to a
CS. Notice that CA accept some protocols that would be rejected by MPST, and
vice-versa.

Gheri et al.\ \citet{Gheri_Design_2022} study the verification of
CA with assertions. Their work and ours are distinct with respect to the
following aspects:
\begin{inlineenum}
\item the communication semantics\label{lbl:BdC_choreography:item1}
\item the choices\label{lbl:BdC_choreography:item2}
\item the logic for predicates\label{lbl:BdC_choreography:item3}
\item the implementation presented in \cite{Gheri_Design_2022} is limited to CA
without assertions (i.e., the design-by-contract approach was not implemented
and left as their future work).
\end{inlineenum}

Regarding \cref{lbl:BdC_choreography:item1}, Gheri et al.\
\citet{Gheri_Design_2022} defines choreography automata with \emph{synchronous}
communication semantics, while the one we developed in this work is
asynchronous. Gheri et al.\ \cite[Section~7]{Gheri_Design_2022} discusses
asynchronous semantics but it remains future works.

Regarding \cref{lbl:BdC_choreography:item2}, we are constrained by the syntax
of RMPST, in which choices can only happen between two selected participants,
while choreography automata accept protocols with choices where a (single)
participant \participant{A} sends to multiple receivers (\participant{B} and
\participant{C}) \cite[Definition~4.15]{Gheri_Design_2022}. Explicit
connections \cite{hu_explicit_2017} is an extension of MPST that accommodates
with choices with multiple receivers.

Regarding \cref{lbl:BdC_choreography:item3}, we kept our refinement logic
abstract, while it is fixed in choreography automata, with a form of first
order logic. Besides, predicates are handled differently in both frameworks as
well: Gheri et al.\ \citet{Gheri_Design_2022} require choreography automata to
be \emph{history-sensitive} \cite{bocchi_theory_2010}, a definition which
serves a similar purpose to our definition of \emph{variable localisation}
(\cref{sec:dynamic_CS} and \inApp{sec:localisation_algo}), which constrains our
decentralised semantics. Our centralised semantics (\cref{def:RConf}) is not
constrained by variable localisation. For instance, the RMPST
$\gtComm{A}{B}{\ell_1}[x]{int}[\taut]{
		\gtComm{C}{D}{\ell_2}[y]{int}[x=y]{
			\gtEnd
		}
	}$
produces valid traces with our centralised semantics, while the corresponding
choreography automata would be rejected.

Besides, our work introduces a general framework that can accommodate
	refined CA in addition to RMPST. We show
	\inApp{sec:refined_choreo_automata} a
	possible way to do so.

\myparagraph{Implementations of Refinements in MPST.}
Neykova et al. \Citet{neykova_session_2018} develop an F\(\#\)
library for
static verification of MPST with refinements.
They present a compiler
plugin which uses an SMT solver (Z3) to statically verify some refinements.
They use a notion of similar to our variable localisation criterion (which they call
\emph{variable knowledge}), and a variant of CFSM with refinements that is
similar to ours. In their work, refinements that are statically asserted by the SMT
solver are pruned in the CFSM, while the rest of refinements are kept in the
CFSM and are dynamically checked.
Similarly, \citet{OOPSLA20SessionStar, zhou_statically_2020} develop a framework for
multiparty session types with refinements in F\(^\star\). They delegate the
management of refinements to F\(^\star\) type system (which internally uses an
SMT solver).
They define refinements on global types, which are then projected onto local
types. They show that a global type and its projection are trace equivalent.
Those two works focus on the \emph{implementation} of MPST with
refinements. \citet{neykova_session_2018} does not focus on the theory of
refinements and the theory developed in \cite{zhou_statically_2020} is tightly
coupled to F\(^\star\). For instance, they do not present a \emph{correctness}
criterion such as \emph{valid refined traces} we present.
Contrary to both works,
our correctness criteria (based on valid refined traces) is
\emph{decoupled from}
(i.e. independent of) any target type theory, programming language or model of
computation: we only require an LTS labelled with actions. Besides, the logic
used for refinements is also a parameter of our framework, and users could use
alternatives, leading to a greater expressivity of our framework.

The main syntactical difference between our RMPST and those
developed in \cite{zhou_statically_2020} is that we attach refinements to the
messages of the protocol, while \citet{zhou_statically_2020} attach refinements
to the payload value. This is due to a different approach:
correctness in \cite{zhou_statically_2020} is related to payload types being
inhabited while our criteria of correctness (developed in
\cref{def:valid_refined_trace}) relies
on actions being allowed.
In binary linear logic-based session types,
\citet{DasSession2020} study the metatheory of binary session types
with arithmetic refinements. In particular, they focus on the type equality,
showing that added refinements make the type equality
undecidable (they provide a sound but incomplete algorithm for type equality).
\Citet{DasRast2022} also implement a library for session types with
refinements, although it only accounts for arithmetic
refinements.

\myparagraph{Other Related Works.}
There are various papers on the dynamic verification of MPST. For
instance \citet{TCS17MPSTMonitor} present a framework that
allows for both static and dynamic verification of MPST. This paper introduces
a theory for (dynamically) monitoring assertions on messages (i.e. the
equivalent of our refinements). Furthermore, the authors introduce theoretical
tools (bisimulations) to relate monitored processes with correct unmonitored
processes. This paper, however, suffers a few limitations. First, it focuses on
\emph{monitorable} types (which intuitively correspond to types satisfying our
conditions for decentralised verification \cref{def:cond_dyn_verif}). Second,
it focuses on dynamic verification of assertions. The paper is compatible with
statically verified processes (which allows turning off the dynamic
monitoring), but it does not present techniques for static verification in
itself.

On the other hand, our paper takes a different approach, by decoupling the
correctness criterion from the verification technique. This allows us to have a
more general framework (our framework accept types that are not
localisable/monitorable, although not all semantics can accommodate those), as
well as to develop static verification techniques.

In Rust, the \texttt{refinement} crate \cite{Dean_refinements_2021} provides
refinement data types. Their approach of refinements is similar to ours,
with a
\rustinline{Predicate} trait that provides a method to perform the predicate
verification (at runtime).
Refinement data types have also been implemented in multiple languages
(e.g.
F\(^\star\), Haskell \cite{Vazou_Liquid_2016}, etc.).
On the practical side, we can note the
similarities between typestates and session types \cite{hu_type-safe_2010}.
\citet{duarte_retrofitting_2021} implements typestates in Rust with a DSL to
verify protocol conformance. While Rumpsteak does not use their library, it
internally uses similar constructs.

Regarding implementations of session types in Rust, there are several
frameworks beside Rumpsteak. \citet{JespersenSession2015} first
integrate binary session types in Rust, but their implementation suffers a few
drawbacks (see \cite[Section~3]{Kokke2019} for a detailed explanation).
Sesh \cite{Kokke2019} and Ferrite
\cite{Chen2021} are two Rust libraries for \emph{binary} session
types, and they implement synchronous and asynchronous ones,
respectively.
MultiCrusty
\cite{lagaillardie_stay_2022} implements synchronous MPST on top
of Sesh, with a mesh of binary sessions. Compared to MultiCrusty,
Rumpsteak implements directly MPST instead of wrapping them into binary
sessions, and focuses on asynchronous MPST.
None of the aforementioned tools develops refinements.
It would be an interesting future work to apply our criteria
to extend their tools with refinements.

Finally, we note the proximity between (MP)ST with refinements and dependent
(MP)ST. For instance, \citet{ThiemannLabeldependent2020} introduce a session
type calculus with label-dependency (their approach does not explicitly account
for payload value refinement).
Other approaches exist, for instance, an intuitionistic linear logic-based
type theory for building value-dependent session
types \cite{ToninhoDependent2011}, and
separation logic-based work for reasoning about session types \cite{HinrichsenActris2020}.

\label{sec:conclusion}




\myparagraph{Future Work}
While, in our work, we consider MPST with payloads (some variants only consider
messages with labels), we restrict our MPST with a single payload (i.e.
\emph{monadic} MPST, where each message carries a single value).
The extension to polyadic MPST, where a message can carry multiple values, is
straightforward, by adapting the RCS rules (\grsnd{} and \grrec{},
\cref{def:semantics_RGA}). 

We presented two optimisations, in order to illustrate the
flexibility of our theoretical framework.
Regarding the decentralised
verification (\cref{sec:dynamic_CS}), there is room for
an extension, e.g.\ with specific domains (i.e. some class of protocols with
specific refinements). Regarding the static elision of redundant refinements,
we envision improving the technique with use of SMT solvers could be promising.
The main difficulty lies in asynchronous communications: one would need to
consider all possible message orderings before solving constraints.

	\bibliography{references}

\iftoggle{full}{
	\newpage
	\appendix
	\newcommand{\separator}{
\begin{center}
\rule{5cm}{.5mm}
\end{center}
}

\section{Extra examples illustrating interesting aspects of RMPST.}

\subsection{Usefulness of multiparty}
Our guessing game running example is somewhat simple and could be implemented
using two successive binary sessions. Here, we present a slightly more complex
example that uses RMPST moore intensively. This example is based on the
guessing game, with an extra communication (a final official guess)
$\processVariable{y}$ from
\participant{C} to \participant{A} at the end\footnote{Based on an idea from
ECOOP reviewer C.}.

$$
\gtComm{A}{B}{secret}[n]{int}{
	\gtRec{T}{
		\gtComm{C}{B}{guess}[x]{int}{
			\gtCommRaw{B}{C}{
				\begin{aligned}
					&\commChoice{more}{}[x < n]{\gtRecVar{T}},\\
					&\commChoice{less}{}[x > n]{\gtRecVar{T}},\\
					&\commChoice{correct}{}[x = n]{
						G_{\text{cont}}
					}
				\end{aligned}
			}
		}
	}
}
$$
with
$$
\gtFmt{G_\text{cont}} = \gtComm{C}{A}{validate}[y]{int}[y=x]{\gtEnd}
$$

\subsection{List Adder}
This example is more involved: \participant{A} sends a list of numbers to
\participant{B}, and eventually \participant{B} returns the sum of all numbers.
We want to have refinements to make sure the returned value is correct. To
achieve that, upon each number received, \participant{B} computes the partial
sum. This example requires extra participants in order to update and store the
partial sum.

$$
	\gtRec{T}{
		\gtCommRaw{A}{B}{
		\begin{aligned}
			&\commChoice{add}[n]{int}{\gtFmt{G_{\text{update}}}
				},\\
			&\commChoice{done}{}{
				\gtComm{B}{A}{total}[tot]{int}[tot = partial]{\gtEnd}
			}
		\end{aligned}
		}
	}
$$
with
\begin{align*}
\gtFmt{G_{\text{update}}} &=
\gtComm{B}{C}{partial}[tmp]{int}[tmp = partial + n]{
	G_{\text{update}}^\prime
	}\\
\gtFmt{G_{\text{update}}^\prime} &=
\gtComm{C}{B}{update}[partial]{int}[partial=tmp]{\gtRecVar{T}}
\end{align*}


\subsection{Diffie-Hellman protocol}

The Diffie-Hellman protocol is a protocol that allows two participants to
securely establish a shared secret, relying on the difficulty of the discrete
logarithm \cite{DiffieNew1976}.

In a nutshell, two participants $\participant{A}$ and $\participant{B}$ each
have a private key (resp. $\processVariable{a}$ and $\processVariable{b}$). In
addition, public values $\processVariable{p}$ and $\processVariable{g}$ are
known and not secret\footnote{In addition, $p$ and $g$ must be co-prime. In our
example, we only focus on the secret keys not being disclosed, and we therefore
ignore this condition.}. During the protocol $\participant{A}$ (resp.
$\participant{B}$) sends
$\processVariable{A} = \processVariable{g}^{\processVariable{a}}\text{ mod
}\processVariable{p}$ to
$\participant{B}$ (resp.
$\processVariable{B} = \processVariable{g}^{\processVariable{b}}\text{ mod
}\processVariable{p}$ to
$\participant{A}$). The shared secret is then obtain on each side with
$\processVariable{A}^{\processVariable{b}} =
\processVariable{B}^{\processVariable{a}}$. In our implementation, we have an
extra participant $\participant{G}$ used for setting-up the values. In
practice, $\processVariable{g}$ and $\processVariable{p}$ can be shared
beforehand or agreed-on on the fly. The private keys can be generated from a
random number generator local to each participant. The protocol we implement is
shown in \cref{fig:app:diffie-hellman}.

\begin{figure}
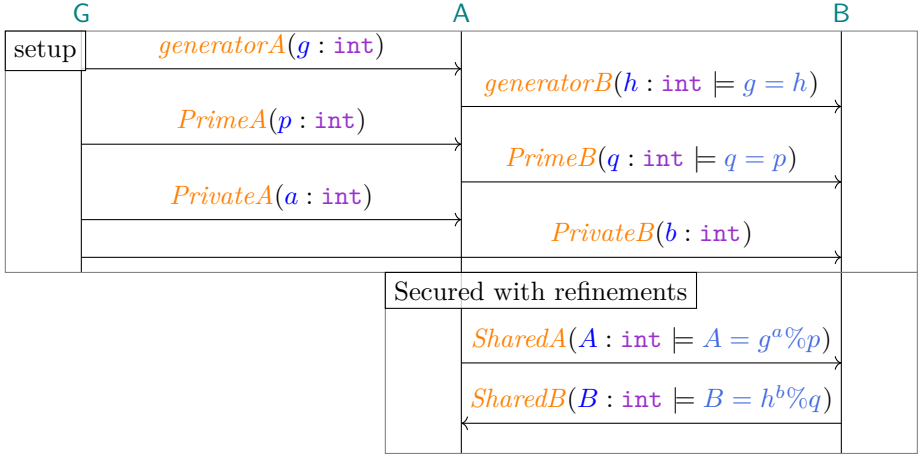

\begin{messagepassing}[yscale=5, ]
	\newprocesswithlength{B}[\participant{B}]{5.6}
	\newprocesswithlength{A}[\participant{A}]{5.6}
	\newprocesswithlength{G}[\participant{G}]{3.2}

	\sendwithname{G}{0.5}{A}{0.5}{
		$\gtLbl{generatorA}(\processVariable{g}: \processType{int})$}[pos=.5]
	\sendwithname{A}{1}{B}{1}{
		$\gtLbl{generatorB}(\refinedTypedVariable[][h]{int}[g = h])$}[pos=.5]
	\sendwithname{G}{1.5}{A}{1.5}{
		$\gtLbl{PrimeA}(\processVariable{p}: \processType{int})$}[pos=.5]
	\sendwithname{A}{2}{B}{2}{
		$\gtLbl{PrimeB}(\refinedTypedVariable[][q]{int}[q = p])$}[pos=.5]
	\sendwithname{G}{2.5}{A}{2.5}{
		$\gtLbl{PrivateA}(\processVariable{a}: \processType{int})$}[pos=.5]
	\sendwithname{G}{3}{B}{3}{
		$\gtLbl{PrivateB}(\refinedTypedVariable[][b]{int})$}[pos=.75]
	\draw[draw=black!50] (0, -3.2) rectangle (3.2, -0.8);
	\node[anchor=north west,
		fill opacity=0.8,
		text opacity=1,
		draw,
		fill=white,
		outer sep=0pt] at (0, -3.2) {setup};

	\sendwithname{A}{4.4}{B}{4.4}{
		$\gtLbl{SharedA}(\refinedTypedVariable[][A]{int}[A = g^a\%p])$}[pos=.5]
	\sendwithname{B}{5.2}{A}{5.2}{
		$\gtLbl{SharedB}(\refinedTypedVariable[][B]{int}[B = h^b\%q])$}[pos=.5]
	\draw[draw=black!50] (3.2, -2.2) rectangle (5.6, -0.8);
\node[anchor=north west,
fill opacity=0.8,
text opacity=1,
draw,
fill=white,
outer sep=0pt] at (3.2, -2.2) {Secured with refinements};
\end{messagepassing}
\caption{Communication diagram for the Diffie-Hellman key exchange protocol with
	refinements.}
\label{fig:app:diffie-hellman}
\end{figure}

The RMPST of this protocol is:
$$
\gtFmt{G_1} = \gtComm{G}{A}{generatorA}[g]{int}{
\gtComm{A}{B}{generatorB}[h]{int}[g = h]{G_2}}
$$
$$
\gtFmt{G_2} = \gtComm{G}{A}{PrimeA}[p]{int}{
\gtComm{A}{B}{PrimeB}[q]{int}[p = q]{G_3}}
$$
$$
\gtFmt{G_3} = \gtComm{G}{A}{PrivateA}[a]{int}{
\gtComm{G}{B}{PrivateB}[b]{int}{
G_4
}}
$$
$$
\gtFmt{G_4} = \gtComm{A}{B}{SharedA}[A]{int}[A = g^a\%p]{
	\gtComm{B}{A}{SharedB}[B]{int}[B = g^b\%q]{\gtEnd}}
$$

While this example is theoretically simple (no recursion), we implemented this
example in Rumpsteak, in order to show how to accommodate arithmetic
refinements.

\section{Refined MPST}
\subsection{Preliminary definitions\label{app:map}}

\begin{definition}[Map]
	A \emph{map} \(\map{M}\) is a set of pairs \(\tuple{\processVariable{t},
	\val{v}}\), where \(\processVariable{t}\) is a variable and \(\val{v}\)
	is a value, such that there are no two pairs with the same variable.
	Maps are equipped with the following operations: lookup
	\(\lookupMap{M}{x}\), update \(\updateMap{M}{x}{v}\),
	domain \(\domMap{M}\), and removal \(\removeMap{M}{x}\).
	\begin{mathpar}
	\lookupMap{M}{x}\define\begin{cases}
	\val{v} & \text{if }\tuple{\processVariable{x}, \val{v}}\in \map{M}\\
	\text{undefined} & \text{otherwise}
	\end{cases}
	\and
	\updateMap{M}{x}{v}\define (\map{M}\setminus \setcomp
	{\tuple{\processVariable{x}, \val{v^\prime}}}
	{\forall \val{v^\prime}})
	\cup \setof{\tuple{\processVariable{x}, \val{v}}}
	\and
	\domMap{M}\define
	\setcomp{
		\processVariable{x}
	}{
		\exists \val{v}\suchthat \tuple{\processVariable{x},
			\val{v}}\in \map{M}
	}
	\and
	\removeMap{M}{x}\define\map{M}\setminus\setcomp{\tuple{\processVariable{x},
	\val{v}}}{\forall\val{v}}
	\and
	\map{M_1}\biguplus\map{M_2} \define\begin{cases}
		M_1 \cup M_2 & \text{if }\domMap{M_1}\cap\domMap{M_2} = \emptyset\\
		\text{undefined} & \text{otherwise}
	\end{cases}
	\end{mathpar}

	We write \(\map{M_{\emptyset}}\) for the empty map.
\end{definition}

Given a map \(\map{M}\) and a refinement \(\refinement{r}\), we note
\(\map{M}\models \refinement{r}\)
if and only if
the refinement \(\refinement{r}\) is closed under the map \(\map{M}\):
\(\refinement{\fv{r}} \subseteq \domMap{M}\),
and
evaluates to truth after substitution:
\(\refeval{r\subst{\refinement{\fv{r}}}{\lookupMap{M}{\refinement{\fv{r}}}}}=\top\).

\begin{definition}[Queues\label{def:queue}]
	A \emph{queue} $\queue{w}$ is a set of FIFOs for each pair of distinct
	participants in $\set{\setRoles}\times\set{\setRoles}$.
	$\lookupQueue{w}{p}{q}$ denotes a FIFO from
	$\participant{p}$ to $\participant{q}$ in $\queue{w}$.
	We define:
	\begin{enumerate}
		\item
		$\pushQueue{w}{e}{p}{q}
		\define
		\setcomp{
			\lookupQueue{w}{p^\prime}{q^\prime}
		}{
			\participant{p^\prime} \neq \participant{p} \vee
			\participant{q^\prime} \neq \participant{q}
		}
		\cup \{\push{\lookupQueue{w}{p}{q}}{e}\}$
		\item $\popQueue{w}{p}{q}
		\define
		\setcomp{
			\lookupQueue{w}{p^\prime}{q^\prime}
		}{
			\participant{p^\prime} \neq \participant{p} \vee
			\participant{q^\prime} \neq \participant{q}
		}
		\cup \{\pop{\lookupQueue{w}{p}{q}}\}$ if $\pop{\lookupQueue{w}{p}{q}}$
		is defined
		\item
		$\nextElemQueue{w}{p}{q} \define
		\nextElem{\lookupQueue{w}{p}{q}}$
	\end{enumerate}
	We write $\emptyqueue$ for the empty queue, which is the queue where
	$\lookupQueue{w}{p}{q} = \emptyfifo$ for all $\participant{p}$ and
	$\participant{q}$.
\end{definition}

\begin{definition}[Trace Ending-Up with Map\label{def:end_up}]
	A trace \(\trace\) \emph{ends up} with \(\map{M_{\traceP}}\) w.r.t.\ an
	initial map \(\map{M_\text{\rm I}}\) if and only if:
	\begin{enumerate}
		\item if \(\trace\) is \(\tempty\),
		then \(\map{M_\text{\rm I}} = \map{M_{\traceP}}\)
		\label{def:end_up:i}; and
		\item if \(\trace\) is
		\(\tconcat{\actionGeneric{\roleP}{\roleQ}{\lblL,(\varX,
				\val{c})}{r}}{\tau^\prime}\),
		then \(\trace[\tau^\prime]\) ends up with \(\map{M_{\traceP}}\)
		w.r.t.\ \(\updateMap{M_\text{\rm I}}{x}{c}\).
		\label{def:end_up:ii}
		\qedhere
	\end{enumerate}
\end{definition}

\subsection{Run and trace of an RCS}

In order to define traces of RCSs, we first define \emph{runs} (sequences of
states, where the order is consistent with the reduction rules) and
explain how to obtain a trace (as defined in~\cref{def:trace}) from a run:
the function $\operatorname{trace\_step}$ extracts the action that
happens between two consecutive configurations; by running this function on
all successive states of a run, we retrieve the sequence of actions
that took place, i.e. the trace of the run.

\begin{definition}[Run of an RCS]
	\label{def:run}
	A \emph{run} of an RCS is a sequence
	\(\sigma_0; \dots\) of refined configurations such that
	\begin{inlineenum}
		\item for all \(i\in\{1, \dots\}\),
		\(\sigma_{i-1}\csred\sigma_i\)
		\item \(\sigma_0\) is initial
		\label{def:run:ii}
		\item if the sequence is finite, then the last configuration
		\(\sigma_n\) is final.
	\end{inlineenum}
\end{definition}

\begin{remark}[Reachable State]
	We say a state $\sigma$ of an RCS is \emph{reachable} if there is a run of
	$R$ that contains $\sigma$.
	This implies the run begins from an
	initial state (\cref{def:run:ii} in \cref{def:run}).
\end{remark}

\begin{definition}[Trace of a Reduction Step\label{def:trace_red_step}]
	\begin{equation*}
		\operatorname{trace\_step}(\tuple{\sigma_1, \sigma_2}) \define
		\begin{cases}
			\ltsLabelRecv{j}{i}{\msgM}[r] & \text{if }
			\sigma_1
			\csred[{\red{s_\participant{i}}{s_{\participant{i}}^\prime}{\mathrel{?}}{\msgM}{j}{i}[r]}]
			\sigma_2
			\\
			\ltsLabelSend{i}{j}{\msgM}[r] & \text{if }
			\sigma_1
			\csred[{\red{s_\participant{i}}{s_{\participant{i}}^\prime}{\mathrel{!}}{\msgM}{i}{j}[r]}]
			\sigma_2
			\\
			\text{undefined} & \text{otherwise}
		\end{cases}
	\end{equation*}
\end{definition}

\begin{definition}[Trace of a Refined Communicating System\label{def:trace_RGA}]
	Given a run \(\sigma_0; \dots\) of an RCS, the trace of those
	reductions is given with the following function:
	\begin{equation*}
		\operatorname{trace}(\sigma_0; \sigma_1; \dots)=
		\begin{cases}
			\tconcat{\withcolor{black}{\operatorname{trace\_step}(\tuple{\sigma_0,
						\sigma_1})}}{\withcolor{black}{\operatorname{trace}(\sigma_1;
					\dots)}} & \text{if
					}\operatorname{trace\_step}(\tuple{\sigma_0,
				\sigma_1})\text{ is defined} \\
			\text{undefined}& \text{otherwise}
		\end{cases}
	\end{equation*}
	\begin{mathpar}
		\operatorname{trace}(\epsilon) = \tempty
		\and
		\operatorname{trace}(\sigma_0) = \tempty
	\end{mathpar}

\end{definition}

\begin{remark}[Trace of step is invertible\label{rmk:trace_invertible}]
	\(\operatorname{trace\_step}\) is injective and therefore
	invertible. For the sake of simplicity, we implicitly convert
	traces into runs and vice-versa.
\end{remark}

\begin{example}[Run and Trace of an RCS\label{ex:run_trace_CS}]
	A possible run of the RCS of \(\gtFmt{G_\pm}\) is:\\[1mm]
	$
	\begin{array}{l}
		\small
		\tuple{\tuple{\participant{A_1}, \participant{B_1}, \participant{C_1}},
			\emptyqueue, \emptyMap};
		\tuple{\tuple{\participant{A_2}, \participant{B_1}, \participant{C_1}},
			\queue{w_1}, \map{\{\tuple{\processVariable{n}, \val{5}}\}}};
		\tuple{\tuple{\participant{A_2}, \participant{B_2}, \participant{C_1}},
			\emptyqueue, \map{\{\tuple{\processVariable{n}, \val{5}}\}}};\\
		\tuple{\tuple{\participant{A_2}, \participant{B_2}, \participant{C_2}},
			\queue{w_2}, \map{M}};
		\tuple{\tuple{\participant{A_2}, \participant{B_3}, \participant{C_2}},
			\emptyqueue, \map{M}};
		\tuple{\tuple{\participant{A_2}, \participant{B_4}, \participant{C_2}},
			\queue{w_3}, \map{M}};
		\tuple{\tuple{\participant{A_2}, \participant{B_4}, \participant{C_3}},
			\emptyqueue, \map{M}}
	\end{array}
	$\\
	with the queues $\queue{w_1} =
	\pushQueue{\emptyqueue}{\msgM[\tuple{\gtLbl{secret},
			\tuple{\processVariable{n}, \val{5}}}]}{A}{B}$; $\queue{w_2} =
	\pushQueue{\emptyqueue}{\msgM[\tuple{\gtLbl{guess},
	\tuple{\processVariable{x},
				\val{5}}}]}{C}{B}$; $\queue{w_3} =
	\pushQueue{\emptyqueue}{\msgM[\tuple{\gtLbl{correct},
			\tuple{\processVariable{\_}, \val{\_}}}]}{B}{C}$; and the map
			\(\map{M} =
	\map{\{\tuple{\processVariable{n}, \val{5}},
		\tuple{\processVariable{x}, \val{5}}\}}\).

	This run produces the trace \(\tconcat{\trace}{\trace[\tau_2]}\) presented
	in
	\cref{ex:WPT}. Notice that, from the configuration with
	\(\tuple{\participant{A_2}, \participant{B_3}, \participant{C_2}}\),
	the system cannot take any of the \grsnd{} transition with
	\(\msgM[\tuple{\gtLbl{more}, \tuple{\processVariable{\_}, \val{\_}}}]\)
	and
	\(\msgM[\tuple{\gtLbl{less}, \tuple{\processVariable{\_},
			\val{\_}}}]\),
	as the refinement in the corresponding transition in the RCFSM of
	\(\participant{B}\) does not hold with \(\map{M}\): \(\map{M}\not\models
	\refinement{x < n}\) (resp.\ \(\map{M}\not\models \refinement{x > n}\))
	(c.f.\ \cref{ex:RCS_transitions}).
\end{example}

\subsection{Syntax\label{app:rmpst_syntax}}

\subsubsection{Roles in a global type}

\begin{definition}[Set of Roles in a Global Type\label{def:role_in_gtype}]
	\[\rolesOfGT{G} = \begin{cases}
		\{\roleP, \roleQ\}\cup \bigcup_{i\in I} \rolesOfGT{G_i} & \text{if } \gSessionType{G} = \gtComm[i][I]{p}{q}{l}[x]{S}[R]{G}\\
		\rolesOfGT{G^\prime} & \text{if } \gSessionType{G} = \gtRec{t}{G^\prime}\\
		\emptyset & \text{otherwise}
	\end{cases}\]
	We note \(\participant{p}\in\gtFmt{G}\) for \(\participant{p}\in\rolesOfGT{G}\).
\end{definition}

\subsubsection{Type occurring in a type}

 Cases (i) --- (iii) of
\inApp{def:type_occurring_type} recursively delve
into the continuations, until we match on the exact type with case (iv).

\begin{definition}[Type Occurring in a Type\label{def:type_occurring_type}]
	We say a type \(\ltFmt{T^\prime}\) occurs in \(\ltFmt{T}\) (noted
	\(\occursLoc{T^\prime}{T}\)) if and only if at least one of the following
	conditions holds:
	\begin{inlineenum}[or]
		\item if \(\ltFmt{T}\) is
		\(\ltIntC[i][I]{\participant{p}}{\ell}{x}{S}[R]{\lSessionType{T}}\),
		there exist \(i\in I\) such that \(\occursLoc{T^\prime}{T_i}\)
		\item if
		\(\ltFmt{T}\) is
		\(\ltExtC[i][I]{\participant{p}}{\ell}{x}{S}{R}{\lSessionType{T}}\),
		there exist \(i\in I\) such that \(\occursLoc{T^\prime}{T_i}\)
		\item if \(\ltFmt{T}\) is \(\ltRec{t}{T_\mu}\),
		\(\occursLoc{T^\prime}{T_\mu}\)
		\item \(\lSessionType{T^\prime} = \lSessionType{T}\).
	\end{inlineenum}
\end{definition}

\begin{definition}[Global Type Occurring in a Global
Type\label{def:g_type_occurring_g_type}]
	We say a type \(\gtFmt{G^\prime}\) occurs in \(\gtFmt{G}\) (noted
	\(\occursGlob{G^\prime}{G}\)) if and only if at least one of the following
	conditions holds:
	\begin{itemize}
		\item if
		\(\gtFmt{T}\) is
		$\gtComm[i][I]{p}{q}{\ell}[x]{S}[r]{G}$,
		there exist \(i\in I\) such that \(\occursGlob{G^\prime}{G_i}\)
		\item if \(\gtFmt{G}\) is \(\gtRec{t}{G_\mu}\),
		\(\occursGlob{G^\prime}{G_\mu}\)
		\item \(\gSessionType{G^\prime} = \gSessionType{T}\).
	\end{itemize}
\end{definition}

\section{Section~\ref{sec:traces}}
\label{app:traces}
\subsection{Proofs of Lemmas}

\begin{lemma}[Concatenating Well-Queued Traces]
	\label{lem:concat_WQ_traces}
	For any traces \(\trace[\tau_1]\) and \(\trace[\tau_2]\),
	for any queues \(\queueQ[i]\), \(\queueQ[t]\) and \(\queueQ[f]\),
	if \(\trace[\tau_1]\) ends up with queue \(\queueQ[t]\) with respect to
	\(\queueQ[i]\),
	and \(\trace[\tau_2]\) ends up with queue \(\queueQ[f]\) with respect to
	\(\queueQ[t]\),
	then \(\tconcat{\tau_1}{\tau_2}\) ends up with queue \(\queueQ[f]\) with
	respect to \(\queueQ[i]\).
\end{lemma}

\begin{proof}
	By induction on the size of \(\trace[\tau_1]\).
	Notice that we leave \(\queueQ[i]\), \(\queueQ[t]\), \(\queueQ[f]\) and \(\trace[\tau_2]\)
	quantified; hence the property we want to prove by induction is
	\(\forall \trace[\tau_2], \queueQ[i], \queueQ[t], \queueQ[f]\),
	if \(\trace[\tau_1]\) ends up with queue \(\queueQ[t]\) w.r.t.
	\(\queueQ[i]\) and \(\trace[\tau_2]\) ends up with queue \(\queueQ[f]\)
	w.r.t. \(\queueQ[t]\), then \(\tconcat{\tau_1}{\tau_2}\) ends up
	with queue \(\queueQ[f]\) w.r.t. \(\queueQ[i]\).

	\begin{description}
		\item[Base case, size of \({\trace[\tau_1]}\) is \(0\):] in this
			case, \(\trace[\tau_1] = \tempty\), which trivially
			holds.
		\item[Base case, size of \({\trace[\tau_1]}\) is \(1\):] in this
			case, \(\trace[\tau_1] = \act\).
			We have to show that, for all \(\trace[\tau_2]\), if
			\(\act\) ends up with queue \(\queueQ[t]\) w.r.t.
			\(\queueQ[i]\),
			and if \(\trace[\tau_2]\) ends up with queue
			\(\queueQ[f]\) w.r.t. \(\queueQ[t]\),
			then \(\tconcat{\act}{\tau_2}\) ends up with
			\(\queueQ[f]\) w.r.t. \(\queueQ[i]\).

			By case analysis on \(\act\):
			\begin{description}
				\item[If \(\act\) is
					\(\actionReceive{\roleP}{\roleQ}{\msgM}{r}\),]
					since \(\act\) (i.e. \(\tconcat{\act}{\tempty}\)) ends up with \(\queueQ[t]\)
					w.r.t. \(\queueQ[i]\),
					then, from \cref{item:def:ends_up_with_queue:ii} in \cref{def:ends_up_with_queue},
					\(\tempty\) ends up with queue \(\queueQ[t]\) w.r.t. \(\popQueue{\queueQ[i]}{\roleP}{\roleQ}\),
					and \(\nextElemQueue{\queueQ[i]}{\roleP}{\roleQ} = \msgM\).

					From \cref{item:def:ends_up_with_queue:i} in \cref{def:ends_up_with_queue},
					\(\queueQ[t] = \popQueue{\queueQ[i]}{\roleP}{\roleQ}\).

					Therefore, we have that:
					\begin{itemize}
						\item \(\trace[\tau_2]\) ends up with \(\queueQ[f]\) w.r.t.
							\(\queueQ[t] = \popQueue{\queueQ[i]}{\roleP}{\roleQ}\)
							(by hypothesis and above); and
						\item \(\nextElemQueue{\queueQ[i]}{p}{q} = \msgM\)
							(as shown above).
					\end{itemize}
					Therefore, from \cref{item:def:ends_up_with_queue:ii} in
					\cref{def:ends_up_with_queue},
					\(\tconcat{\act}{\tau_2}\) ends up with \(\queueQ[f]\) w.r.t.
					\(\queueQ[i]\).

				\item[If \(\act\) is
					\(\actionSend{\roleP}{\roleQ}{\msgM}{r}\),]
					since \(\act\) (i.e. \(\tconcat{\act}{\tempty}\)) ends up with \(\queueQ[t]\)
					w.r.t. \(\queueQ[i]\),
					then, from \cref{item:def:ends_up_with_queue:iii} in \cref{def:ends_up_with_queue},
					\(\tempty\) ends up with queue \(\queueQ[t]\) w.r.t.
					\(\pushQueue{\queueQ[i]}{\msgM}{\roleP}{\roleQ}\).

					From \cref{item:def:ends_up_with_queue:i} in \cref{def:ends_up_with_queue},
					\(\queueQ[t] = \pushQueue{\queueQ[i]}{\msgM}{\roleP}{\roleQ}\).

					Therefore, we have that \(\trace[\tau_2]\) ends up with \(\queueQ[f]\) w.r.t.
					\(\queueQ[t] = \pushQueue{\queueQ[i]}{\msgM}{\roleP}{\roleQ}\)
					(by hypothesis and above).
					Therefore, from \cref{item:def:ends_up_with_queue:iii} in
					\cref{def:ends_up_with_queue},
					\(\tconcat{\act}{\tau_2}\) ends up with \(\queueQ[f]\) w.r.t.
					\(\queueQ[i]\).
			\end{description}
		\item[Inductive case, size of \({\trace[\tau_1]}\) is \(n+1\)
			(\(n \geq 1\)):]
			The induction hypothesis (IH) is:
			for all \(\trace[\tau_1^{IH}]\) with length less or equal to \(n\),
			for all \(\trace[\tau_2^{IH}], \queueQ[i]^{IH}, \queueQ[t]^{IH}, \queueQ[f]^{IH}\),
			if \(\trace[\tau_1^{IH}]\) ends up with queue \(\queueQ[t]^{IH}\) w.r.t \(\queueQ[i]^{IH}\),
			and \(\trace[\tau_2^{IH}]\) ends up with queue \(\queueQ[f]^{IH}\) w.r.t. \(\queueQ[t]^{IH}\),
			then \(\tconcat{\tau_1^{IH}}{\tau_2^{IH}}\) ends up with queue \(\queueQ[f]^{IH}\)
			w.r.t. \(\queueQ[i]^{IH}\).

			Since the length of \(\trace[\tau_1]\) is \(n+1 \geq 2\),
			\(\trace[\tau_1] = \tconcat{\act}{\tau_1^\prime}\).
			Notice that the length of \(\trace[\tau_1^\prime]\) is \(n\).

			Since \(\trace[\tau_1]\) ends up with \(\queueQ[t]\) w.r.t. \(\queueQ[i]\),
			then \(\trace[\tau_1^\prime]\) ends up with queue \(\queueQ[t]\) w.r.t. either
			\(\pushQueue{\queueQ[i]}{\msgM}{p}{q}\) or \(\popQueue{\queueQ[i]}{p}{q}\) (depending on
			\(\act\)), let call it \(\queueQ[\act]\).

			Therefore, by applying the induction hypothesis with
			\(\trace[\tau_1^{IH}] = \trace[\tau_1^\prime]\),
			\(\trace[\tau_2^{IH}] = \trace[\tau_2]\),
			\(\queueQ[i]^{IH} = \queueQ[\act]\),
			\(\queueQ[t]^{IH} = \queueQ[t]\),
			\(\queueQ[f]^{IH} = \queueQ[f]\),
			we have that \(\tconcat{\tau_1^\prime}{\tau_2}\) ends up with queue
			\(\queueQ[f]\) w.r.t. \(\queueQ[\act]\).

			By applying the induction hypothesis a second time, with
			\(\trace[\tau_1^{IH}] = \act\),
			\(\trace[\tau_2^{IH}] = \tconcat{\tau_1^\prime}{\tau_2}\),
			\(\queueQ[i]^{IH} = \queueQ[i]\),
			\(\queueQ[t]^{IH} = \queueQ[t_i]\),
			\(\queueQ[f]^{IH} = \queueQ[f]\),
			we have that \(\tconcat{\act}{\tconcat{\tau_1^\prime}{\tau_2}} =
			\tconcat{\tau_1}{\tau_2}\) ends up with queue \(\queueQ[f]\)
			w.r.t. \(\queueQ[i]\), which concludes the inductive step.
	\qedhere
	\end{description}
\end{proof}

\separator

\begin{lemma}[Concatenating Well-Predicated Traces]
	\label{lem:concat_WP_traces}
	For any map \(\map{M}\), for any traces \(\trace[\tau_1]\) and
	\(\trace[\tau_2]\), if \(\trace[\tau_1]\)
	is well-predicated under \(\map{M}\) and ends up with
	\(\map{M_{\trace[\tau_1]}}\) with respect to \(\map{M}\), and if
	\(\trace[\tau_2]\) is well-predicated under
	\(\map{M_{\trace[\tau_1]}}\), then \(\tconcat{\tau_1}{\tau_2}\) is
	well-predicated under \(\map{M}\).
\end{lemma}

\begin{proof}
	By induction on the size of \(\trace[\tau_1]\):
	\begin{description}
		\item[Case size of \({\trace[\tau_1]}\) is \(0\):]
			In that case, \(\trace[\tau_1] = \tempty\),
			therefore \(\tconcat{\tau_1}{\tau_2} = \trace[\tau_2]\)
			and \(\map{M_{\tau_1}} = \map{M}\) (from \cref{def:end_up:i} in \cref{def:end_up}).
			The result then trivially holds.
		\item[Case size of \({\trace[\tau_1]}\) is \(n+1\):]
			the induction hypothesis is: "for all
			\(\map{M^\prime}\), \(\trace[\tau_1^\prime]\) and
			\(\trace[\tau_2^\prime]\),
			(H1) if the size of \(\trace[\tau_1^\prime]\) is less than or equal to \(n\), then
			(H2) if \(\trace[\tau_1^\prime]\) is well-predicated with respect to \(\map{M^\prime}\)
			and ends up with \(\map{M^\prime_{\trace[\tau_1^\prime]}}\)
			w.r.t.\ to \(\map{M^\prime}\),
			and (H3) if \(\trace[\tau_2^\prime]\) is well-predicated with respect to \(\map{M^\prime_{\trace[\tau_1^\prime]}}\),
			then \(\tconcat{\tau_1^\prime}{\tau_2^\prime}\)
			is well-predicated with respect to \(\map{M^\prime}\)".

			Let \(\trace[\tau_1] = \tconcat{\act[\alpha_0]}{\tconcats{\act[\alpha_1]}{\act[\alpha_n]}}\),
			with \(\act[\alpha_0] =
			\actionGeneric{\roleP}{\roleQ}{\lblL,(\varX, \val{c})}{r}\).

			Given that \(\trace[\tau_1]\) is well-predicated with respect to \(\map{M}\),
			then
			\begin{inlineenum}
			\item \(\refinement{r}\) holds under \(\updateMap{M}{x}{c}\)
				\label{lem:concat_WP_traces:i}
			\item \(\tconcats{\act[\alpha_1]}{\act[\alpha_n]}\) is well-predicated with respect to \(\updateMap{M}{x}{c}\).
				\label{lem:concat_WP_traces:ii}
			\end{inlineenum}

			Let \(\map{M}_{\tconcats{\act[\alpha_1]}{\act[\alpha_n]}}\)
			be the map \(\tconcats{\act[\alpha_1]}{\act[\alpha_n]}\)
			ends up to with respect to \(\updateMap{M}{x}{c}\).

			From the induction hypothesis (with \(\map{M^\prime}\) being \(\updateMap{M}{x}{c}\),
			\(\trace[\tau_1^\prime]\) being \(\tconcats{\act[\alpha_1]}{\act[\alpha_n]}\),
			\(\trace[\tau_2^\prime]\) being \(\trace[\tau_2]\)),
			and given that
			\begin{inlineenum}
			\item (H1) holds trivially
			\item (H2) holds from \cref{lem:concat_WP_traces:ii}
			\item (H3) holds from \cref{def:end_up:ii} in \cref{def:end_up}
				(which directly shows that \(\map{M_{\trace[\tau_1]}} = \map{M_{\tconcats{\act[\alpha_1]}{\act[\alpha_n]}}}\))
			\end{inlineenum}
			then \(\tconcat{\tconcats{\act[\alpha_1]}{\act[\alpha_n]}}{\tau_2}\) is well-predicated under \(\updateMap{M}{x}{c}\).

			From \cref{lem:concat_WP_traces:i} above and \cref{def:WPT:ii} in \cref{def:WPT},
			we have that \(\tconcat{\act[\alpha_0]}{\tconcat{\tconcats{\act[\alpha_1]}{\act[\alpha_n]}}{\tau_2}} = \tconcat{\tau_1}{\tau_2}\)
			is well-predicated under \(\map{M}\), which concludes the inductive step.
			\qedhere
	\end{description}
\end{proof}

\section{Refined Automata\label{app:refined_automata}}

\begin{figure}
\begin{subfigure}{\textwidth}
	\begin{center}
	\begin{tikzpicture}[node distance=4.2cm]
		\draw node[draw, rounded rectangle] (G1) {\(\participant{A_1}\)};
		\draw node[draw, rounded rectangle, right=3.5cm of G1] (G2) {\(\participant{A_2}\)};
		\coordinate [above=.5cm of G1] (start);
		\draw[-{Stealth[scale=1.5]}] (start) -- (G1);

		\draw[-{Stealth[scale=1.5]}] (G1) to node[anchor=south, pos=.5, sloped] {\(\actionSend{A}{B}{\gtLbl{secret}, \tuple{\processVariable{n}, \val{c_n}}}{\taut}\)} (G2);
	\end{tikzpicture}
		\caption{RCFSM of \(\participant{A}\) in the \(\gtFmt{G_\pm}\)
		protocol.}
	\end{center}
\end{subfigure}

	\vspace{.5cm}
\begin{subfigure}{\textwidth}
	\begin{center}
	\begin{tikzpicture}[node distance=4.2cm]
		\draw node[draw, rounded rectangle] (G1) {\(\participant{B_1}\)};
		\draw node[draw, rounded rectangle, right=3.5cm of G1] (G2) {\(\participant{B_2}\)};
		\draw node[draw, rounded rectangle, right=3.5cm of G2] (G3) {\(\participant{B_3}\)};
		\draw node[draw, rounded rectangle, right=3.9cm of G3] (G4) {\(\participant{B_4}\)};
		\coordinate [above=.5cm of G1] (start);
		\draw[-{Stealth[scale=1.5]}] (start) -- (G1);

		\draw[-{Stealth[scale=1.5]}] (G1) to node[anchor=south, pos=.5, sloped] {\(\actionReceive{A}{B}{\gtLbl{secret}, \tuple{\processVariable{n}, \val{c_n}}}{\taut}\)} (G2);
		\draw[-{Stealth[scale=1.5]}] (G2) to node[anchor=south, pos=.5, sloped] {\(\actionReceive{C}{B}{\gtLbl{guess}, \tuple{\processVariable{x}, \val{c_x}}}{\taut}\)} (G3);

		\coordinate (midG23) at ($ (G3)!.5!(G2) $);
		\coordinate [above=.7cm of midG23] (midG23up);
		\coordinate [below=.7cm of midG23] (midG23down);

		\draw[-{Stealth[scale=1.5]}]
			(G3) |- (midG23up) node[anchor=south] {\(\actionSend{B}{C}{\gtLbl{more}, \tuple{\processVariable{\_}, \val{\_}}}{x < n}\)} -|  (G2);
		\draw[-{Stealth[scale=1.5]}]
			(G3) |- (midG23down) node[anchor=north] {\(\actionSend{B}{C}{\gtLbl{less}, \tuple{\processVariable{\_}, \val{\_}}}{x > n}\)}  -|  (G2);
		\draw[-{Stealth[scale=1.5]}] (G3) to node[anchor=south, pos=.5, sloped] {\(\actionSend{B}{C}{\gtLbl{correct}, \tuple{\processVariable{\_}, \val{\_}}}{x = n}\)} (G4);
	\end{tikzpicture}
		\caption{RCFSM of \(\participant{B}\) in the \(\gtFmt{G_\pm}\)
		protocol.}
	\end{center}
\end{subfigure}

	\vspace{.5cm}
\begin{subfigure}{\textwidth}
	\begin{center}
	\begin{tikzpicture}[node distance=3.5cm]
		\draw node[draw, rounded rectangle] (G1) {\(\participant{C_1}\)};
		\draw node[draw, rounded rectangle, right=of G1] (G2) {\(\participant{C_2}\)};
		\draw node[draw, rounded rectangle, right=of G2] (G3) {\(\participant{C_3}\)};
		\coordinate [left=.5cm of G1] (start);
		\draw[-{Stealth[scale=1.5]}] (start) -- (G1);

		\draw[-{Stealth[scale=1.5]}] (G1) to node[anchor=south, pos=.5, sloped] {\(\actionSend{C}{B}{\gtLbl{guess}, \tuple{\processVariable{x}, \val{c_x}}}{\taut}\)} (G2);

		\coordinate (midG12) at ($ (G2)!.5!(G1) $);
		\coordinate [above=.7cm of midG12] (midG12up);
		\coordinate [below=.7cm of midG12] (midG12down);

		\draw[-{Stealth[scale=1.5]}]
			(G2) |- (midG12up) node[anchor=south] {\(\actionReceive{B}{C}{\gtLbl{more}, \tuple{\processVariable{\_}, \val{\_}}}{\taut}\)} -|  (G1);
		\draw[-{Stealth[scale=1.5]}]
			(G2) |- (midG12down) node[anchor=north] {\(\actionReceive{B}{C}{\gtLbl{less}, \tuple{\processVariable{\_}, \val{\_}}}{\taut}\)}  -|  (G1);
		\draw[-{Stealth[scale=1.5]}] (G2) to node[anchor=south, pos=.5, sloped] {\(\actionReceive{B}{C}{\gtLbl{correct}, \tuple{\processVariable{\_}, \val{\_}}}{\taut}\)} (G3);
	\end{tikzpicture}
		\caption{RCFSM of \(\participant{C}\) in the \(\gtFmt{G_\pm}\)
		protocol.}
	\end{center}
\end{subfigure}
\caption{RCFSM of the participants of the \(\gtFmt{G_\pm}\) protocol.}
\label{fig:ex:RCFSM_GPM}
\end{figure}

\begin{lemma}
	\label{thm:run_send_before_receive}
	For all runs $\sigma_0; \dots; \sigma_{i-1}; \sigma_i; \dots$,
	if $\sigma_{i-1}\csred[t]\sigma_i$ with $t =
	\redact{s_\participant{i}}{s_{\participant{i}}^\prime}{\actionReceive{j}{i}{\gtLbl{\ell},
			\tuple{\processVariable{x}, \val{c}}}{r}}$,
	then exist $j \leq i-1$ and $\refinement{r^\prime}$ such that
	$\sigma_{j-1}\csred[t^\prime]\sigma_j$ with
	$t^\prime=\redact{s_\participant{i}}{s_{\participant{i}}^\prime}{\actionSend{j}{i}{\gtLbl{\ell},
			\tuple{\processVariable{x}, \val{c}}}{r^\prime}}$.
\end{lemma}
\begin{proof}
	Let $\queue{w^\prime}$ be the queues of $\sigma^\prime$.
	From the premises of \grrec{}, $\msgM=\msgM[\tuple{\gtLbl{\ell},
		\tuple{\processVariable{x}, \val{c}}}]\in\lookupQueue{w}{j}{i}$.
	The only rule that enqueues $\msgM$ in $\lookupQueue{w}{j}{i}$ is \grsnd{},
	with
	$\redact{s_\participant{i}}{s_{\participant{i}}^\prime}{\actionSend{j}{i}{\gtLbl{\ell},
			\tuple{\processVariable{x}, \val{c}}}{r^\prime}}$
	(for some $\refinement{r^\prime}$)
\end{proof}

\tracesGRAvalidrefined*
\begin{proof}
	From \cref{def:valid_refined_trace}, we have to show that
	\begin{inlineenum}
	\item \(\trace\) is well-queued with regards to the empty queue \(\queueQ[\emptyset]\)
	\item \(\trace\) is well-predicated by the empty variable context \(\varsV[\emptyset]\).
	\end{inlineenum}
	We show the two points separately.

	\paragraph*{Well-Queued}
	By case analysis on the length of \(\trace\).
	\begin{description}
		\item [Case length of \(\trace\) is 0 (\(\trace = \tempty\)):]
			since the trace is empty, initial and final, the
			corresponding run of the automaton is composed
			of a single state \(\sigma = \tuple{\tuple{s_\participant{i}}_{\participant{i}\in I}, \queueQ, \map{m}}\), which is therefore both
			initial and final.

			From \cref{def:init_RGS} (Initial Refined Global
			State), the queues \(\queueQ\) of the state are empty
			(\(\queueQ = \emptyqueue\)).

			Therefore, from \cref{def:well_queued_trace},
			\(\trace\) is well-queued with regards to
			\(\emptyqueue\).

		\item [Case length of \(\trace\) is 1 (\(\trace = \act\)):]
			We prove this case leads to a contradiction, and therefore cannot
			happen.

			From the definition of \(\delta\),
			the label of the transition is either
			\(\actionReceive{p}{q}{\gtLbl{\ell}, \tuple{\processVariable{x},
			\val{c}}}{r}\) or
			\(\actionSend{p}{q}{\gtLbl{\ell}, \tuple{\processVariable{x},
			\val{c}}}{r}\).

			Therefore, for any \(\sigma_1\), \(\sigma_2\) such that \(\operatorname{trace}(\sigma_1;\sigma_2) = \act\),
			we show that it is not simulatenously possible for \(\sigma_1\) to be initial
			and for \(\sigma_2\) to be final.
			Let \(\queueQ[\sigma_1]\) (resp. \(\queueQ[\sigma_2]\)) be
			the queue of \(\sigma_1\) (resp. \(\sigma_2\)).

			If the label is
			\(\actionReceive{p}{q}{\gtLbl{\ell}, \tuple{\processVariable{x},
			\val{c}}}{r}\),
			then from \cref{def:trace_RGA,def:trace_red_step},
			\(\sigma_1\csred\sigma_2\) with a \grrec{} transition.
			From the premise of \grrec{},
			\(\nextElemQueue{\queueQ[\sigma_1]}{p}{q} = \msgM[\tuple{\gtLbl{\ell}, \tuple{\processVariable{x}, \val{c}}}]\),
			i.e. \(\lookupQueue{\queueQ[\sigma_1]}{p}{q}\) is not \(\emptyfifo\),
			i.e. \(\queueQ[\sigma_1] \neq \emptyqueue\), therefore \(\sigma_1\) is not initial.

			Similarly, if the label is
			\(\actionSend{p}{q}{\gtLbl{\ell}, \tuple{\processVariable{x},
			\val{c}}}{r}\),
			then from \cref{def:trace_RGA,def:trace_red_step},
			\(\sigma_1\csred\sigma_2\) with a \grsnd{} transition.
			From the conclusion of \grsnd{},
			\(\queueQ[\sigma_2] = \pushQueue{\queueQ[\sigma_1]}{\msgM[\tuple{\gtLbl{\ell}, \tuple{\processVariable{x}, \val{c}}}]}{p}{q}\).
			Therefore, from~\cref{def:queue}, \(\queueQ[\sigma_2] \neq\emptyqueue\),
			therefore \(\sigma_2\) is not final.

			Therefore, if \(\trace\) contains a single element,
			then \(\trace\) is not initial and final, which contradicts the hypothesis.

		\item [Case length of \(\trace\) is greater than \(1\) (\(\trace = \tconcat{\tau_1}{\tau_2}\) for non-empty \(\traceRaw{\tau_1}\) and \(\traceRaw{\tau_2}\)):]
			first, let's notice that \(\trace\) is initial and final; therefore the initial state \(\sigma_i\)
			and the final state \(\sigma_f\) of \(\trace\) both have \(\emptyqueue\) as their queues
			(from \cref{def:init_RGS}). Therefore, we only have to show that \(\trace\) ends
			up with \(\queueQ[f]\) (the queue of \(\sigma_f\)) with respect to \(\queueQ[i]\) (the queue of \(\sigma_i\)).

			By contradiction, suppose \(\trace\) does not end up in \(\queueQ[f]\) w.r.t \(\queueQ[i]\).

			Let \(\sigma_t\) be the final state of \(\trace[\tau_1]\), starting from \(\sigma_i\).
			Let \(\queueQ[t]\) be the queue of \(\sigma_t\).

			From the contraposition of \cref{lem:concat_WQ_traces}, either
			\begin{inlineenum}[or]
			\item \(\trace[\tau_1]\) does not end up with queue \(\queueQ[t]\) w.r.t. \(\queueQ[i]\)
			\item \(\trace[\tau_2]\) does not end up with queue \(\queueQ[f]\) w.r.t. \(\queueQ[t]\).
			\end{inlineenum}

			In either case, we can recursively apply the contraposition of \cref{lem:concat_WQ_traces}
			until we have a trace composed of a single transition which trace is \(\act[\alpha_c]\) that does not end up in \(\queueQ[c]^\prime\) w.r.t. \(\queueQ[c]\).

			By case analysis of \(\act[\alpha_c]\):
			\begin{description}
				\item[Case \({\act[\alpha_c]} =
				\actionSend{\roleP}{\roleQ}{\msgM}{r}\) :]
					From \cref{rmk:trace_invertible}, we deduce that
					\(\act[\alpha_c]\) is a \grsnd{} transition
					\(\redact{s_\participant{p}}{s_\participant{p}^\prime}{\actionSend{\roleP}{\roleQ}{\msgM}{r}}\).

					From the definition of \grsnd{}, we have that
					\(\queueQ[c]^\prime = \pushQueue{\queueQ[c]}{\msgM}{p}{q}\).
				\item[Case \({\act[\alpha_c]} =
				\actionReceive{\roleP}{\roleQ}{\msgM}{r}\) :]
					From \cref{rmk:trace_invertible}, we deduce that \(\act[\alpha_c]\) is a
					\grrec{} transition with
					\(\redact{s_\participant{q}}{s_\participant{q}^\prime}{\actionReceive{\roleP}{\roleQ}{\msgM}{r}}\).

					From the definition of \grsnd{}, we have that
					\(\queueQ[c]^\prime = \popQueue{\queueQ[c]}{p}{q}\) and
					\(\nextElemQueue{\queueQ[c]}{p}{q} = \msgM\).
			\end{description}
			In both cases, \(\act[\alpha_c]\) ends up in \(\queueQ[c]^\prime\) w.r.t. \(\queueQ[c]\).
			Contradiction.
	\end{description}
	\paragraph*{Well-Predicated}
	By induction on the length of \(\trace\).
	\begin{description}
		\item [Case \(\trace\) is \(\tempty\):] the result trivially hold from \cref{def:WPT:i} in \cref{def:WPT}.
		\item [Case \(\trace\) is \(\tconcat{\tau^\prime}{\act[\alpha_n]}\):]
			the induction hypothesis is that \(\trace[\tau^\prime]\) is
			well-predicated by the empty map $\emptyMap$.

			Suppose that \(\trace[\tau^\prime]\) ends up with map \(\map{M_{\trace[\tau^\prime]}}\).
			By case analysis of the action \(\act[\alpha_n]\):
			\begin{description}
				\item [Case \({\act[\alpha_n]}\) is
				\(\actionReceive{p}{q}{\gtLbl{\ell},
				\tuple{\processVariable{x}, \val{c}}}{r}\) :]
					from \cref{rmk:trace_invertible}, this corresponds to a
					\grrec{} transition, which corresponds to a reduction:
					\[\tuple{\tuple{s_\participant{1}, \dots, s_\participant{p}, \dots, s_\participant{n}}, \queueQ, \map{M_{\trace[\tau^\prime]}}} \csred \tuple{\tuple{s_\participant{1}, \dots, s_\participant{p}^\prime, \dots, s_\participant{n}}, \popQueue{w}{q}{p}, \updateMap{M_{\trace[\tau^\prime]}}{x}{c}}\]
					with
					\(\redact{s_\participant{p}}{s_\participant{p}^\prime}{\actionReceive{p}{q}{\gtLbl{\ell},
					 \tuple{\processVariable{x}, \val{c}}}{r}}\)
					and
					\(\updateMap{M_{\trace[\tau^\prime]}}{x}{c}\models
					\refinement{r}\) (from the premises of \grrec{}).

					Since \(\trace[\tau^\prime]\) is a well-predicated trace, that ends up with \(\map{M_{\trace[\tau^\prime]}}\),
					from \cref{lem:concat_WP_traces} and the induction hypothesis,
					we simply have to show that \(\act[\alpha_n]\) is well-predicated under \(\map{M_{\trace[\tau^\prime]}}\).

					This holds directly from that
					\(\updateMap{M_{\trace[\tau^\prime]}}{x}{c}\models
					\refinement{r}\) (from above) and
					\cref{def:WPT:ii} in \cref{def:WPT}.
				\item [Case \({\act[\alpha_n]}\) is
				\(\actionSend{p}{q}{\gtLbl{\ell}, \tuple{\processVariable{x},
				\val{c}}}{r}\) :]
					from \cref{rmk:trace_invertible}, this corresponds to a
					\grsnd{} transition, which corresponds to a reduction:
					\[\tuple{\tuple{s_\participant{1}, \dots, s_\participant{p}, \dots, s_\participant{n}}, \queueQ, \map{M_{\trace[\tau^\prime]}}} \csred \tuple{\tuple{s_\participant{1}, \dots, s_\participant{p}^\prime, \dots, s_\participant{n}}, \pushQueue{w}{\msgM[\tuple{\gtLbl{\ell}, \tuple{\processVariable{x}, \val{c}}}]}{p}{q}, \updateMap{M_{\trace[\tau^\prime]}}{x}{c}}\]
					with
					\(\redact{s_\participant{p}}{s_\participant{p}^\prime}{\actionSend{p}{q}{\gtLbl{\ell},
					 \tuple{\processVariable{x}, \val{c}}}{r}}\)
					and
					\(\updateMap{M_{\trace[\tau^\prime]}}{x}{c}\models
					\refinement{r}\) (from the premises of \grsnd{}).

					Since \(\trace[\tau^\prime]\) is a well-predicated trace, that ends up with \(\map{M_{\trace[\tau^\prime]}}\),
					from \cref{lem:concat_WP_traces} and the induction hypothesis,
					we simply have to show that \(\act[\alpha_n]\) is well-predicated under \(\map{M_{\trace[\tau^\prime]}}\).

					This holds directly from that
					\(\updateMap{M_{\trace[\tau^\prime]}}{x}{c}\models
					\refinement{r}\) (from above) and
					\cref{def:WPT:ii} in \cref{def:WPT}.
					\qedhere
			\end{description}
	\end{description}
\end{proof}

\section{Decentralised Verification}

\subsection{Decentralised run of an RCS}

\begin{example}[Run of a Decentralised Configuration]
	We present a run of a decentralised configuration, which corresponds the
	run of the (centralised) configuration in \cref{ex:run_trace_CS}. Compared
	to that example, notice that there is no single global map, but instead
	each local state is associated with a local map (with $\queue{w_1}$,
	$\queue{w_2}$, $\queue{w_3}$ and $\map{M}$ as in
	\cref{ex:run_trace_CS}).\\[1mm]
	$
	\begin{array}{l}
		\small
		\tuple{\tuple{\participant{A_1},\emptyMap}, \tuple{\participant{B_1},
				\emptyMap}, \tuple{\participant{C_1}, \emptyMap}, \emptyqueue};
		\tuple{\tuple{\participant{A_2},\emptyMap}, \tuple{\participant{B_1},
				\emptyMap}, \tuple{\participant{C_1}, \emptyMap}, \queue{w_1}};
		\\
		\tuple{\tuple{\participant{A_2},\emptyMap}, \tuple{\participant{B_2},
				\map{\{\tuple{\processVariable{n}, \val{5}}\}}},
				\tuple{\participant{C_1},
				\emptyMap}, \emptyqueue};
		\tuple{\tuple{\participant{A_2},\emptyMap}, \tuple{\participant{B_2},
				\map{\{\tuple{\processVariable{n}, \val{5}}\}}},
				\tuple{\participant{C_2},
				\emptyMap}, \queue{w_2}};\\
		\tuple{\tuple{\participant{A_2},\emptyMap}, \tuple{\participant{B_3},
		\map{M}},
			\tuple{\participant{C_2}, \emptyMap}, \emptyqueue};
		\tuple{\tuple{\participant{A_2},\emptyMap}, \tuple{\participant{B_4},
		\map{M}},
			\tuple{\participant{C_2}, \emptyMap},
			\queue{w_3}};\\
		\tuple{\tuple{\participant{A_2},\emptyMap}, \tuple{\participant{B_4},
		\map{M}},
			\tuple{\participant{C_3}, \emptyMap}, \emptyqueue};
	\end{array}
	$
\end{example}

\subsection{Simulation\label{sec:app:simulation}}

\begin{definition}[Simulation \citet{SangiorgiIntroduction2012}]
	Given two labelled transition systems with states taken from \(S_1\) and
	\(S_2\), a relation \(\mathcal{R} \subseteq S_1 \times S_2\) is a
	simulation if:
	\(\forall s_1, s_1^\prime \in S_1\), if \(s_1 \rightarrow s_1^\prime\),
	then \(\forall s_2 \in S_2\) such that
	\(\tuple{s_1, s_2} \in \mathcal{R}\), there exists \(s_2^\prime\) such that
	\(s_2 \rightarrow s_2^\prime\) and \(\tuple{s_1^\prime, s_2^\prime} \in
	\mathcal{R}\).
\end{definition}

\begin{remark}[Bisimulation]
	A relation $\mathcal{R}$ is a \emph{bisimulation} if
	$\mathcal{R}$ and
	$\mathcal{R}^{-1}$ are simulations.
\end{remark}

\subsection{Main theorem\label{sec:app:centralised_simulated_decentralised}}
\centralisedSimulatesDecentralised
\begin{proof}
	The proof is quite straightforward, we match each
	\drec{} by \grrec{} and each \dsnd{}
	by \grsnd{}. We can trivially note that the
	queue of messages is the same in both
	\(\DConfOfType{G}\) and \(\ConfOfType{G}\).
	Similarly, each \(s_\participant{i}\) match in both systems,
	and remain matched after taking transitions.

	Finally, the map of \(\ConfOfType{G}\) is always a
	superset of the union of all \(\map{M}_\participant{i}\) in
	\(\DConfOfType{G}\).

	We show the proof for \dsnd{} and \grsnd{}. The proof for
	\drec{} and \grrec{} is similar.

	Suppose that we have \(D =
	\tuple{\tuple{\dots, \tuple{s_\participant{i}, \map{M}_\participant{i}}
			\dots}, \queueQ}\)
	reducing to
	\(D^\prime =
	\tuple{\tuple{\dots, \tuple{s_\participant{i}^\prime,
				\map{M_\participant{i}^\prime}} \dots}, \queue{w^\prime}}\)
	with \dsnd{} and let \(\participant{j}\) be the destination
	of the message.
	Suppose \(S =
	\tuple{\tuple{\dots, s_\participant{i}, \dots}, \queueQ,
		\map{M}}\), such that \(\map{M} \subseteq
	\biguplus_{\participant{i}\in\rolesOfGT{G}}
	\map{M_\participant{i}}\).

	Notice that the premises of \dsnd{} entails the
	premises of \grsnd{}. Therefore, \grsnd{} can be
	triggered: let
	\(S^\prime = \tuple{\tuple{\dots, s_\participant{i}^\prime, \dots},
		\queue{w^\prime}, \map{M^\prime}}\) be the resulting configuration.
	To prove the simulation, we just have to show that
	the resulting \(s_\participant{i}^\prime\) of both configurations are
	the same, which is direct from the rule; that the queues
	\(\queue{w^\prime}\) of both configurations
	are the same, which is also direct from the rule; and that
	\(\map{M^\prime}\) is indeed a subset of
	\(\biguplus_{\participant{i}\in\rolesOfGT{G}}\map{M_\participant{i}^\prime}\).

	Let \(\processVariable{x}\) be the variable sent and \(\val{c}\)
	its associated value. By definition
	\(\map{M_\participant{i}^\prime} =
	\removeMap{\map{M_\participant{i}}}{x}\), and all other local
	maps are unchanged. Also, by hypothesis (no duplication
	hypothesis), no other map contains
	\(\processVariable{x}\). Therefore, \(\map{M^\prime}
	= \updateMap{M}{x}{c}
	\varsupsetneq \removeMap{\map{M}}{x}
	\supseteq
	\removeMap{\color{black}(\biguplus_{\participant{i}\in\rolesOfGT{G}}\map{M_\participant{i}})}{x}
	= \biguplus_{\participant{i}\in\rolesOfGT{G}}\map{M_\participant{i}^\prime}
	\).
	Therefore, \(\ConfOfType{G}\) simulates \(\DConfOfType{G}\).
\end{proof}

\subsection{Static Verification of the Two Conditions\label{sec:static_verif_conditions}}

In the previous section, we introduced two conditions for the
refined configurations to simulate decentralised configurations
(\cref{def:cond_dyn_verif}). We now aim to statically verify whether those two
conditions hold for a given type. The algorithm we present keeps track of
variable moves and tries to find which participant has a copy of which variable
at any point in the execution of the protocol, which we call \emph{variable
localisation}.

Intuitively, from a global type \(\gtFmt{G}\), it is possible to infer
constraints on variable localisation: for instance, given the global type
\(\gtFmt{G} = \gtComm{A}{B}{\ell}[x]{int}[x \geq 0]{\gtEnd}\), we can infer that
participant \(\participant{A}\) has \(\processVariable{x}\) before the
communication, and that \(\processVariable{x}\) is transfered to
\(\participant{B}\) during the communication, therefore \(\processVariable{x}\)
is at \(\participant{B}\) after the communication. The algorithm we propose
generates such constraints to find variables localisations at any time.
We represent RMPST as graphs (\cref{sec:type_graph}), and the algorithm finds an assignation of
variables that is consistent with the actions of the type. Finding such an
assignation is trivially solved with logic programming
(\cref{sec:localisation_algo}). The only non-trivial part is to carefully
pinpoint where variables are first used in loops. We illustrate this problem and
solve it by unrolling loops in the graph of the type, which we present in
\cref{sec:unroll_mpst}.

\subsubsection{Type Graphs\label{sec:type_graph}}

First, we need to create the \emph{graph of a global type}. This is analogous to
the CFSM of a local type, but for global types instead.

\begin{definition}[Graph of a Global Type\label{def:global_type_graph}]
	Given a global type \(\gtFmt{T_0}\), the \emph{graph} \(\GG{T_0}\) of that type is
	the labeled graph \(\tuple{V, E}\), where
	\(V = \{\gtFmt{T^\prime} | \gtFmt{T^\prime}\in \gtFmt{T_0}\wedge \gtFmt{T^\prime}\neq\gtRecVar{t} \wedge \gtFmt{T^\prime}\neq\gtRec{t}{T_\mu}\}\),
	\(E = \delta_g\),
	and where \(\delta_g\) is defined as the smallest relation such that%
	\footnote{Notice that \(\delta_g\) is almost like \(\delta\), but we erase the value in the messages' payloads},
	for all \(\gtFmt{T}\in V\), with \(\gtFmt{T} =
	\gtComm[i][I]{p}{q}{\ell}[x]{S}[r]{T}\):
	\begin{inlineenum}
		\item if \(\operatorname{strip}(\gSessionType{T_i})\) is a
			communication or a termination,
			then, for all \(i\in I\),
			\(\tuple{\gSessionType{T}, \tuple{\participant{p}, \participant{q},
			\processVariable{x_i}, \fv{r_i}},
			\operatorname{strip}(\gSessionType{T_i}})\in\delta_g\)
		\item if \(\operatorname{strip}(\gSessionType{T_i}) = \gtRecVar{t}\) with
			\(\gtRec{t}{T^\prime}\in \gSessionType{T_0}\),
			then \(\tuple{\gSessionType{T}, \tuple{\participant{p},
			\participant{q}, \processVariable{x_i}, \fv{r_i}},
			\operatorname{strip}(\gSessionType{T^\prime})}\in\delta_g\)
	\end{inlineenum}
\end{definition}
Where \(\operatorname{strip}\) is defined on global types as on local types.

\begin{example}[Graph of a global type]
	The type graph of \(\gtFmt{G_\pm}\) is shown below. Notice that the two
	transitions with \(\gtLbl{more}\) and \(\gtLbl{less}\)
	result in a single edge in the graph, as we erase the label, and both
	lead to the same state.
	
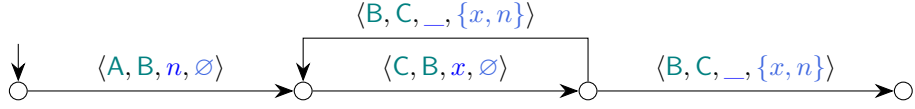
\begin{center}
\begin{tikzpicture}[node distance=4.2cm]
	\draw node[draw, rounded rectangle] (G1) {};
	\draw node[draw, rounded rectangle, right=3.5cm of G1] (G2) {};
	\draw node[draw, rounded rectangle, right=3.5cm of G2] (G3) {};
	\draw node[draw, rounded rectangle, right=3.9cm of G3] (G4) {};
	\coordinate [above=.5cm of G1] (start);
	\draw[-{Stealth[scale=1.5]}] (start) -- (G1);

	\draw[-{Stealth[scale=1.5]}] (G1) to node[anchor=south, pos=.5, sloped] {\(\tuple{\participant{A}, \participant{B}, \processVariable{n}, \refinement{\emptyset}}\)} (G2);
	\draw[-{Stealth[scale=1.5]}] (G2) to node[anchor=south, pos=.5, sloped] {\(\tuple{\participant{C}, \participant{B}, \processVariable{x}, \refinement{\emptyset}}\)} (G3);

	\coordinate (midG23) at ($ (G3)!.5!(G2) $);
	\coordinate [above=.7cm of midG23] (midG23up);
	\coordinate [below=.7cm of midG23] (midG23down);

	\draw[-{Stealth[scale=1.5]}]
		(G3) |- (midG23up) node[anchor=south] {\(\tuple{\participant{B}, \participant{C}, \processVariable{\_}, \refinement{\{x, n\}}}\)} -|  (G2);
	\draw[-{Stealth[scale=1.5]}] (G3) to node[anchor=south, pos=.5, sloped] {\(\tuple{\participant{B}, \participant{C}, \processVariable{\_}, \refinement{\{x, n\}}}\)} (G4);
\end{tikzpicture}
\qedhere
	\end{center}
\end{example}

Among the vertices of the graph type \(\GG{T_0}\) of a global type
\(\gSessionType{T_0}\), we distinguish the \emph{initial} node, which is the
vertice labeled \(\gSessionType{T_0}\). In the graphical representation, this
state is shown with an arrow.

\subsubsection{Localising Algorithm\label{sec:localisation_algo}}

We now present the \emph{localisation algorithm}. The goal of this algorithm is
to find whether there exists a state, in the execution of the protocol specified
by a type \(\gtFmt{G}\), where a variable is duplicated, and whether there
exists an action which refinement requires a variable that is not locally
available.

Our algorithm infers the location of each variable (i.e. which participant has
access to which variables), at each step of the protocol. We infer such location information from the
actions of the protocol. This algorithm is expressed as a set of inference
rules, which can be directly encoded in a logic programming language as DataLog.

The input of our algorithm is given by two provided atoms that are extracted from the global type graph (we show provided atoms in green and computed atoms in orange):
\begin{itemize}
	\item \PLAtomGiven{Send}{s_1, \participant{p}, \processVariable{x}, \participant{q}, s_2}
		holds when the graph contains an edge from \(s_1\) to \(s_2\)
		with label \(\) \(\tuple{\participant{p}, \participant{q}, \processVariable{x}, \_}\);
	\item \PLAtomGiven{FVRefinement}{s_1, \processVariable{x}, s_2}
		holds when there is an edge from \(s_1\) to \(s_2\) with a label
		\(\tuple{\_, \_, \_, \fv{r}}\) and \(\fv{r}\) contains
		\(\processVariable{x}\).
\end{itemize}

The core part of our algorithm is the atom \PLAtom{In}{s, \participant{p}, \processVariable{x}},
which holds if \(\participant{p}\) has access to variable
\(\processVariable{x}\) in state \(s\) (or, equivalently, if
\(\processVariable{x}\) is at \(\participant{p}\) in \(s\)).

The first rule simply capture the fact that, if a participant
\participant{p} sends a variable \(\varX\) to \participant{q}, then \(\varX\) is
located at \participant{q} after the exchange (notice that \(\varX\) is
\emph{not} necessarily available at \(\participant{p}\) before the exchange, as
it might be the first time it appears):
\PLAtom{In}{s_2, \participant{q}, \processVariable{x}} \(\rightarrow\) \PLAtomGiven{Send}{\_, \_, \processVariable{x}, \participant{q}, s_2}.

We then define two rules that infer which variables are available at which
location after a communication. When
\(\participant{p}\) sends \(\processVariable{x}\), all
\(\processVariable{y}\neq\processVariable{x}\)
are kept at \(\participant{p}\). For the other participants
\(\participant{r}\neq\participant{p}\), all variables are
preserved.

\begin{minipage}{.5\textwidth}
\begin{prolog}{\textwidth}
	\PLRule{\PLAtom{In}{s_2, \participant{p}, \processVariable{x}}}{\PLAtom{In}{s_1, \participant{p}, \processVariable{x}}\PLAnd
	\PLAtomGiven{Send}{s_1, \participant{p}, \processVariable{y}, \_, s_2}\PLAnd
	(\(\processVariable{x}\neq \processVariable{y}\))}
\end{prolog}
\end{minipage}%
\begin{minipage}{.5\textwidth}
\begin{prolog}{\textwidth}
	\PLRule{\PLAtom{In}{s_2, \participant{r}, \processVariable{x}}}{\PLAtom{In}{s_1, \participant{r},\processVariable{x}}\PLAnd
	\PLAtom{Send}{s_1, \participant{p}, \_, \_, s_2}\PLAnd
	\((\participant{p}\neq\participant{r})\)}
\end{prolog}
\end{minipage}%

Once variables are localised, we can check whether the two conditions hold.
First, the atom \PLAtomName{NotVerifFV} checks whether there is a state
\(s\) in which \(\participant{p}\) sends a message with a refinement that
contains a free variable \(\processVariable{x}\) which \(\participant{p}\)
cannot access (that is a variable \(\processVariable{x}\) that is neither
located at \(\participant{p}\) nor in the message sent). If that happens, the
system is not verifiable, as \(\participant{p}\) will not be able to check its
refinement when taking the transition. Second, the atom
\PLAtomName{NotVerifDupl} checks whether a variable is duplicated. It simply
records whether there is a state \(s\) in which two distinct participants
\(\participant{p}\) and \(\participant{q}\) can access the same variable
\(\processVariable{x}\).

\begin{minipage}{.5\textwidth}
\begin{prolog}{\textwidth}
	\PLRule{\PLAtom{NotVerifFV}{s, s_2, \processVariable{x},\participant{p}}}{\PLAtomGiven{FVRefinement}{s, \processVariable{x}, s_2}\PLAnd
	\PLAtomGiven{Send}{s, \participant{p}, \processVariable{y}, \_, \_}\PLAnd
	!\PLAtom{In}{s, \participant{p}, \processVariable{x}}\PLAnd
	\((\processVariable{x}\neq\processVariable{y})\)}
\end{prolog}
\end{minipage}%
\begin{minipage}{.5\textwidth}
\begin{prolog}{\textwidth}
	\PLRule{\PLAtom{NotVerifDupl}{s, \processVariable{x}, \participant{p}, \participant{q}}}{
		\PLAtom{In}{s, \participant{p}, \processVariable{x}}\PLAnd
		\PLAtom{In}{s, \participant{q}, \processVariable{x}}\PLAnd
		(\(\participant{p}\neq \participant{q}\))}
\end{prolog}
\end{minipage}

Finally, the two conditions hold if there is no \PLAtomName{NotVerifFV} nor
\PLAtomName{NotVerifDupl}.

\subsubsection{Recursion Unrolling\label{sec:unroll_mpst}}
The localisation algorithm does not consider the case where a variable is first sent
in a loop. Consider the type:
\[\gtFmt{\gtRec{T}{\gtCommRaw{A}{B}{
		\begin{aligned}
			&\commChoice{\ell_1}[x]{int}[x = y]{\gtComm{B}{A}{\ell_3}[x]{int}{\gtRecVar{T}}}\\
			&\commChoice{\ell_2}[y]{int}[x \neq y]{\gtComm{B}{A}{\ell_4}[y]{int}{\gtRecVar{T}}}\\
		\end{aligned}
}}}\]

This type cannot be verified, as if \participant{A} chooses the branch
\(\gtLbl{\ell_1}\) on the first iteration of the loop, then
\(\processVariable{y}\) is not defined, and similarly for the branch
\(\gtLbl{\ell_2}\). However, our localisation algorithm first computes the
location of variables (\(\PLAtomName{In}\)), which results in the following
graph and variable locations:
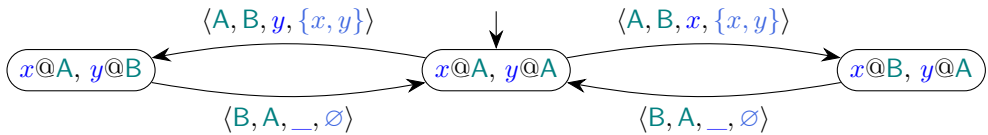
\begin{center}
\begin{tikzpicture}[node distance=4.2cm]
	\draw node[draw, rounded rectangle] (G1) {\(\processVariable{x}@\participant{A}\), \(\processVariable{y}@\participant{A}\)};
	\draw node[draw, rounded rectangle, right=3.5cm of G1] (G2) {\(\processVariable{x}@\participant{B}\), \(\processVariable{y}@\participant{A}\)};
	\draw node[draw, rounded rectangle, left=3.5cm of G1] (G3) {\(\processVariable{x}@\participant{A}\), \(\processVariable{y}@\participant{B}\)};
	\coordinate [above=.5cm of G1] (start);
	\draw[-{Stealth[scale=1.5]}] (start) -- (G1);

	\draw[-{Stealth[scale=1.5]}] (G1) to[bend left=10] node[anchor=south, pos=.5, sloped]
	{\(\tuple{\participant{A}, \participant{B}, \processVariable{x}, \refinement{\{x, y\}}}\)} (G2);
	\draw[-{Stealth[scale=1.5]}]
		(G2) to[bend left=10] node[anchor=north] {\(\tuple{\participant{B}, \participant{A}, \processVariable{\_}, \refinement{\emptyset}}\)} (G1);

	\draw[-{Stealth[scale=1.5]}] (G1) to[bend right=10] node[anchor=south, pos=.5, sloped]
	{\(\tuple{\participant{A}, \participant{B}, \processVariable{y}, \refinement{\{x, y\}}}\)} (G3);
	\draw[-{Stealth[scale=1.5]}]
		(G3) to[bend right=10] node[anchor=north] {\(\tuple{\participant{B}, \participant{A}, \processVariable{\_}, \refinement{\emptyset}}\)} (G1);
\end{tikzpicture}
\end{center}
Therefore, we have to distinguish the first time a branch is taken from the
following iterations. We therefore introduce our \emph{unrolling} algorithm,
which unrolls each loop once to distinguish the first iteration from the
following ones. This task is non trivial, as it is not a simple syntactical
replacement of a recursion variable by its definition since we have to take
into account all branches. Specifically, in the example above, \participant{A} can possibly
take the branch \(\gtLbl{\ell_1}\) multiple times before choosing
\(\gtLbl{\ell_2}\): we have to take all declarations orders into account.

For the sake of simplicity, consider the type \(\gtFmt{G^\prime}\) (which is not
localisable, but greatly reduces the size of the graph and is sufficient to
explain our algorithm):
\(\gtRec{T}{\gtCommRaw{A}{B}{
		\commChoice{\ell_1}[x]{int}{\gtRecVar{T}};
		\commChoice{\ell_2}[y]{int}{\gtRecVar{T}}
}}\).
The different steps of the execution of our algorithm are shown
in \cref{fig:graph_G_prime}.

We first distinguish forward and backward edges in the type graph: a backward
arrow is an arrow that contains a recursion, i.e.\ it ends in a previous
state or in the same state. We initially mark all those backward edges as
unvisited. In \cref{fig:graph_G_prime}, those unvisited backward arrows are shown
in red. We then unroll those unvisited backward edges one at a time.
For each unvisited backward edge \(E\) from \(N\) to \(N^\prime\), we copy the
graph, we mark \(E\) as visited in the copied graph and we
replace \(E\) by a forward visited edge from \(N\) in the original graph to
\(N^\prime\) in the copied graph.

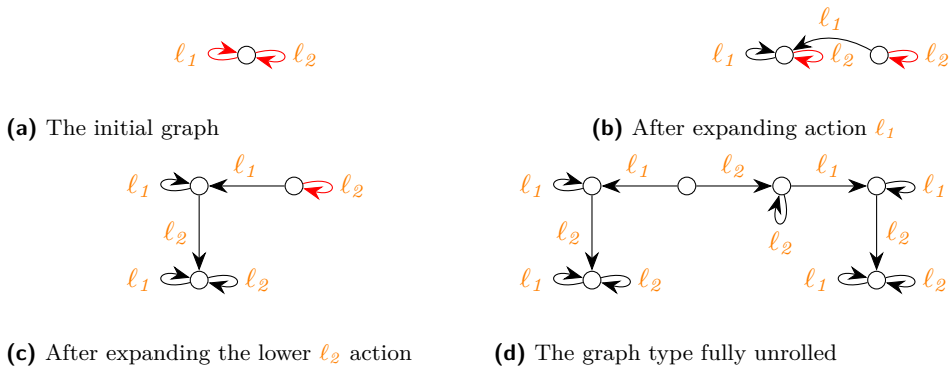
\begin{figure}
\begin{subfigure}{.45\textwidth}
\begin{center}
\begin{tikzpicture}[]
		\draw node[draw, rounded rectangle] (G1) {};

		\draw[-{Stealth[scale=1.5]}, color=red] (G1) to[loop left,] node[anchor=east, pos=.5] {\(\gtLbl{\ell_1}\)} (G1);
		\draw[-{Stealth[scale=1.5]}, color=red] (G1) to[loop right,] node[anchor=west, pos=.5] {\(\gtLbl{\ell_2}\)} (G1);

\end{tikzpicture}
\end{center}
	\caption{The initial graph}
\end{subfigure}
\begin{subfigure}{.45\textwidth}
\begin{center}
\begin{tikzpicture}[]
		\draw node[draw, rounded rectangle] (G2) {};
		\draw node[draw, rounded rectangle, left=of G2] (G3) {};

		\draw[-{Stealth[scale=1.5]}, ] (G2) to[bend right] node[anchor=south, pos=.5] {\(\gtLbl{\ell_1}\)} (G3);
		\draw[-{Stealth[scale=1.5]}, color=red] (G2) to[loop right,] node[anchor=west, pos=.5] {\(\gtLbl{\ell_2}\)} (G2);
		\draw[-{Stealth[scale=1.5]}, ] (G3) to[loop left,] node[anchor=east, pos=.5] {\(\gtLbl{\ell_1}\)} (G3);
		\draw[-{Stealth[scale=1.5]}, color=red] (G3) to[loop right,] node[anchor=west, pos=.5] {\(\gtLbl{\ell_2}\)} (G3);
\end{tikzpicture}
\end{center}
	\caption{After expanding action \(\gtLbl{\ell_1}\)}
\end{subfigure}
\begin{subfigure}{.45\textwidth}
\begin{center}
\begin{tikzpicture}[]
		\draw node[draw, rounded rectangle] (G4) {};
		\draw node[draw, rounded rectangle, left=of G4] (G5) {};
		\draw node[draw, rounded rectangle, below =of G5] (G6) {};

		\draw[-{Stealth[scale=1.5]}, ] (G4) to[] node[anchor=south, pos=.5] {\(\gtLbl{\ell_1}\)} (G5);
		\draw[-{Stealth[scale=1.5]}, color=red] (G4) to[loop right,] node[anchor=west, pos=.5] {\(\gtLbl{\ell_2}\)} (G4);
		\draw[-{Stealth[scale=1.5]}, ] (G5) to[loop left,] node[anchor=east, pos=.5] {\(\gtLbl{\ell_1}\)} (G5);
		\draw[-{Stealth[scale=1.5]}, ] (G5) to[] node[anchor=east, pos=.5] {\(\gtLbl{\ell_2}\)} (G6);
		\draw[-{Stealth[scale=1.5]}, ] (G6) to[loop left,] node[anchor=east, pos=.5] {\(\gtLbl{\ell_1}\)} (G6);
		\draw[-{Stealth[scale=1.5]}, ] (G6) to[loop right,] node[anchor=west, pos=.5] {\(\gtLbl{\ell_2}\)} (G6);
\end{tikzpicture}
\end{center}
	\caption{After expanding the lower \(\gtLbl{\ell_2}\) action}
\end{subfigure}
\begin{subfigure}{.45\textwidth}
\begin{center}
\begin{tikzpicture}[]
		\draw node[draw, rounded rectangle] (G7) {};
		\draw node[draw, rounded rectangle, left=of G7] (G8) {};
		\draw node[draw, rounded rectangle, below =of G8] (G9) {};
		\draw node[draw, rounded rectangle, right=of G7] (G10) {};
		\draw node[draw, rounded rectangle, right=of G10] (G11) {};
		\draw node[draw, rounded rectangle, below =of G11] (G12) {};

		\draw[-{Stealth[scale=1.5]}, ] (G7) to[] node[anchor=south, pos=.5] {\(\gtLbl{\ell_1}\)} (G8);
		\draw[-{Stealth[scale=1.5]}, ] (G7) to[] node[anchor=south, pos=.5] {\(\gtLbl{\ell_2}\)} (G10);
		\draw[-{Stealth[scale=1.5]}, ] (G8) to[loop left,] node[anchor=east, pos=.5] {\(\gtLbl{\ell_1}\)} (G8);
		\draw[-{Stealth[scale=1.5]}, ] (G8) to[] node[anchor=east, pos=.5] {\(\gtLbl{\ell_2}\)} (G9);
		\draw[-{Stealth[scale=1.5]}, ] (G9) to[loop left,] node[anchor=east, pos=.5] {\(\gtLbl{\ell_1}\)} (G9);
		\draw[-{Stealth[scale=1.5]}, ] (G9) to[loop right,] node[anchor=west, pos=.5] {\(\gtLbl{\ell_2}\)} (G9);

		\draw[-{Stealth[scale=1.5]}, ] (G10) to[] node[anchor=south, pos=.5] {\(\gtLbl{\ell_1}\)} (G11);
		\draw[-{Stealth[scale=1.5]}, ] (G10) to[loop below,] node[anchor=north, pos=.5] {\(\gtLbl{\ell_2}\)} (G10);
		\draw[-{Stealth[scale=1.5]}, ] (G11) to[loop right,] node[anchor=west, pos=.5] {\(\gtLbl{\ell_1}\)} (G11);
		\draw[-{Stealth[scale=1.5]}, ] (G11) to[] node[anchor=west, pos=.5] {\(\gtLbl{\ell_2}\)} (G12);
		\draw[-{Stealth[scale=1.5]}, ] (G12) to[loop left,] node[anchor=east, pos=.5] {\(\gtLbl{\ell_1}\)} (G12);
		\draw[-{Stealth[scale=1.5]}, ] (G12) to[loop right,] node[anchor=west, pos=.5] {\(\gtLbl{\ell_2}\)} (G12);
\end{tikzpicture}
\end{center}
	\caption{The graph type fully unrolled}
\end{subfigure}
\caption{The graph of \(\gSessionType{G^\prime}\), with \(\text{unvisited}\)
edges in red. For the sake of simplicity, we only use \(\gtLbl{\ell_1}\) and \(\gtLbl{\ell_2}\) as
the two communication labels. Notice that the unrolled graph is not the
smallest (in that some loops might be unrolled more than needed), but still
unrolls every loop.}
\label{fig:graph_G_prime}
\end{figure}

\vspace{.2cm}

\noindent
\begin{minipage}{.45\textwidth}
	\makeatletter
	\let\@latex@error\@gobble
	\makeatother
	\begin{center}
\begin{algorithm}[H]
	\KwData{\(G\) -- A type graph}
	\While{\(\exists e = \tuple{v_1, v_2} \in \set{B}_G\cap\set{U}_G\)}{
		Mark \(e\) as \(\text{visited}\) in \(G\)\;
		Let \(G^\prime\) be a copy of the continuation of \(v_1\) using \(e\) in \(G\)\;
		Let \(v_2^\prime\) be the copy of \(v_2\) in \(G^\prime\)\;
		Remove \(e\) from \(G\)\;
		Let \(G\) be the union of \(G\) and \(G^\prime\)\;
		Add an edge from \(v_1\) to \(v_2^\prime\) in \(G\)\;
	}
	Remove unreachable vertices\;
\end{algorithm}
	\end{center}
\end{minipage}
\begin{minipage}{.55\textwidth}
	\begin{description}
		\item[Depth of a node:]
			The depth of a node \(N\) is the length of the shortest
			path from the initial node to \(N\).
		\item[Backward edge:]
			An edge from \(N_1\) to \(N_2\) is said \emph{backward}
			if the depth of \(N_1\) is greater than or equal to the
			depth of \(N_2\). We note \(\set{B}_G\) the set of
			backward edges of graph \(G\).
		\item[State of an edge:]
			We assign each edge a state taken from
			\(\{\text{visited}, \text{unvisited}\}\). We note
			\(\set{U}_G\) the set of \(\text{unvisited}\) edges of
			graph \(G\).
		\item[Continuation:]
			The continuation of a node \(N\) following an edge
			\(E\) is the (sub)graph reachable from the destination
			of \(E\), provided that the origin of \(E\) is \(N\).
	\end{description}
\end{minipage}

\section{Static Elision of Redundant Refinements}

\subsection{Static Elision of Refinements in RCS}
\label{sec:app:static_elision_RCS}

\subsubsection{Definitions.}

\begin{definition}[Independent transitions]
	\label{def:indenpendent_trans}
	A transition
	$t = \redact{\sigma}{\sigma^\prime}
	{\ltsLabelGeneric{p}{q}{\gtLbl{\ell},\tuple{\processVariable{x},\val{\_}}}[r]}$
	\emph{depends} on $\fv{r}$, the free variables of its refinement.

	We say that $t$ \emph{depends} on another transition
	$t^\prime = \redact{\rho}{\rho^\prime}
	{\ltsLabelGeneric{r}{s}{\gtLbl{\ell^\prime},\tuple{\processVariable{y},\val{\_}}}[r^\prime]}$
	if $t$ depends on $\processVariable{y}$, the payload of $t^\prime$.
	Otherwise, we say $t$ is independent of $t^\prime$. When $t$ depends
	(resp.\ does not depend) on itself, we say it is \emph{self-dependent}
	(resp.\ \emph{self-independent}).
\end{definition}

\begin{definition}[Well-defined transitions]
	\label{def:static:well_defined}
	Given a RCS with a CFSM containing a transition $t =
	\redact{s_\participant{i}}{s_\participant{i}^\prime}{\actionGeneric{\_}{\_}{\gtLbl{\_},\_}{r}}$,
	we say $t$ is \emph{well-defined} if and only if, in all reachable states
	$\tuple{\tuple{\dots, s_\participant{i}, \dots}, \queue{\_}, \map{M}}$,
	$\fv{r}\subseteq\domMap{M}$.
\end{definition}

\subsubsection{Lemmas and proofs.}

\begin{lemma}
	\label{thm:indep_trans}
	Let
	\(t_1 = \redact{\sigma}
	{\sigma^\prime}
	{
		\actionGeneric{p}{q}{\gtLbl{\ell},\tuple{\processVariable{x},\val{c}}}{r}
	}
	\) and
	\(
	t_2 =
	\redact{\rho}
	{\rho^\prime}
	{
		\actionGeneric{r}{s}{\gtLbl{\ell^\prime},\tuple{\processVariable{y},\val{c^\prime}}}{r^\prime}
	}
	\) such that $t_1$ is independent from $t_2$.

	For all $\map{M}$, $\map{M_2}$ such that \(
	\tuple{\_, \queue{\_}, \map{M}}
	\csred \tuple{\_, \queue{\_}, \map{M_2}}
	\) with $t_2$; $\map{M}\models\refinement{r}$ if and only if
	$\map{M_2}\models\refinement{r}$.
\end{lemma}
\begin{proof}[Proof of \cref{thm:indep_trans}.]
	First, $t_1$ is independent of $t_2$, i.e.\
	(\cref{def:indenpendent_trans}), $t_1$ does not depend on
	$\processVariable{y}$; that is
	$\processVariable{y}\not\in\fv{r}$. Also, since \(
	\tuple{\_, \queue{\_}, \map{M}}
	\csred \tuple{\_, \queue{\_}, \map{M_2}}
	\), then $\map{M_2} = \updateMap{M}{y}{c^\prime}$.

	In addition, by definition (\cref{sec:pred_lang_sem}),
	$\map{M}\models \refinement{r}$ if and only if:
	\begin{inlineenum}
		\item \(\fv{r} \subseteq \domMap{M}\)
		\item \(\refeval{r\subst{\fv{r}}{\lookupMap{M}{\fv{r}}}}=\top\).
	\end{inlineenum}

	Therefore, $\refinement{r}\subst{\fv{r}}{\lookupMap{M}{\fv{r}}} =
	\refinement{r}\subst{\fv{r}}{\lookupMap{M_2}{\fv{r}}}$ and
	$\fv{r}\subseteq\domMap{M}$ if and only if
	$\fv{r}\subseteq\domMap{M}\cup\{\processVariable{y}\}=\domMap{M_2}$.

	Therefore, $\map{M}\models \refinement{r}$ if and only if
	$\map{M_2}\models \refinement{r}$.
\end{proof}

\begin{lemma}
	\label{thm:static:reachable_states_dom_refinement}
	Given an RCS $R$ containing an RCFSM with a transition
	$\redact{s_\participant{i}}{s_\participant{i}^\prime}
	{\actionGeneric{\_}{\_}{\_}{r}}$, and such that for each transition
	$\redact{\_}{\_}{\actionSend{\_}{\_}{\_}{r_w}}
	\in\bigcup_{\processVariable{x}\in\fv{r}}\set{T}_{\processVariable{x}}$,
	for all map $\map{M}$, $\map{M}\models \refinement{r_w}$ entails
	$\map{M}\models\refinement{r}$.

	For all reachable states
	$\sigma=\tuple{\tuple{\vec{s_\participant{i}}}, \queue{w}, \map{M}}$
	of $R$, if $\fv{r}\subseteq\domMap{M}$, then
	$\map{M}\models\refinement{r}$.
\end{lemma}
\begin{proof}
	Since $\sigma$ is reachable, there is a run $\sigma_0; \dots; \sigma;
	\dots$ that contains $\sigma$.
	By induction on the run.
	\begin{description}
		\item [Base case ($\sigma$ is initial):] Since $\sigma$ is an
		initial state, then $\map{M} = \emptyMap$, our claim vacuously
		holds.
		\item [Inductive case
		({ $\sigma^\prime \csred[t]\sigma$ with
			$t=\redact{\_}{\_}{\actionGeneric{\_}{\_}{\gtLbl{\_},\tuple{\processVariable{x},\val{c_x}}}{r^\prime}}$
		})]
		Let $\sigma^\prime =
		\tuple{\tuple{\vec{s_\participant{i}^\prime}},
			\queue{w^\prime}, \map{M^\prime}}$.
		The induction hypothesis is that if
		$\fv{r}\subseteq\map{M^\prime}$, then
		$\map{M^\prime}\models\refinement{r}$.
		Also, $\map{M}\models\refinement{r^\prime}$, from the premises
		of \grrec{} or \grsnd{}.
		By case analysis on the transition $t$:
		\begin{description}
			\item[If $\processVariable{x}\not\in\fv{r}$:] In that case,
			$\fv{r}\subseteq\map{M}$ if and only if
			$\fv{r}\subseteq\map{M^\prime}$. The conclusion holds
			directly from the induction hypothesis.
			\item[If $\processVariable{x}\in\fv{r}$:]
			In that case,
			$t\in\bigcup_{\processVariable{x}\in\fv{r}}\set{T}_{\processVariable{x}}$.
			We distinguish the case of send and receive:
			\begin{description}
				\item[Case receive] In that case, from
				\cref{thm:run_send_before_receive}, there are
				$\sigma_{j-1}; \sigma_j$ in $\sigma_0; \dots; \sigma_{i-1}$
				such that
				$\sigma_{j-1}
				\csred[\redact{\_}{\_}{\actionSend{\_}{\_}{\gtLbl{\_},\tuple{\processVariable{x},\val{c_x}}}{r^\prime}}]
				\sigma_j$, therefore $\processVariable{x}\in\domMap{M^\prime}$,
				therefore $\fv{r}\subseteq\domMap{M}$ if and only if
				$\fv{r}\subseteq\domMap{M^\prime}$.
				Let $\map{M_j}$ be the map of $\sigma_j$.
				From the induction hypothesis, if
				$\fv{r}\subseteq\domMap{M^\prime}$, then
				$\map{M^\prime}\models \refinement{r}$,
				therefore if $\fv{r}\subseteq\domMap{M}$, then
				$\map{M^\prime}\models\refinement{r}$.
				\begin{description}
					\item[If
					$\forall\processVariable{y}\in\fv{r}\suchthat\lookupMap{M}{y}=\lookupMap{M^\prime}{y}$:]
					From the rule \grrec{}, we therefore have that $\map{M} =
					\map{M^\prime}$, therefore $\fv{r}\subseteq\domMap{M}$
					implies $\map{M}\models\refinement{r}$.
					\item[If
					$\exists\processVariable{y}\in\fv{r}\suchthat\lookupMap{M}{y}\neq\lookupMap{M^\prime}{y}$:]
					Without loss of generality, suppose there is a single such
					$\processVariable{y}$. In that case, there is at least one
					$\sigma_k$ ($k > j$, $k
					<i$)
					such that
					$\sigma_{k-1}
					\csred[\redact{\_}{\_}{\actionGeneric{\_}{\_}{\gtLbl{\_},\tuple{\processVariable{y},\val{c_y}}}{r_y}}]
					\sigma_k$.
					Without loss of generality, consider there is a single such
					$\sigma_k$. Let $\map{M_k}$ and $\map{M_{k-1}}$ be the map
					of $\sigma_k$ and $\sigma_{k-1}$ respectively.
					\begin{description}
						\item[If $\processVariable{x} = \processVariable{y}$:]
						in that case:
						\begin{inlineenum}
							\item $\lookupMap{M}{x} = \lookupMap{M_j}{x} =
							\val{c_x}$
							\item for all $\processVariable{v}\in\fv{r}$ such
							that $\processVariable{v}\neq\processVariable{x}$,
							$\lookupMap{M}{v} = \lookupMap{M^\prime}{v} =
							\lookupMap{M_k}{v} = \lookupMap{M_{k-1}}{v} =
							\lookupMap{M_j}{v}$
							\item $\fv{r}\subseteq\map{M}$ if and only if
							$\fv{r}\subseteq\map{M_j}$.
						\end{inlineenum}
						Therefore, $\map{M}\models \refinement{r^\prime}$ if and
						only if $\map{M_j}\models\refinement{r^\prime}$.
						Therefore,
						if $\fv{r}\subseteq\map{M}$, then
						$\map{M}\models\refinement{r^\prime}$, which itself
						entails
						$\map{M}\models\refinement{r}$.
						\item[If $\processVariable{x}\neq\processVariable{y}$:]
						in that case, $\map{M} = \updateMap{M^\prime}{x}{c_x} =
						\map{M^\prime}$. The conclusion holds directly from the
						induction hypothesis.
					\end{description}
				\end{description}
				\item[Case send] In that case,
				$\map{M}\models\refinement{r^\prime}$ entails
				$\map{M}\models\refinement{r}$.
			\end{description}
		\end{description}
	\end{description}
\end{proof}

\StaticElisionCorrect*

\begin{proof}
	Let $\mathcal{R}$ being the identity of reachable states of $R$ and
	$R^\prime$. We show that $\mathcal{R}$ is a bisimulation.

	$R^\prime$ simulating $R$ is trivial, since the only difference is the lack
	of one refinement, on transition $t^\prime$ w.r.t.\ $t$.

	We now show that $R$ can simulate $R^\prime$. Since our candidate
	simulation relation $\mathcal{R}$ is the reachable state identity, we have
	to prove that for each $R^\prime$ transition
	${\sigma}\csred[t^{\prime\prime}]{\sigma^\prime}$,
	there exists a transition from $\sigma$ to $\sigma^\prime$ in $R$. We
	distinguish two cases:
	\begin{description}
		\item[If $t^{\prime\prime}$ is $t^\prime$:] In that case, we have to
		prove that $R$ can take the transition $\sigma\csred[t]\sigma^\prime$,
		i.e.\ we have to show that
		$\map{M_{\sigma^\prime}}\models\refinement{r}$.

		First, since $\sigma\csred[t^\prime]\sigma^\prime$, from the definition
		of \grrec{} or \grsnd{}:
		$\sigma^\prime = \tuple{\tuple{\dots, s_\participant{i}^\prime, \dots},
			\queue{\_}, \map{M_{\sigma^\prime}}}$.

		Since $\sigma^\prime$ is
		reachable, from
		\cref{thm:static:reachable_states_dom_refinement}, if
		$\fv{r}\subseteq\domMap{\map{M_{\sigma^\prime}}}$, then
		$\map{M_{\sigma^\prime}}\models\refinement{r}$.
		Since $t$ is well-defined, from \cref{def:static:well_defined},
		$\fv{r}\subseteq\domMap{M_{\sigma^\prime}}$.
		Therefore $\map{M_{\sigma^\prime}\models\refinement{r}}$, i.e.\ $R$ can
		take the transition $\sigma\csred[t]\sigma^\prime$.

		\item[If $t^{\prime\prime}$ is not $t^\prime$:] In that case,
		$t^{\prime\prime}\in\delta$, so $R$ can also take the transition
		$\sigma\csred[t^{\prime\prime}]\sigma^\prime$.
	\end{description}
\end{proof}

\subsection{Application to RMPST Protocols}
\label{sec:app:elision_rmpst}

\subsubsection{Definitions}

\begin{definition}[Step of a communication]
	\label{def:step_of_communication}
	A type $\gtComm[i][I]{p}{q}{\ell}[x]{S}[r]{G}$ has step $z$ (noted
	$\participant{p}\rightarrow\participant{q}\msgM[\tuple{\gtLbl{\ell},
		\processVariable{x}}]\models\refinement{r}$) if there is an $i\in I$
		such
	that $\processVariable{x} = \processVariable[i]{x}$, $\gtLbl{\ell} =
	\gtLbl{\ell_i}$ and $\refinement{r} = \refinement{r_i}$.
\end{definition}

\begin{definition}[Step in a Type]
	\label{thm:step_in_type}
	A step $z$ \emph{occurs} in a type $\gtFmt{G}$ if:
	\begin{itemize}
		\item if $\gtFmt{G}=\gtComm[i][I]{r}{s}{\ell}[x]{S}[r]{G}$,
		either:
		\begin{itemize}
			\item $z$ is a step of $\gtFmt{G}$; or
			\item there exists $i\in I$ such that $z$ occurs in $\gtFmt{G_i}$;
		\end{itemize}
		\item if $\gtFmt{G}=\gtRec{t}{G^\prime}$,
		$z$ occurs in $\gtFmt{G^\prime}$.
	\end{itemize}
\end{definition}

\begin{definition}[Happens-before in type]
	Given a global type $\gtFmt{G} =
	\gtComm[i][I]{p}{q}{\ell}[x]{S}[r]{G}$,
	$\gtFmt{G}$ \emph{happens-before} all
	$\gtFmt{G^\prime} = \gtComm[j][J]{r}{s}{\ell}[x]{S}[r]{G}\in\gtFmt{G_i}$
	where $\participant{r}=\participant{p}$ or
	$\participant{r}=\participant{q}$,
	noted $\gtFmt{G} <^\prime \gtFmt{G^\prime}$.

	The happens-before relation (noted $<$) is the transitive closure of
	$<^\prime$.
\end{definition}

\begin{remark}
	The \emph{happens-before} characterises the order of the \emph{first
		occurrence} of each step, in particular in recursive types, where a step
	can occur multiple times.
\end{remark}

\begin{example}[Happens-before in a type]
	Considering the same $\gtFmt{G_s}$ and $\gtFmt{G_y}$ than in
	\cref{ex:step_in_type}, since they both have the same sender
	$\participant{A}$, $\gtFmt{G_s}$ happens-before $\gtFmt{G_y}$.
\end{example}

\begin{definition}[Well-defined step in a type]
	\label{def:se:well_defined_step}
	Given two global types $\gtFmt{G}$ and $\gtFmt{G_s}$ such that
	$\gtFmt{G_s}\in\gtFmt{G}$, and $z =
	\participant{p}\rightarrow\participant{q}\msgM[\tuple{\gtLbl{\ell},
		\processVariable{x}}]\models\refinement{r}$ a step of $\gtFmt{G_s}$,
	we say $z$ is \emph{well-defined} if for all
	$\processVariable{x}\in\fv{r}$, there
	exists $\gtFmt{G}_{\processVariable{x}} =
	\gtComm[i\in I]{r}{s}{\_}[x_i]{\_}[\_]{G_i}$
	such that, $\gtFmt{G}_{\processVariable{x}} < \gtFmt{G_s}$
	and for one $i\in I$, $\gtFmt{G_s}\in\gtFmt{G_i}$ and
	$\processVariable{x} = \processVariable{x_i}$.
\end{definition}

\subsubsection{Lemmata}

\begin{lemma}
	\label{thm:step_implies_transition}
	Given a projectable type $\gtFmt{G}$, if a step $z =
	\participant{p}\rightarrow\participant{q}\msgM[\tuple{\gtLbl{\ell},
		\processVariable{x}}]\models\refinement{r}$ occurs in $\gtFmt{G}$,
	then there exists a transition
	$\redact{s}{s^\prime}{\actionSend{p}{q}{\gtLbl{\ell},
			\tuple{\processVariable{x}, \val{\_}}}{r}}$ in
	$\ltToAut{\proj{p}{G}}$,
	and
	$\redact{s}{s^\prime}{\actionReceive{p}{q}{\gtLbl{\ell},
			\tuple{\processVariable{x}, \val{\_}}}{\taut}}$ in
	$\ltToAut{\proj{q}{G}}$.
\end{lemma}
\begin{proof}
	We prove the case for
	$\redact{s}{s^\prime}{\actionSend{p}{q}{\gtLbl{\ell},
			\tuple{\processVariable{x}, \val{\_}}}{r}}$.
	The case for
	$\redact{s}{s^\prime}{\actionReceive{p}{q}{\gtLbl{\ell},
			\tuple{\processVariable{x}, \val{\_}}}{\taut}}$ is similar.

	By structural induction on $\gtFmt{G}$:
	\begin{itemize}
		\item if $\gtFmt{G}=\gtComm[i][I]{r}{s}{\ell}[x]{S}[r]{G}$
		\begin{itemize}
			\item $\participant{r}=\participant{p}$,
			$\participant{s}=\participant{q}$, and there exist $i\in I$ such
			that $\gtLbl{\ell}=\gtLbl{\ell_i}$,
			$\processVariable{x}=\processVariable[i]{x}$,
			and $\refinement{r}=\refinement{r_i}$:
			then, by definition of projection (\cref{def:proj}),
			$\proj{p}{G} = \ltExtC[i][I]{q}{\ell}{x}{S}{r}{G}$.
			From \cref{def:LT_to_RCFSM}, we have that
			$\redact{\proj{p}{G}}{\_}{\actionSend{p}{q}{\gtLbl{\ell},
					\tuple{\processVariable{x}, \val{\_}}}{r}}$.
			(Notice that the final state is not necessarily $\proj{p}{G_i}$, as
			there are possibly recursions.)
			\item if there exist $i\in I$ such that $z$ occurs in
			$\gtFmt{G_i}$: direct, from the induction hypothesis.
		\end{itemize}
		\item if $\gtFmt{G} = \gtRec{t}{G^\prime}$ and $s$ occurs in
		$\gtFmt{G^\prime}$: direct, from the induction hypothesis.
	\end{itemize}
\end{proof}

\begin{lemma}[Converse of \cref{thm:step_implies_transition}]
	\label{thm:transition_implies_step}
	For all types $\gtFmt{G}$ and all participants $\participant{p}$ of
	$\gtFmt{G}$, if there exists a transition
	$\redact{s}{s^\prime}{\actionSend{p}{q}{\gtLbl{\ell},
			\tuple{\processVariable{x}, \val{\_}}}{r}}$ in
	$\ltToAut{\proj{p}{G}}$, then
	a step $\participant{p}\rightarrow\participant{q}\msgM[\tuple{\gtLbl{\ell},
		\processVariable{x}}]\models\refinement{r}$ occurs in $\gtFmt{G}$.
\end{lemma}
\begin{proof}
	From \cref{def:LT_to_RCFSM}, if
	$\redact{s}{s^\prime}{\actionSend{p}{q}{\gtLbl{\ell},
			\tuple{\processVariable{x}, \val{\_}}}{r}}$ is in the set of
			transitions of
	$\ltToAut{\proj{p}{G}}$, there exist some $\ltFmt{T_1}$ and $\ltFmt{T_2}$
	in $\ltToAut{\proj{p}{G}}$ such
	that $\ltFmt{T_1}=\ltIntC[i][I]{q}{\ell}{x}{S}[r]{T}$ and $\ltFmt{T_2} =
	\ltFmt{T_i}$ for some $i\in I$.
	From \cref{def:proj}, $\ltFmt{T_1}\in\gtFmt{\proj{p}{G}}$ only if there
	exist $\gtFmt{G_1}=\gtComm[i][I]{p}{q}{\ell}[x]{S}[r]{G}$ in $\gtFmt{G}$,
	where there is one $i\in I$ such that
	$\gtLbl{\ell}=\gtLbl{\ell_i}$,
	$\processVariable{x}=\processVariable[i]{x}$, and
	$\refinement{r}=\refinement{r_i}$.
	Therefore, from \cref{thm:step_in_type},
	$\participant{p}\rightarrow\participant{q}\msgM[\tuple{\gtLbl{\ell},
		\processVariable{x}}]\models\refinement{r}$ occurs in $\gtFmt{G}$.
\end{proof}

\begin{lemma}
	\label{thm:happens_before_correct}
	Let $\gtFmt{G_1} =
	\gtComm[i][I]{p}{q}{\ell}[x]{S}[r]{G}$.
	For any global type $\gtFmt{G_2}\in\gtFmt{G_i}$,
	such that $\gtFmt{G_1}<\gtFmt{G_2}$ and such that $\participant{r}$ is the
	sender of $\gtFmt{G_2}$.
	All paths $p$ to a state
	$\tuple{\tuple{\dots, s_\participant{r}, \dots}, \queue{w}, \map{M}}$,
	where
	$s_\participant{r} = \proj{r}{G_2}$
	contain a transition $t=\sigma
	\csred[\redact{\_}{\_}{\actionSend{p}{q}{\gtLbl{\ell_i},\tuple{\processVariable{x_i},\val{\_}}}{r_i}}]
	\sigma^\prime$.
\end{lemma}

\begin{proof}
	Since $<$ is the transitive closure of $<^\prime$, we prove the result for
	$<^\prime$, the general case is then direct by induction.

	In that case, $\gtFmt{G_1} =
	\gtComm[i][I]{p}{q}{\ell}[x]{S}[r]{G}$ and $\gtFmt{G_2} =
	\gtComm[j][J]{r}{s}{\ell}[x]{S}[r]{G}$ ($\participant{r}=\participant{p}$
	or $\participant{r}=\participant{q}$) and $\gtFmt{G_2}\in\gtFmt{G_i}$ for
	some $i$.
	\begin{description}
		\item [Case $\participant{r}=\participant{p}$:]
		$\proj{p}{G_1} = \ltIntC[i][I]{q}{\ell}{x}{S}[r]{L}$, where each
		$\ltFmt{L_i}=\proj{p}{G_i}$. Since $\proj{p}{G_2}$ only appears in one
		of $\ltFmt{L_i}$, then all paths to $\proj{p}{G_2}$ in
		$\ltToAut{\proj{p}{G}}$ contain a transition to the state
		$\proj{p}{L_i}$, which is only reachable with a transition
		$\redact{\proj{p}{G_1}}{\ltFmt{L_i}}{\actionReceive{p}{q}{\gtLbl{\ell_i},\tuple{\processVariable{x_i},\val{\_}}}{\taut}}$.
		The result then follows directly from the global reduction rules.
		\item [Case $\participant{r}=\participant{q}$:]
		Therefore, $\proj{q}{G_1} = \ltExtC[i][I]{p}{\ell}{x}{S}{r}{L}$
		(with $\ltFmt{L_i} = \proj{q}{G_i}$), and $\proj{q}{G_2}\in\ltFmt{L_i}$.
		Since labels are uniquely used, $\proj{q}{G_2}$ only appears in
		$\ltFmt{L_i}$. Therefore, all paths to $\proj{q}{G_2}$ in
		$\ltToAut{\proj{q}{G}}$
		contain a transition
		$\redact{\proj{q}{G_1}}{\ltFmt{L_i}}{\actionReceive{p}{q}{\gtLbl{\ell_i},\tuple{\processVariable{x_i},\val{\_}}}{\taut}}$.
		The result then follows directly from the global reduction rules.
	\end{description}
\end{proof}

\begin{lemma}[Well-defined steps imply well-defined transitions]
	\label{thm:well_def_step_well_def_trans}
	For all projectable types $\gtFmt{G}$, for all well-defined steps
	$z=
	\participant{p}\rightarrow\participant{q}\msgM[\tuple{\gtLbl{\ell},
		\processVariable{x}}]\models\refinement{r}$ that occur in $\gtFmt{G}$,
	the transition
	$\redact{s_\participant{p}}{s_\participant{p}^\prime}{\actionSend{p}{q}{\gtLbl{\ell},
			\tuple{\processVariable{x}, \val{\_}}}{r}}$ in
			$\ltToAut{\proj{p}{G}}$ is
	well-defined in $\RcsOfType{G}$.
\end{lemma}
\begin{proof}
	Given a well-defined step
	$z=
	\participant{p}\rightarrow\participant{q}\msgM[\tuple{\gtLbl{\ell},
		\processVariable{x}}]\models\refinement{r}$, without loss of generality,
	consider $\refinement{r}$ has a single free
	variable $\processVariable{y}$.

	Since $z$ is a well-defined step occurring in $\gtFmt{G}$, let
	$\gtFmt{G_s}\in\gtFmt{G}$ the type $z$ is the step of.
	From \cref{thm:step_implies_transition}, there exist
	$t=\redact{s_\participant{p}}{s_\participant{p}^\prime}{\actionSend{p}{q}{\gtLbl{\ell},
			\tuple{\processVariable{x}, \val{\_}}}{r}}$ in
	$\ltToAut{\proj{p}{G}}$.
	From
	\cref{def:se:well_defined_step}, there exists
	$\gtFmt{G}_{\processVariable{y}} =
	\gtComm[i\in I]{r}{s}{\_}[y_i]{\_}[\_]{G_i}\in\gtFmt{G}$
	such that, $\gtFmt{G}_{\processVariable{y}}<\gtFmt{G_s}$ and,
	for one $i\in I$, $\gtFmt{G_s}\in\gtFmt{G_i}$ and
	$\processVariable{y} = \processVariable{y_i}$.

	By contradiction,
	suppose transition $t$ is not well-defined in $\RcsOfType{G}$, i.e.\ there
	is a reachable state $\sigma_\participant{p} = \tuple{\tuple{\dots,
			s_\participant{p}, \dots},
		\_, \map{M_\participant{p}}}$ such that $\processVariable{y}$ is not in
		the
	map
	$\map{M_\participant{p}}$.

	Consider a path $p$ to $s_\participant{p}$. From
	\cref{thm:happens_before_correct}, there is a transition
	$\sigma
	\csred[\redact{\_}{\_}{\actionSend{r}{s}{\gtLbl{\_},\tuple{\processVariable{y},\val{\_}}}{\_}}]
	\sigma^\prime$ with
	$\sigma^\prime = \tuple{\_, \queue{\_}, \map{M_{\sigma^\prime}}}$ in the
	path $p$.
	From the premises of rule $\grsnd{}$,
	$\processVariable{y}\in\domMap{M_{\sigma^\prime}}$. Since variables cannot
	be removed from the map, $\map{M_\participant{p}}$ contains
	$\processVariable{y}$. Contradiction.
\end{proof}

\StaticElisionRMPST*

\begin{proof}
	We prove this by showing that \cref{thm:static_elision:correct} applies to
	$\RcsOfType{G}$ and $\RcsOfType{G^\prime}$.

	First, since $\gtFmt{G^\prime}$ is $\gtFmt{G}$ where $\gtFmt{G_s}$ is
	replaced with $\gtFmt{G_{s^\prime}}$, i.e.\ the only difference is that
	$\refinement{r_i}$ is replaced with $\taut$, then the RCFSM of
	$\participant{p}$ in $\RcsOfType{G}$ contains
	$\redact{s_\participant{p}}{s_\participant{p}^\prime}
	{\actionGeneric{p}{q}{m}{r}}$,
	and $\RcsOfType{G^\prime}$ contains
	$\redact{s_\participant{p}}{s_\participant{p}^\prime}
	{\actionGeneric{p}{q}{m}{\taut}}$.
	Also, from \cref{thm:well_def_step_well_def_trans},
	$\redact{s_\participant{p}}{s_\participant{p}^\prime}
	{\actionGeneric{p}{q}{m}{r}}$ is well-defined in $\RcsOfType{G}$, and since
	$\processVariable[i]{x}\not\in\fv{r_i}$, it is also self-independent.
	This is the only change and the RCFSM of all others participants are
	unchanged.

	Second, we have to show that for each transition
	$t_w = \redact{r}{s}{\actionSend{\_}{\_}{\gtLbl{\_},
			\processVariable{x_w}}{r_w}}$
	in
	$\bigcup_{\processVariable{x}\in\fv{r}}\set{T}_{\processVariable{x}}$,
	for all map $\map{M}$,
	$\map{M}\models \refinement{r_w}$ entails $\map{M}\models\refinement{r}$.
	By contradiction, suppose there exists such a transition such that
	for all map $\map{M}$,
	$\map{M}\models \refinement{r_w}$ does not entails
	$\map{M}\models\refinement{r}$.
	From \cref{thm:transition_implies_step}, if such transition exists,
	a step
	$\participant{r}\rightarrow\participant{s}\msgM[\tuple{\gtLbl{\_},
		\processVariable[w]{x}}]\models\refinement{r_w}$ occurs in $\gtFmt{G}$.
	By hypothesis, $\map{M}\models\refinement{r_w}$ entails
	$\map{M}\models\refinement{r}$. Contradiction.

	Therefore, from \cref{thm:static_elision:correct}, there is a bisimulation
	between the states of $\RcsOfType{G}$ and those of $\RcsOfType{G^\prime}$.
\end{proof}

\section{RCS from Choreography Automata\label{sec:refined_choreo_automata}}
Despite those differences, with our framework, we can adapt CA to accommodate
refinements, which we call \emph{Refined Choreography Automata} (RCA). In this
paragraph, in order to show the versatility of our framework and to show that
the source formalism is not bound to RMPST only, we informally present how to
obtain RCS from a \emph{refined} variant of asynchronous \emph{choreography
	automata} \cite{COORDINATION20ChoreographyAutomata}, obtaining correctness
result w.r.t.\ valid refined traces (\cref{def:valid_refined_trace}).

In order to introduce refinements into CA, we can simply adapt the transition
labels in the CA, to add payloads and predicates. Adapting RCA projection would
then return RCFSMs, and we would obtain RCS with our framework, as with RMPST.
To illustrate this, we expand on (1)~p.~87 in
\cite{COORDINATION20ChoreographyAutomata}. In this example, a client
\participant{C} requests an entry from a server \participant{S}. A logger
\participant{L} is used to log the information. After receiving the result,
\participant{C} can ask \participant{S} to refine the result ($\gtLbl{ref}$), to
restart the protocol ($\gtLbl{ok}$), or to quit ($\gtLbl{bye}$). In the original
protocol, \participant{C} performs the choice.
Our goal is to let \participant{S} decide
whether \participant{C} is allowed to continue, or if \participant{C} must
terminate the protocol (in state $3$). For that, \participant{S} can embed a
boolean variable $\processVariable{q}$ in its communications, which is
then tested in the outgoing transitions of state $3$ (notice that
$\participant{C}$ can always decide to quit, independently of
$\processVariable{q}$, so we only test it for the transitions $\gtLbl{ok}$ and
$\gtLbl{ref}$). For the sake of clarity, we only show attached variables and
predicates when they are needed.

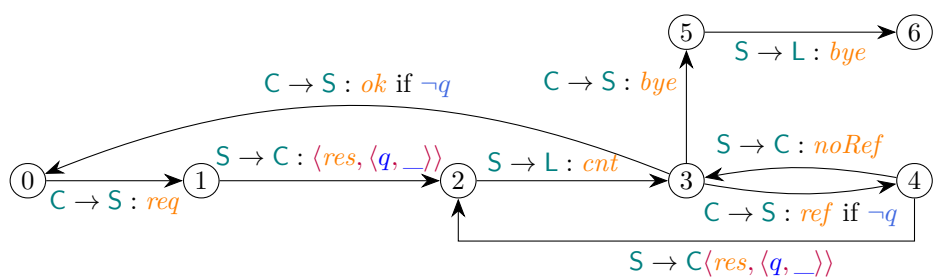
\begin{center}
	\begin{tikzpicture}[node distance=0.18\textwidth]
		\node [draw, rounded rectangle] (0) {0};
		\node [draw, rounded rectangle, right=1.8cm of 0] (1) {1};
		\node [draw, rounded rectangle, right=2.9cm of 1] (2) {2};
		\node [draw, rounded rectangle, right=of 2] (3) {3};
		\node [draw, rounded rectangle, right=of 3] (4) {4};
		\node [draw, rounded rectangle, above=1.5cm of 3] (5) {5};
		\node [draw, rounded rectangle, right=of 5] (6) {6};

		\coordinate (mid24) at ($ (2)!.6!(4) $);
		\coordinate [below=.8cm of mid24] (mid24down);

		\draw[-{Stealth[scale=1.5]}] (0) to node[anchor=north]
		{{$\participant{C}\rightarrow\participant{S}:\gtLbl{req}$}}
		(1);
		\draw[-{Stealth[scale=1.5]}] (1) to node[anchor=south]
		{{$\participant{S}\rightarrow\participant{C}:\msgM[\tuple{\gtLbl{res},
					\tuple{\processVariable{q}, \val{\_}}}]$}} (2);
		\draw[-{Stealth[scale=1.5]}] (2) to node[pos=.4, anchor=south]
		{{$\participant{S}\rightarrow\participant{L}:\gtLbl{cnt}$}} (3);
		\draw[-{Stealth[scale=1.5]}, bend right=10] (3) to node[anchor=north]
		{{$\participant{C}\rightarrow\participant{S}:\gtLbl{ref}$ if
				$\refinement{\neg q}$}} (4);
		\draw[-{Stealth[scale=1.5]}, bend right=10] (4) to node[anchor=south]
		{{$\participant{S}\rightarrow\participant{C}:\gtLbl{noRef}$}} (3);
		\draw[-{Stealth[scale=1.5]}, bend left] (4) |- (mid24down)
		node[anchor=north]
		{{$\participant{S}\rightarrow\participant{C}\msgM[\tuple{\gtLbl{res},
					\tuple{\processVariable{q}, \val{\_}}}]$}} -| (2);
		\draw[-{Stealth[scale=1.5]}, bend right=22] (3) to node[anchor=south]
		{{$\participant{C}\rightarrow\participant{S}:\gtLbl{ok}$ if
				$\refinement{\neg q}$}} (0);
		\draw[-{Stealth[scale=1.5]}] (3) to node[pos=.7, anchor=east]
		{{$\participant{C}\rightarrow\participant{S}:\gtLbl{bye}$}} (5);
		\draw[-{Stealth[scale=1.5]}] (5) to node[anchor=north]
		{{$\participant{S}\rightarrow\participant{L}:\gtLbl{bye}$}} (6);
	\end{tikzpicture}
\end{center}

By preserving predicates, the projection onto CFSM is adapted to
accommodate newly added refinements. We therefore obtain RCFSMs, which we
compose into RCS. As an illustration, the RCFSM of participant \participant{C}
would be the following:
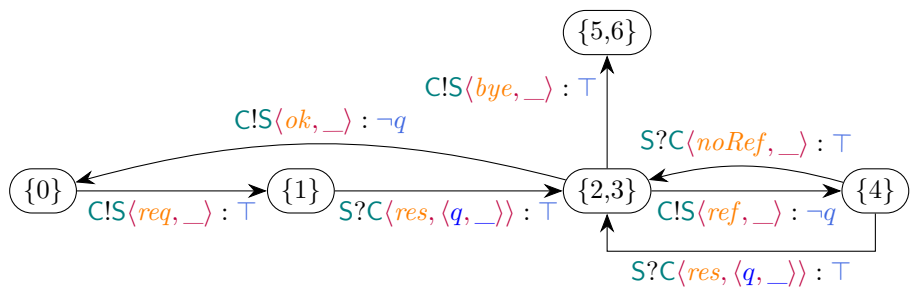
\begin{center}
	\begin{tikzpicture}[node distance=2.5cm]
		\node [draw, rounded rectangle] (0) {\{0\}};
		\node [draw, rounded rectangle, right=of 0] (1) {\{1\}};
		\node [draw, rounded rectangle, right=3cm of 1] (2) {\{2,3\}};
		\node [draw, rounded rectangle, right=of 2] (4) {\{4\}};
		\node [draw, rounded rectangle, above=1.5cm of 2] (5) {\{5,6\}};

		\coordinate (mid24) at ($ (2)!.5!(4) $);
		\coordinate [below=.8cm of mid24] (mid24down);

		\draw[-{Stealth[scale=1.5]}] (0) to node[anchor=north, sloped]
		{$\actionSend{C}{S}{\gtLbl{req}, \_}{\taut}$} (1);
		\draw[-{Stealth[scale=1.5]}] (1) to node[anchor=north, sloped]
		{$\actionReceive{S}{C}{\gtLbl{res},
				\tuple{\processVariable{q},\val{\_}}}{\taut}$} (2);
		\draw[-{Stealth[scale=1.5]}, bend right=15] (2) to node[anchor=south,
		sloped]
		{$\actionSend{C}{S}{\gtLbl{ok}, \_}{\neg q}$} (0);
		\draw[-{Stealth[scale=1.5]}] (2) to node[pos=.7, anchor=east]
		{$\actionSend{C}{S}{\gtLbl{bye}, \_}{\taut}$} (5);
		\draw[-{Stealth[scale=1.5]}] (2) to node[anchor=north, sloped]
		{$\actionSend{C}{S}{\gtLbl{ref}, \_}{\neg q}$} (4);
		\draw[-{Stealth[scale=1.5]}, bend right=15] (4) to node[anchor=south,
		sloped]
		{$\actionReceive{S}{C}{\gtLbl{noRef}, \_}{\taut}$} (2);
		\draw[-{Stealth[scale=1.5]}, bend left=40] (4) |- (mid24down)
		node[anchor=north, sloped] {$\actionReceive{S}{C}{\gtLbl{res},
				\tuple{\processVariable{q},\val{\_}}}{\taut}$} -| (2);
	\end{tikzpicture}
\end{center}

Similarly, we can obtain the RCFSM of \participant{S}
and \participant{L}. Thanks
to \cref{thm:traces_GRA_valid_refined}, we know that traces produced are valid
refined traces.

\section{Artifact}
The submitted artifact contains scripts to reproduce the benchmark results.
In order to run the benchmarks, the following is required (version numbers
indicate the version we tested our artifact on, we expect it to work on newer version, although not tested):
\begin{itemize}
	\item a linux system (tested on Debian)
	\item an internet connection
	\item an Ocaml compiler (4.11.2) release with Dune (3.6.1) and Opam (2.0.8).
	\item Regarding Opam, one needs to install:
		\begin{itemize}
			\item the ocamlgraph (2.0.0) library:\\
				\texttt{opam install ocamlgraph}
			\item a specific branch of \(\nu\)Scr:\\
				\texttt{git clone -b https://github.com/Bromind/nuscr/tree/develop-refinements}\\
				\texttt{cd nuscr}\\
				\texttt{opam pin add nuscr.dev-refinements-local .}
		\end{itemize}
	\item a Rust compiler (rustc 1.67.0-nightly (e9493d63c 2022-11-16)) with Cargo (rustc 1.67.0-nightly (e9493d63c 2022-11-16)).
	\item Python3 (3.9.2) with the libraries: argparse, json, unittest, numpy, statistics, matplotlib, and scipy.
\end{itemize}

The artifact is composed of 2 parts: one for \cref{tab:benchmark_localisation}
and one for \cref{tab:runtime_benchmarks}.
Each part is covered by a shell script, which manages the \emph{download}, the
\emph{preparation} and the \emph{cleanup} of the system.

\subsection{Refined Rumpsteak v.\ Vanilla Rumpsteak}
The script \texttt{compare.sh} performs the comparision between Rumpsteak with refinements and without refinements.
The scripts prints detailled results on \texttt{stdout}.
Each microbenchmark begins with:
\begin{center}
	\texttt{********* Analysis NAME\_OF\_BENCHMARK **********}
\end{center}
In \cref{tab:runtime_benchmarks}, we report the \(p\)-value printed on the
line \texttt{MannwhitneyuResult(statistic=..., pvalue=...)}
In addition, the scripts produces a folder in
\texttt{/tmp/Rumpsteak\_benchmarks.XXX} (where \texttt{XXX} is a random sequence
of three characters) which contains records of each run of each microbenchmark.
The runtimes reported in \cref{tab:runtime_benchmarks} are the median runtime
for both the \texttt{vanilla} and \texttt{refinements} measurements, which are
shown, for each benchmark after the lines:
\begin{center}
	***** 1st quartile, median, 3rd quartile (vanilla) *******
\end{center}
and
\begin{center}
	***** 1st quartile, median, 3rd quartile (refinements) *******
\end{center}

\subsection{Evaluation of the localisation algorithm}
The script \texttt{dynamic\_verify.sh} measures the runtime of the localisation tools.
After setting up the tools, the scripts runs the benchmarks and prints results
on \texttt{stdout}. Each microbenchmark begins with:
\begin{center}
	\texttt{********** NAME\_OF\_BENCHMARK ***********}
\end{center}

For each benchmark, we output a result in the form:
\begin{center}
	\texttt{Time (mean ± \(\sigma\)):      22.4 ms ±   0.8 ms    [User: 10.5 ms, System: 11.4 ms]}\\
\end{center}
Results reported in \cref{tab:benchmark_localisation} are the printed mean and
standard deviation as output by the script.

}{}

\end{document}